\documentclass[final]{amsart}

\RequirePackage{amsthm,amsmath}
\RequirePackage[numbers]{natbib}
\RequirePackage[colorlinks,citecolor=blue,urlcolor=blue]{hyperref}
\RequirePackage{graphicx,pdflscape,multirow}

\numberwithin{equation}{section}
\theoremstyle{plain}
\newtheorem{thm}{Theorem}[section]

\newtheorem{lem}[thm]{Lemma}
\newtheorem{prop}[thm]{Proposition}
\newtheorem{cor}[thm]{Corollary}

\newtheorem{rem}{Remark}[section]

\newtheorem{hyp}{Assumption}[section]
\newtheorem{ex}{Example}[section]
\numberwithin{equation}{section}

\usepackage{amsfonts,amsmath,amssymb,mathtools}

\usepackage{ulem}

\usepackage[color,notref,notcite]{showkeys}
\usepackage{enumerate}
\definecolor{labelkey}{rgb}{0.6,0,1}

\def\disp{\displaystyle}

\newcommand{\dB}{\ensuremath{\mathbb{B}}}

\newcommand{\dF}{\ensuremath{\mathbb{F}}}

\newcommand{\I}{\ensuremath{\mathbb{I}}}

\newcommand{\dR}{\ensuremath{\mathbb{R}}}

\newcommand{\dT}{\ensuremath{\mathbb{T}}}

\newcommand{\dZ}{\ensuremath{\mathbb{Z}}}

\newcommand{\cA}{\ensuremath{\mathcal{A}}}

\newcommand{\cC}{\ensuremath{\mathcal{C}}}

\newcommand{\cF}{\ensuremath{\mathcal{F}}}
\newcommand{\cG}{\ensuremath{\mathcal{G}}}

\newcommand{\cN}{\ensuremath{\mathcal{N}}}
\newcommand{\cO}{\ensuremath{\mathcal{O}}}
\newcommand{\cP}{\ensuremath{\mathcal{P}}}
\newcommand{\cQ}{\ensuremath{\mathcal{Q}}}
\newcommand{\cR}{\ensuremath{\mathcal{R}}}

\def\i1{ [-\infty,\infty]}
\def\ve{\varepsilon}

\def\var{ {\rm var~}}

\def\sn{\sqrt{n}\; }
\def\sN{\sqrt{N}\; }

\def\unif{{\rm U}(0,1)}

\def\ba{\mathbf{a}}
\def\bb{\mathbf{b}}

\def\be{\mathbf{e}}
\def\bh{\mathbf{h}}

\def\bs{\mathbf{s}}

\def\bu{\mathbf{u}}
\def\bv{\mathbf{v}}
\def\bx{\mathbf{x}}
\def\by{\mathbf{y}}
\def\bz{\mathbf{z}}

\def\bG{\mathbf{G}}

\def\bU{\mathbf{U}}
\def\bV{\mathbf{V}}
\def\bX{\mathbf{X}}
\def\bY{\mathbf{Y}}
\def\bZ{\mathbf{Z}}

\def\balpha{{\boldsymbol\alpha}}
\def\bbbeta{{\boldsymbol\beta}}

\def\bbeta{{\boldsymbol\eta}}
\def\bgamma{{\boldsymbol\gamma}}
\def\bGamma{{\boldsymbol\Gamma}}

\def\bmu{{\boldsymbol\mu}}
\def\bphi{{\boldsymbol\phi}}

\def\btheta{{\boldsymbol\theta}}
\def\bzeta{{\boldsymbol\zeta}}
\def\bxi{{\boldsymbol\xi}}

\def\bTheta{{\boldsymbol\Theta}}

\def\bvarphi{{\boldsymbol\varphi}}

\def\setK{\{1,\ldots,K\}}

\newcommand{\BB}[1]{\textcolor{black}{#1}}

\begin{document}

\title{On factor copula-based mixed regression models}

\title{On factor copula-based mixed regression models}
\author{Pavel Krupskii}
\address{School of Mathematics and Statistics, University of Melbourne, Australia.}
\email{pavel.krupskiy@unimelb.edu.au}
\author{Bouchra R. Nasri}%
\address{Département de médecine sociale et préventive, École de santé publique, Université de Montréal, C.P. 6128, succursale Centre-ville
Montréal (Québec)  H3C 3J7 }
\email{bouchra.nasri@umontreal.ca}
\author{Bruno N. R\'emillard}
\address{Department of Decision Sciences,
HEC Montr\'eal, 3000 chemin de la C\^ote Sainte-Catherine,
Montr\'eal (Qu\'ebec), Canada H3T 2A7}\email{bruno.remillard@hec.ca}

\begin{abstract}
 In this article, a copula-based method for mixed regression models is proposed, where the conditional distribution of the response variable, given covariates, is modelled by a parametric family of continuous or discrete distributions, and a latent variable models the dependence between observations in each cluster. We demonstrate the estimation of copula and margin parameters, outlining the procedure for determining the asymptotic behaviour of the estimation errors. Numerical experiments are performed to assess the precision of the estimators for finite samples. An example of its application is given using COVID-19 vaccination hesitancy data from several countries. All developed methodologies are implemented in \href{https://cran.r-project.org/web/packages/CopulaGAMM/index.html}{CopulaGAMM},  available in CRAN.
\end{abstract}

\thanks{$^*$ The three authors contributed equally to all aspects of this paper. Partial funding in support of this work was provided by Mathematics for Public Health, the Natural Sciences and Engineering Research Council of Canada, the Fonds de recherche du Qu\'ebec -- Nature et technologies, and the Fonds de recherche du Qu\'ebec -- Sant\'e.}

\maketitle

\section{Introduction}\label{sec:intro}

Clustered data naturally arises in various fields, such as actuarial science, hydrology, clinical medicine, and public health \citep{Goldstein:2011, Dias/Garcia/Schmidt:2013, Schmidt/deMoraes/Migon:2017, Velozo/Alves/Schmidt:2014,Leyland/Groenewegen:2020} where the observations, represented by a response variable and covariates, exhibit dependence inherited from a hierarchical structure. Statistical analyses are carried out with respect to a hierarchical structure of the data while considering the characteristics of the sampled units or individuals. For example, in public health, when analyzing the number of deaths due to cardiovascular diseases during heat waves, it becomes crucial to scrutinize the distribution of these deaths, taking into account factors such as geographical areas (e.g., province, region, city, neighborhood), medical resources (e.g., hospitals, accommodation centers), as well as epidemiological, climatic,  and sociological factors. The shared characteristics among data points (same hospital, same neighborhood, same city) inherently makes the data dependent (clustered) due to common underlying characteristics (factors); see Figure \ref{fig:figure1}.

A commonly used approach for analyzing clustered data is through a mixed linear regression, i.e., a linear regression with random coefficients \citep{Goldstein:2011}. In the simplest of models,  the regression equation has a random intercept common to all the observations in the same cluster, creating dependence among observations within each cluster. Suppose the clustered data comes from a single variable of interest $Y$ and its associated covariates $ \bX $. A simple 2-level mixed regression model can be expressed as
\begin{equation}\label{eq:mixed-reg1}
Y_{ki} = \bb^\top \bX_{ki}+\eta_k +\ve_{ki}, \qquad  i\in\{1,\ldots,n_k\}, \quad k \in \setK,
\end{equation}
where $Y_{ki}$ is the $i$-th observation of the  variable of interest from cluster $ k $, $ \bX_{k i} $ is the corresponding vector of covariates,
$ \eta_k $ is a latent (non-observable) random factor common to the cluster $ k $, and $ \ve_{k i} $ is an error term. Additionally, we assume that the error terms $ \ve_{k i} $ and the latent variables $ \eta_k $, $ i \in \{1, \ldots, n_k \} $, $ k \in \setK $, are independent. This implies that the correlation between two observations within the same cluster is $\frac{\sigma_\eta^2}{\sigma_\eta^2+\sigma_\epsilon^2}$, where $\sigma_\eta^2$ is the variance of $\eta_k$ and $\sigma_\epsilon^2$  is the variance of $\epsilon_{ki}$.
In cases where the latent variables are not independent, more complex models such as 3-level models may be considered. See, e.g., Figure \ref{fig:figure1} for an illustration of a 3-level model, where the latent variables in clusters 11 and 21 are dependent.

\begin{figure}
    \centering
    \includegraphics[height=1in,width=4.5in]{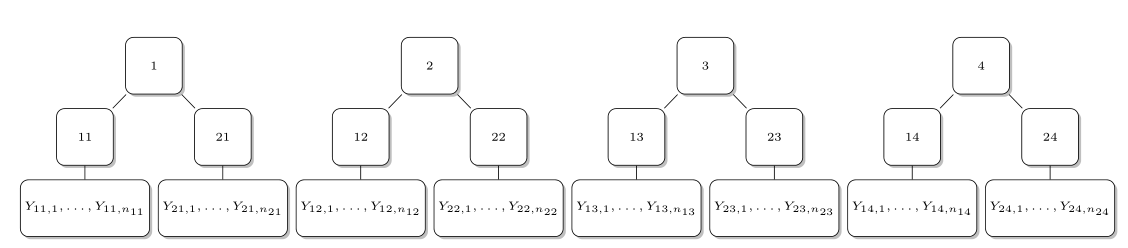}
    \caption{An example of clustered data: 3-level model}
    \label{fig:figure1}
\end{figure}

Other generalized models can also be used, such as generalized linear mixed models (GLMM) and generalized additive mixed models (GAMM). On the one hand, when the intercept is random, the model is overly restrictive, and the difference in estimations between the clusters is a random translation. On the other hand, when the intercept and the slope are random,  the model requires additional parameters to be estimated. A better trade-off that ensures at least as much complexity with fewer parameters' estimation can be achieved by copula models. It has been shown that copula-based approaches offer advantages and provide flexible and good alternatives to classical regression approaches for monotonic and non-monotonic dependence structures \citep{Song/Li/Yuan:2009, Noh/ElGhouch/Bouezmarni:2013, Remillard/Nasri/Bouezmarni:2017, Nasri/Remillard/Bouezmarni:2019, Nasri:2020}. Despite the relevance of these approaches, there has been very little work combining copulas and clustered data \citep{Shi/Feng/Boucher:2016, Zhuang/Diao/Yi:2020,Pereda-Fernandez:2021,Nikoloulopoulos:2017,Nikoloulopoulos:2022}. In \citet{Shi/Feng/Boucher:2016}, the authors used Gaussian copulas to model the dependence between all variables of interest for a given cluster; in their setting, the covariates are binary, and the copula is independent of the covariates. In contrast, in \citet{Zhuang/Diao/Yi:2020}, the authors proposed a Bayesian clustering approach for multivariate data without covariates, where the dependence between variables is modelled for each cluster by a copula with a random multivariate parameter. In \citet{Pereda-Fernandez:2021}, the author proposed a random intercept linear model for a binary variable of interest, where Archimedean copulas are used to capture the dependence between the random intercepts. Finally, in \citet{Nikoloulopoulos:2017, Nikoloulopoulos:2022}, the author considered a specific case of clustered data with a mixture of binomial distributions that can be applied for joint meta-analysis designs for multiple tests.

While each of these approaches is applied to a specific data structure, such as a particualar type of variables, dependence model, or data design, none of them comprehensively covers or generalizes GAMMs. The primary aim of this article is to present a general copula-based regression approach for clustered data. This approach encompasses GAMMs as specific cases  and  can account for any dependence structure. In addition, the variable of interest may be either continuous or discrete. Our proposed approach is described in Section \ref{sec:model}, along with the main assumptions concerning the joint distribution of the observations. In Section \ref{sec:est}, the estimation method is presented, including the asymptotic behaviour of the estimators, while the precision of the proposed method is assessed in Section \ref{ssec:simul} through finite sample simulations, together with comparisons with GAMMs in Section \ref{ssec:comp}.
Finally, in Section \ref{sec:ex}, we apply the proposed methodology to explain factors contributing to COVID-19 vaccination hesitancy across multiple countries.

\section{Description of the model and its properties}\label{sec:model}

In the subsequent sections, for each cluster $k\in \setK$, the observations are denoted by $(Y_{ki},\bX_{ki})\in \dR^{1+d}$, $ i \in \{1, \ldots, n_k \} $.
The model is based on the following set of assumptions:\\

\begin{hyp}\label{hyp:margin}
For a fixed $k$, the observations  $(Y_{ki},\bX_{ki})\in \dR^{1+d}$, $ i \in \{1, \ldots, n_k \}$,
are exchangeable and the distribution of $ (Y_{k, 1}, \bX_{k, 1})$ is consistent across all $k$. Furthermore, for any $k\in\setK$ and $ i \in \{1, \ldots, n_k \} $, the conditional distribution of $ Y_{k i}$ given $\bX_{k i}=\bx $ is
\begin{equation}\label{eq:hyp-cdf-uncond}
P(Y_{ki}\le y|\bX_{ki}=\bx) = G_{\bh(\balpha, \bx)  }(y),
\end{equation}
where $G_\ba, \ba\in \cA\subset \dR^m$ is a parametric family of univariate distributions, 
$\bh: \dR^m \times \dR^d \mapsto \cA$ is a known link function, $\balpha$ is an unknown parameter, and each component \BB{of the $m$-dimensional vector $\bh(\balpha,\bx)$} can be a linear or a non-linear function of $\bx$.
\end{hyp}

\begin{hyp}\label{hyp:copula}
For the copula-based models, we assume that the latent variables $V_1,\ldots, V_K$  are independent and uniformly distributed over $(0,1)$, and are  independent of the covariates  $\bX = (\bX_{11}=\bx_{11},\ldots, \bX_{K n_K}=\bx_{K n_K})$. It is important to note that this assumption is similar to one of GAMM. Furthermore,  assume that  conditionally on $(V_1,\ldots, V_K)$ and  $\bX$, the responses $ Y_{k i}$, $ i \in \{1, \ldots, n_k \} $ are independent. Finally, using \BB{a conditional version of Sklar's theorem due to \cite{Patton:2004}, we assume that
$$
P(Y_{ki}\le y, V_k\le v|\bX_{ki}=\bx) = C_{\bphi(\btheta,\bx)}\left\{G_{\bh(\balpha,\bx}(y),v)\right\},\quad (y,v)\in\dR\times (0,1),
$$
}
so one gets
\begin{eqnarray}
P(Y_{ki}\le y|\bX_{ki}=\bx,V_k=v)& = &\partial_v C_{\bphi(\btheta,  \bx)}\left\{G_{\bh(\balpha, \bx)}(y), v\right\}\nonumber\\
&=&
\cC_{\bphi(\btheta,  \bx)}\left\{G_{\bh(\balpha, \bx)}(y), v\right\}, \label{eq:hyp-cdf}
\end{eqnarray}
where $C_\bvarphi, \bvarphi \in \cO\subset \dR^p$ is a parametric family of bivariate copulas with density $c_\bvarphi$,
$\cC_\bvarphi(u,v)= \partial_v C_\bvarphi(u,v)$, $\bvarphi :\dR^p \times \dR^d \mapsto \cO$ is a known link function, each component of \BB{the $p$-dimensional vector $\phi(\btheta,\bx)$} can be a linear or non-linear function of $\bx$, and $\btheta = (\balpha,\bbbeta)$  is the vector of unknown parameters. \BB{Note that} the margin's parameter $\balpha$ can also appear in the copula \BB{parameter, if} necessary, as exemplified by the case of a mixed linear Gaussian regression with a random slope discussed in Appendix \ref{app:gaussian}.
\end{hyp}

Under Assumptions \ref{hyp:margin}--\ref{hyp:copula},
 the conditional distribution function of $\bY = (Y_{11},\ldots, Y_{K n_k})$ given $\bX=\bx$ is
 \begin{equation}\label{eq:cdftotale}
\prod_{k=1}^K  \left[ \int_0^1 \left\{ \prod_{i=1}^{n_k} \cC_{\bphi(\btheta, \bx_{ki}) }\{ G_{\bh(\balpha, \bx_{ki})} (y_{ki}),v_k\} \right\} dv_k \right] ,
 \end{equation}
 which is an extension of the so-called 1-factor copula model introduced by \citep{Krupskii/Joe:2013}.
As a result, the log-likelihood $L(\btheta)$
 can be written  as
$\disp
L(\btheta) = \sum_{k=1}^K \log f_{\btheta,k}(\by_k)$, where
\begin{equation}\label{eq:LL-unified}
f_{\btheta,k}(\by_k) = \int_0^1 \left\{\prod_{i=1}^{n_k} f_{\btheta,ki}(y_{ki},v)\right\}dv,
\end{equation}
and for $v\in (0,1)$,
$f_{\btheta,ki}(y_{ki},v)$ is a density with respect to a reference measure $\nu$ (Lebesgue's measure or counting measure on $\dZ$), i.e., $\disp \int_{\dR} f_{\btheta,ki}(y,v)\nu(dy) =1 $. In particular, in the continuous case,
$$
f_{\btheta,ki}(y,v) = g_{\bh(\balpha, \bx_{ki})}(y)
c_{\bphi(\btheta, \bx_{ki})}\left\{ G_{\bh(\balpha,\bx_{ki})}(y),v\right\},
$$
where $g_\ba$ is the density of $G_\ba$, while in the discrete case,
$$
f_{\btheta,ki}(y,v) = \cC_{\bphi(\btheta, \bx_{ki})}\left\{ G_{\bh(\balpha,\bx_{ki})}(y),v\right\}-\cC_{\bphi(\btheta, \bx_{ki})}\left\{ G_{\bh(\balpha,\bx_{ki})}(y-),v\right\}.
$$
As a by-product of \eqref{eq:LL-unified},  the conditional density of $V_k$, given the observations in cluster $k$, is
\begin{equation}\label{eq:Vfivenobs}
\prod_{i=1}^{n_k} f_{\btheta,ki}(y_{ki},v) \Big{/}f_{\btheta,k}(\by_k),
\qquad v \in (0,1).
\end{equation}
Hence, \BB{replacing $\btheta$ by its estimation in} \eqref{eq:Vfivenobs}, \BB{the latter can be used to estimate} $V_k$, \BB{e.g., using the mean or the median of \eqref{eq:Vfivenobs}}. We note that in the continuous case, \eqref{eq:Vfivenobs} simplifies to
\begin{equation}\label{eq:Vfivenobs-cont}
\prod_{i=1}^{n_k} c_{\bphi(\btheta, \bx_{ki})}\left\{ G_{\bh(\balpha,\bx_{ki})}(y),v\right\}  \Big{/}\int_0^1 \left[\prod_{i=1}^{n_k}  c_{\bphi(\btheta, \bx_{ki})}\left\{ G_{\bh(\balpha,\bx_{ki})}(y),s\right\}\right]ds,
\end{equation}
$v \in (0,1)$.
Subsequently, for observations  in cluster $k$, the conditional distribution, the conditional quantile, and the conditional expectation of $Y$ given $\bX=\bx, V_k=v$, are expressed respectively for $y\in\dR$ and $v\in (0,1)$ by
\begin{equation}\label{eq:conddist}
P(Y\le y|\bX=\bx,V_k=v) =  \cF(y,\bx,v) = \cC_{\bphi(\btheta,  \bx)}\left\{G_{\bh(\balpha, \bx)}(y),v\right\},
\end{equation}
\begin{equation}\label{eq:condquantile}
\cQ(u,\bx,v) =
G_{\bh(\balpha, \bx)}^{-1} \left\{ \cC_{\bphi(\btheta,  \bx)}^{-1}\left(u, v\right)\right\},\quad u\in (0,1),
\end{equation}
\begin{eqnarray}
E(Y|\bX=\bx,V_k=v) &=&\int_0^1 \cQ(u,\bx,v) du  \nonumber\\
&=& \int_0^\infty \left\{1- \cF(y,\bx,v) \right\}dy
-\int_{-\infty}^0 \cF(y,\bx,v) dy , \label{eq:condexp}
\end{eqnarray}
where the inverse functions are the usual left-inverse. \BB{More precisely, for any $u,v\in(0,1)$,
$$
G_{\bh(\balpha, \bx)}^{-1}(v) =\inf\left\{y \in \dR:\;  G_{\bh(\balpha, \bx)}(y)\ge v\right\},
$$
$$
\cC_{\bphi(\btheta,\bx)}^{-1} (u,v)=\inf\left\{w\in [0,1]: \; C_{\bphi(\btheta,\bx)} (u,w)\ge v\right\}.
$$}
In particular, if $G_\ba$ is a location-scale distribution, i.e., $G_{\bh(\balpha, \bx)}(y) = G_0\left(\frac{y-\mu_{\balpha,\bx}}{\sigma_{\balpha,\bx}}\right)$, then
\begin{equation}\label{eq:condexpnew}
E(Y|\bX=\bx,V_k=v)=\mu_{\balpha,\bx} + \sigma_{\balpha,\bx} \kappa_{\btheta,\bx,v},
\end{equation}
where
$\disp \kappa_{\btheta,\bx,v} = \int_0^\infty \left[ 1- \cC_{\bphi(\btheta,  \bx)} \left\{G_0(z),v\right\} \right] dz  - \int_{-\infty}^0 \cC_{\bphi(\btheta,  \bx)} \left\{G_0(z),v\right\} dz$. \\
The proposition below demonstrates that GAMMs are specific instances within our models.
\begin{prop}\label{prop:gamm}
Suppose the latent variable $\eta$ has a density $f_\bb$ and cdf $F_\bb$,
\begin{itemize}
    \item[(a)]
Consider the GAMM with continuous margin and latent variable $\eta$ defined by
$$
P(Y\le y|\bX=\bx, \eta=z) = \tilde G_{\tilde h\{\balpha^\top \bs(\bx)+z\},\sigma}(y),
$$
where $\tilde h(a)$ is the mean of $\tilde G_{\tilde h(a),\sigma}$ and $\sigma$ is a scaling parameter. Consequently, the associated copula-based model is given by the margin  $G_{\balpha^\top \bs(\bx),\sigma,\bb}$ and copula $C_{\balpha^\top \bs(\bx),\sigma,\bb}$, where
$$
G_{a,\sigma,\bb}(y) = \int_{-\infty}^\infty f_\bb(z) \tilde G_{\tilde h(a+z),\sigma}(y)dz
$$
and $C_{a,\sigma,\bb}(u,v)=\int_0^v \cC_{a,\sigma,\bb}(u,s)ds$,
with
$
\cC_{a,\sigma,\bb}(u,v) = \tilde G_{\tilde h\{a+F_\bb^{-1}(v)\},\sigma}\circ G_{a, \sigma,\bb}^{-1}(u)$.
\item[(b)]
Consider a GAMM with discrete margin and latent variable $\eta$ defined by
$$
P(Y=k|\bX=\bx, \eta=z) = \tilde g_{\tilde h\{\balpha^\top\bs(\bx)+z\},\sigma}(k),
$$
where $\tilde h(a)$ is the mean of $\tilde g_{\tilde h(a),\sigma}$ and $\sigma$ is a scaling parameter. Consequently, we can find an associated copula-based model with margin  $G_{\balpha^\top \bs(\bx),\sigma,\bb}$ and copula $C_{\balpha^\top \bs(\bx),\sigma,\bb}$, where
$$
G_{a,\sigma,\bb}(k) = P(Y\le k|\bX=\bx) = \sum_{j=0}^k \int_{-\infty}^\infty  \tilde g_{\tilde h\{a+z\},\sigma}(j) f_\bb(z)dz,
$$
and $C_{a,\sigma,\bb}(u,v)=\int_0^v \cC_{a,\sigma,\bb}(u,s)ds$,
with
\begin{equation}\label{eq:copdisgamm}
\cC_{a,\sigma,\bb}\{G_{a,\sigma,\bb}(k),v\} = \sum_{j=0}^k  \tilde g_{\tilde h\{\balpha^\top\bs(\bx)+F_\bb^{-1}(v)\},\sigma}(j).
\end{equation}
\end{itemize}
\end{prop}
\begin{proof}
Without loss of generality, let us assume that the scaling parameter is implicitly defined in the margins.
For the proof of (a), from the very definitions of $G$ and $\cC$, we obtain that
$\cC_{\balpha^\top \bs(\bx),\bb}\left\{G_{\balpha^\top \bs(\bx),\bb}(y),v\right\} = \tilde G_{\tilde h\{\balpha^\top \bs(\bx)+F_\bb^{-1}(v)\}}(y)$. To complete the proof, we note that $\eta$ has the same distribution as $F_\bb^{-1}(V)$. For the proof of (b), we obtain
$
\cC_{a, \bb}\{G_{\balpha^\top \bs(\bx),\bb}(k),v\} -\cC_{a, \bb}\{G_{\balpha^\top \bs(\bx),\bb}(k-1),v\}  =
\tilde g_{\tilde h\{\balpha^\top\bs(\bx)+F_\bb^{-1}(v)\}}(k)$.
\end{proof}
Appendix \ref{difference} provides graphical examples to illustrate the difference between GAMM and our proposed approach.

\section{Estimation of parameters}\label{sec:est}

Let $N_K = \disp \sum_{k=1}^K n_k$ be the total number of observations. For the estimation of $\btheta = (\alpha,\bbbeta)$,
we consider the fully parametric maximum likelihood estimator, i.e.,
$$
\btheta_K=(\balpha_K,\bbbeta_K) = \disp \arg\max_{(\balpha,\bbbeta)\in  \cP_1\times \cP_2}L(\balpha,\bbbeta).
$$

In addition to Assumptions \ref{hyp:margin}--\ref{hyp:copula}, we need the following hypotheses:
\begin{hyp}\label{hyp:density}
The link function $\bh(\balpha,\bx)$ is twice continuously differentiable with respect to $\balpha$, and the density $g_\ba$ (in the continuous case) or the cdf $G_\ba$ (in the discrete case) is twice continuously differentiable with respect to $\ba$. The first and second order derivatives are denoted respectively by $\dot g_\ba$ and $\ddot g_\ba$, or $\dot G_\ba$ and $\ddot G_\ba$. Furthermore, we assume that $\bphi(\btheta,\bx)$ is twice continuously differentiable with respect to $\btheta = (\balpha,\bbbeta)$ and $c_\bvarphi$ is twice continuously differentiable with respect to $\bvarphi$, with first and second order derivatives denoted respectively by $\dot c_\bvarphi$ and $\ddot c_\bvarphi$.
Using \eqref{eq:LL-unified}, we obtain
$\disp \dot L(\btheta) = \nabla_{\btheta}L(\btheta) = \sum_{k=1}^K \frac{\dot f_{\btheta,k}(\by_{k})}{f_{\btheta,k}(\by_k)}$,
where
\begin{equation}\label{eq:gradient-unified}
    \dot f_{\btheta,k}(\by_k) = \sum_{i=1}^{n_k} \int \left\{\prod_{j\neq i} f_{\btheta,kj}(y_{kj},v)\right\} \dot f_{\btheta,ki}(y_{ki},v) \; dv.
\end{equation}
\end{hyp}

\begin{hyp}\label{hyp:mu}
Set $\disp \bbeta_{\btheta,ki}= \frac{\int_0^1 \left\{\prod_{j\neq i} f_{\btheta,kj}(Y_{kj},v)\right\}\dot f_{\btheta,ki}(Y_{ki},v)\; dv}{f_{\btheta,k}(\bY_k)}$. In some bounded neighborhood $\cN$ of $\btheta_0$, $\disp \bmu(\btheta) = E_{\btheta_0} \left[ \bbeta_{\btheta,ki}\right] =0$ iff $\btheta=\btheta_0$, its Jacobian $\dot \bmu(\btheta)$ exists, is continuous, and is positive definite or negative definite at $\btheta_0$.
Additionally,
$\disp \limsup_{K\to\infty} \max_{1\le k\le K} \max_{1\le i\le n_k}  E_{\btheta_0}\left(\sup_{\btheta\in \cN}\|\bbeta_{\btheta,ki}\|^2\right) <\infty$,
 for some neighborhood $\cN$ of $\btheta_0$.
\end{hyp}

The following condition might also be needed.
\begin{hyp}\label{hyp:l4bound}
$\disp
\limsup_{K\to\infty} \max_{1\le k\le K} \max_{1\le i\le n_k} E_{\btheta_0}\left(
\sup_{\btheta\in \cN}\|\bbeta_{\btheta,ki}\|^4\right) <\infty$,  for some neighborhood $\cN$ of $\btheta_0$.
\end{hyp}


The proof of the following convergence result is presented in Appendix \ref{app:pf-main}. It relies on auxiliary results outlined in Appendix \ref{app:aux}.
To simplify notations, set $\bbeta_{ki} =\bbeta_{\btheta_0,ki} $,
where $\btheta_0$ is the true parameter.
Recall that $N_K = \disp \sum_{k=1}^K n_k$.

\begin{rem}\label{rem:Ksmall} For clustered data, it is essential that the number of clusters tends to infinity. In fact, in the simple Gaussian case $y_{ki}=\eta_k+\mu+\epsilon_{ki}$, with $n_k\equiv n$, the estimation of $\mu$ is $\bar{\bar y}= \mu+\bar \eta +\bar{\bar \epsilon}$ which does not converge to $\mu$ if $K$ is fixed. Its variance is \BB{$\frac{1}{K}\left(\frac{\sigma_\epsilon^2}{n}+ \sigma_\eta^2\right)$}, which converges to $\frac{\sigma_\eta^2}{K}$, as $n\to\infty$.
\end{rem}

\begin{thm}\label{thm:main}
We assume that as $K\to\infty$,
$\disp \lambda_K = \frac{1}{N_K}\sum_{k=1}^{K} n_k^2 \to\lambda \in [1,\infty)$ and $\disp \frac{1}{N_K^2}\sum_{k=1}^K n_k^4\to 0$. We further assume that Assumptions \ref{hyp:density}--\ref{hyp:l4bound} hold. Then $\bTheta_K = N_K^{1/2}(\btheta_K-\btheta_0)$ converges in distribution to a centered Gaussian random vector with covariance matrix $\Sigma^{-1}$, where $\Sigma = \Sigma_{11}+(\lambda-1)\Sigma_{12}$, with $\Sigma_{11} = E\left(\bbeta_{k1}\bbeta_{k1}^\top \right)$ and
$\Sigma_{12} = E\left(\bbeta_{k1}\bbeta_{k2}^\top \right)$. Moreover, $\Sigma$ can be estimated by
$\disp
    \hat\Sigma = \frac{1}{N_K}\sum_{k=1}^K \frac{\dot f_{\btheta_K,k} \dot f_{\btheta_K,k}^\top }{f_{\btheta_K,k}^2}$.
\BB{When $n_k\equiv n$, convergence holds under Assumptions \ref{hyp:density}--\ref{hyp:mu}. }
\end{thm}
\begin{rem}
Note that, because of Assumption \ref{hyp:mu}, we can differentiate under the integral sign in \eqref{eq:gradient-unified}. A Newton-Raphson-type algorithm can be used to obtain parameter estimates together with Gaussian quadrature method to integrate over the latent factors. The main computational difficulty may arise when $n_k$ is large, implying that the function to integrate in \eqref{eq:gradient-unified} can be either very small or very large. To adrress this, we can multiply the $f_{\btheta,kj}(y_{kj},v)$s by a specific constant to avoid the underflow problem.
 Additional details about the computation of the log-likelihood function are provided in the supplementary material. These computations are implemented in the R package CopulaGAMM \citep{Krupskii/Nasri/Remillard:2023b} available on CRAN.
\end{rem}

Finally, it is important to note that the likelihood-based criteria such as AIC or BIC could be used to select the best margins and the copula family that best fit the data. In Section \ref{ssec:comp}, we use these two criteria for model selection.

\section{Numerical experiments}

We now propose numerical experiments to evaluate the performance of both the proposed estimators and the copula-based models in comparison with GAMMs.

\subsection{Performance of the estimators}\label{ssec:simul}

In this section, we assess the performance of the maximum likelihood estimator proposed in Section \ref{sec:est} in eight numerical experiments. It is worth noting that the running time for each sample is 1-2 seconds, demonstrating a relatively fast execution. For each of the eight models we considered, we generated 1000 samples from Equation \eqref{eq:cdftotale} and considered two scenarios. In the first scenario, $n_k\equiv 5$ and $K\in\{5, 20, 100, 500\}$, while in the second scenario, $K=5$, and $n_k\equiv n \in \{5, 20, 100, 500\}$.
For the first numerical experiment, we considered Clayton copula with parameter $2e^{\beta_1}= 2$, so that the Kendall's tau equals $0.5$, and the margins $G_\ba$ are normal, with mean and variance parameters $\ba = (10, 1)$.
Figure \ref{fig:exp1} displays the boxplots of the estimated parameters for the two scenarios.
\begin{figure}[ht!]
    \centering
    \includegraphics[scale=0.28]{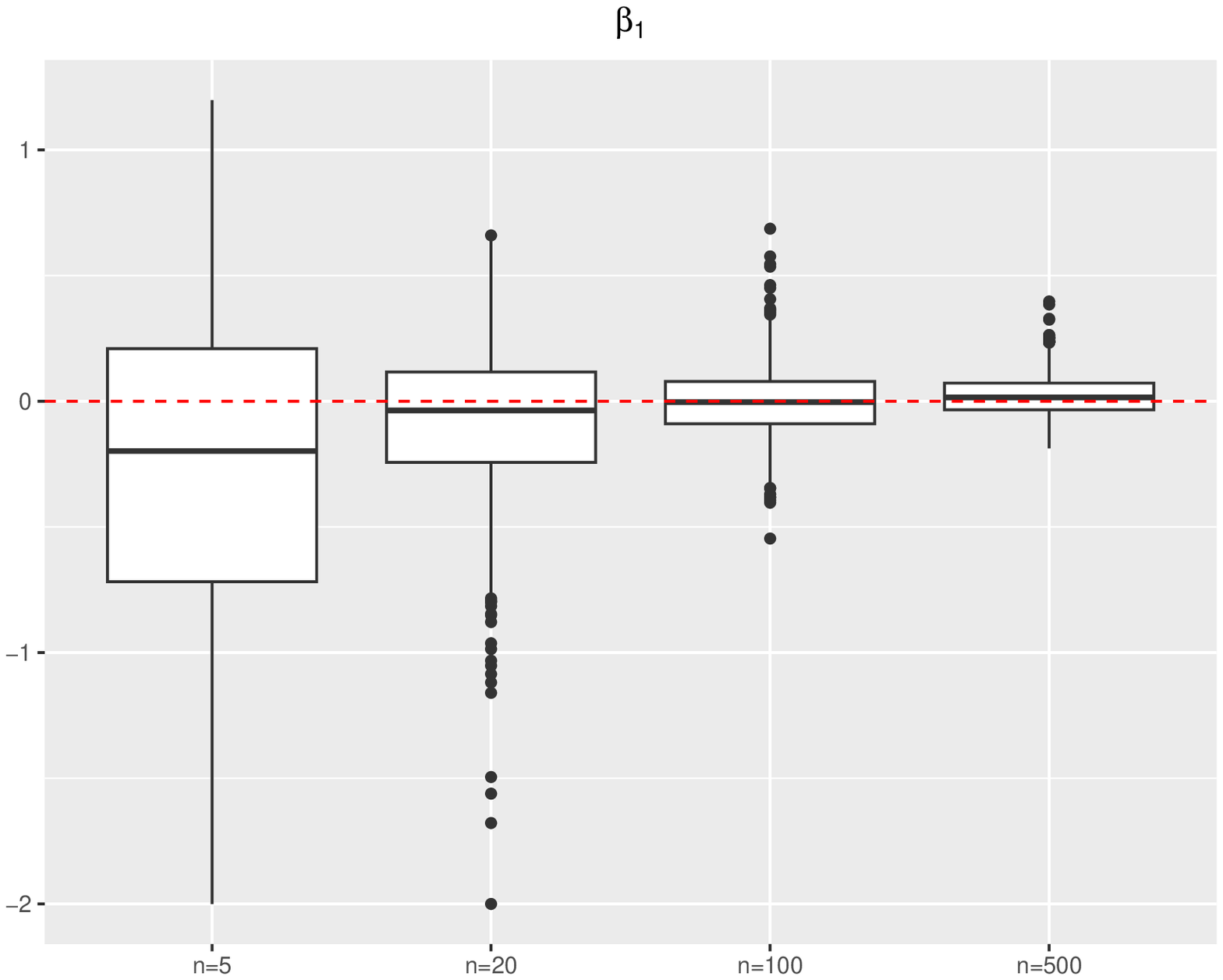}   \includegraphics[scale=0.28]{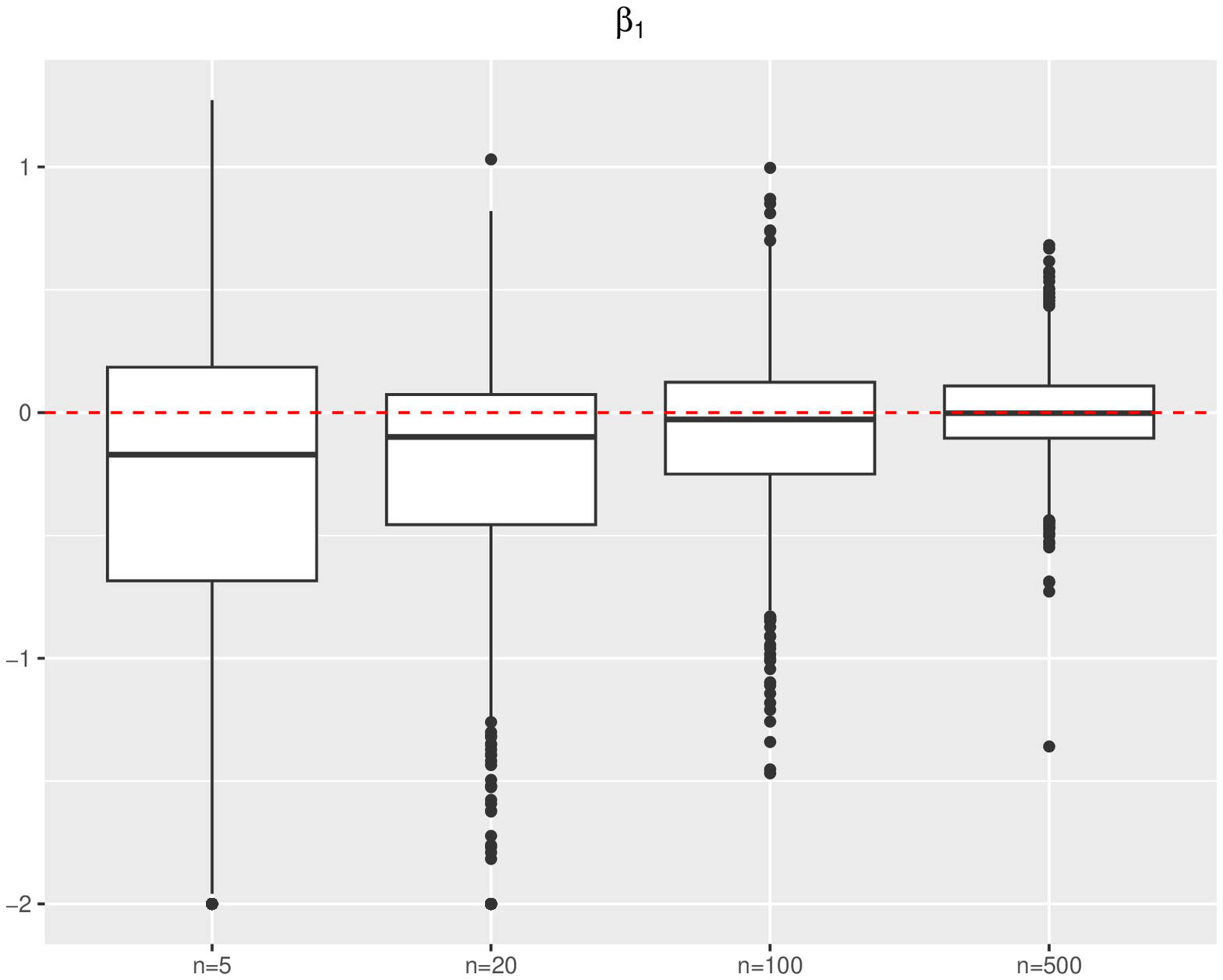}
    \includegraphics[scale=0.28]{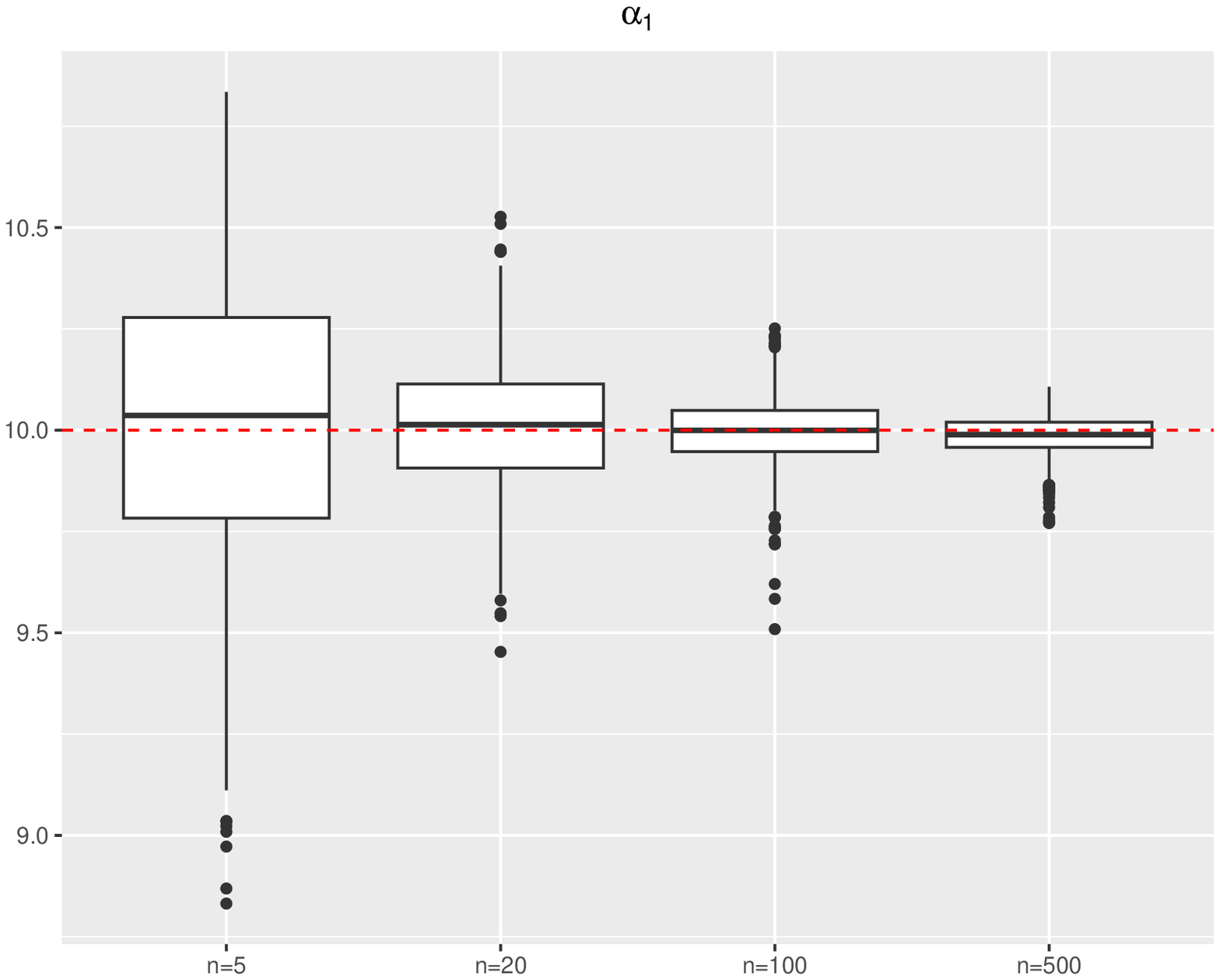}    \includegraphics[scale=0.28]{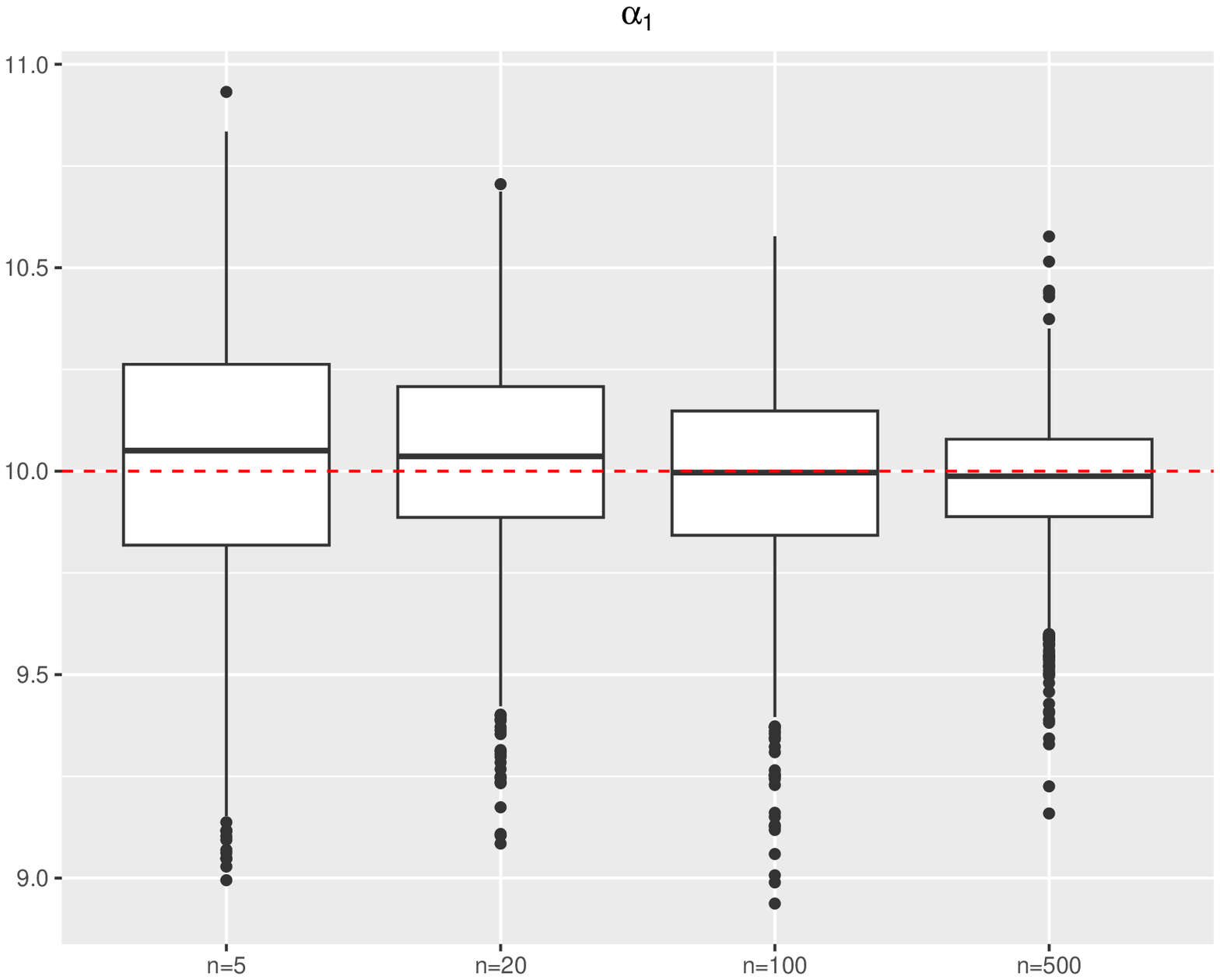}
    \includegraphics[scale=0.28]{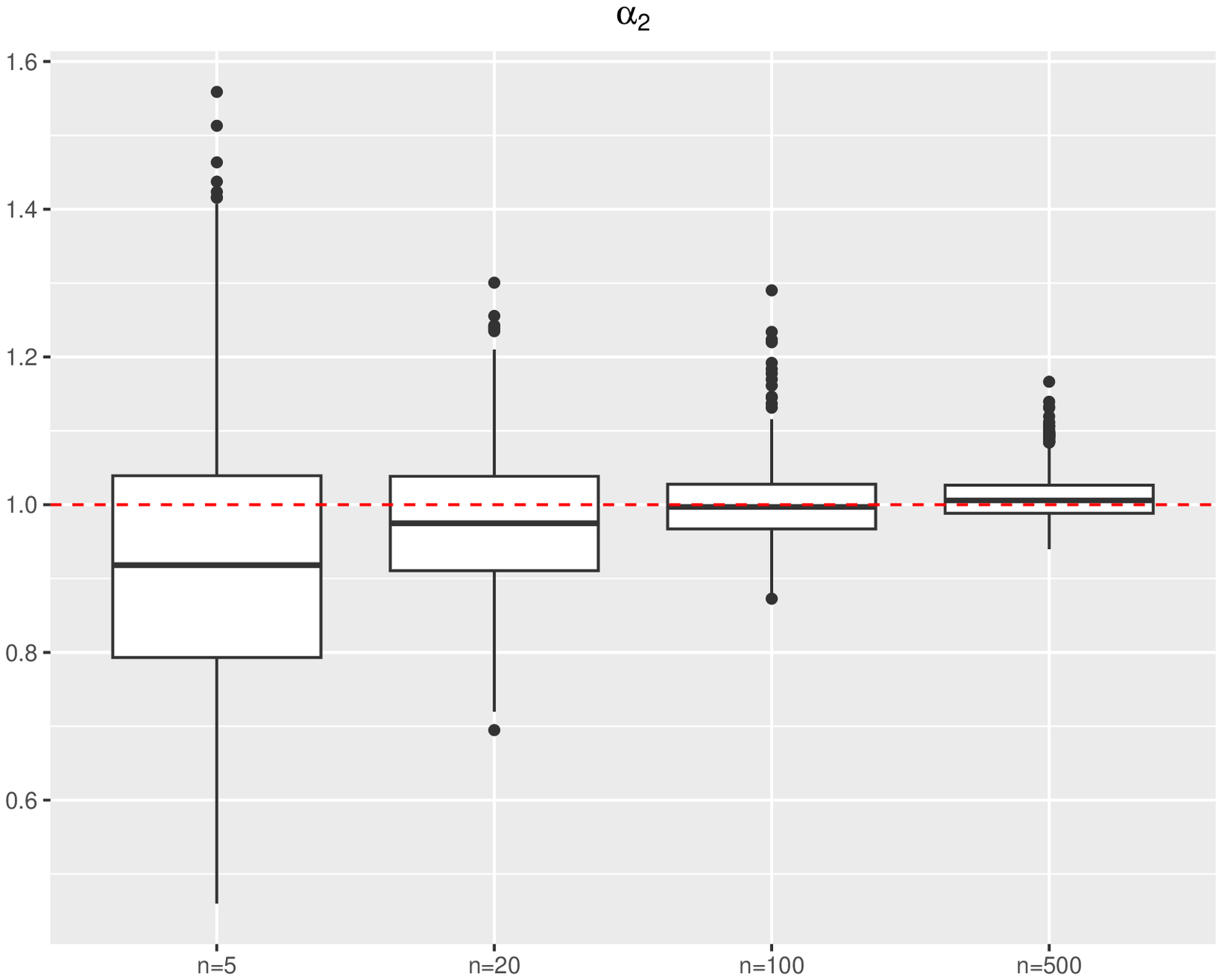}    \includegraphics[scale=0.28]{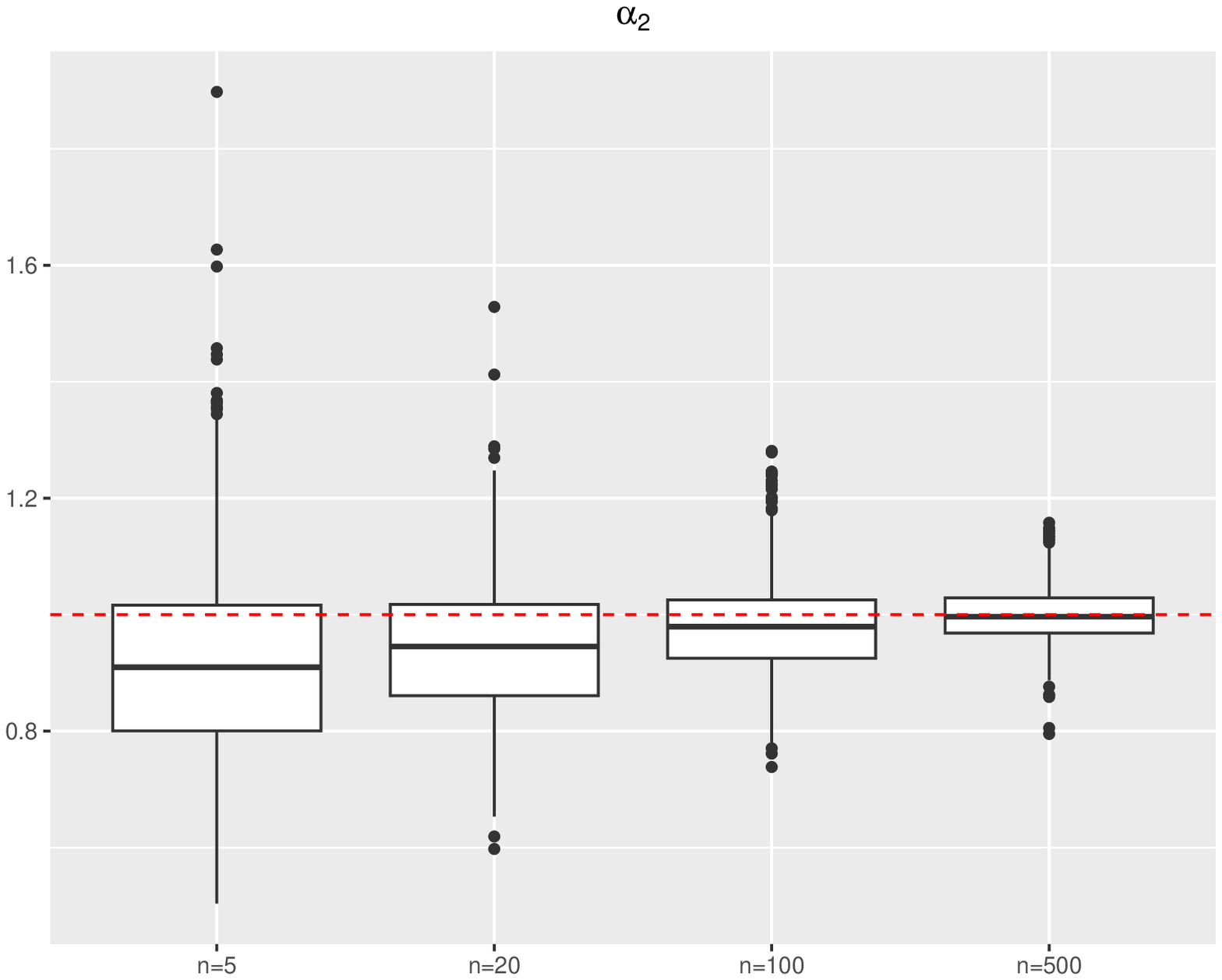}

       \caption{Exp1: Boxplots of parameter estimates for scenario 1 (left) and scenario 2 (right), based on 1000 samples  from  Clayton copula with parameter $2$ and $N(10,1)$ margins.}
    \label{fig:exp1}
\end{figure}
Next, in the second numerical experiment, we used the same margins as in the first experiment, but we considered Clayton copulas $C_{\bphi(\bbbeta, \bX_{ki})}$ with parameter $2e^{\bphi(\bbbeta,\bX_{ki}) }$, where $\bphi(\bbbeta,\bX_{ki}) = 1-1.5U_{ki}$, and $U_{ki}$ are iid $\unif$, $k\in\setK$, $i\in \{1,\ldots, n_k\}$.  Figure \ref{fig:exp2} shows the corresponding results.
 \begin{figure}[ht!]
    \centering
     \includegraphics[scale=0.3]{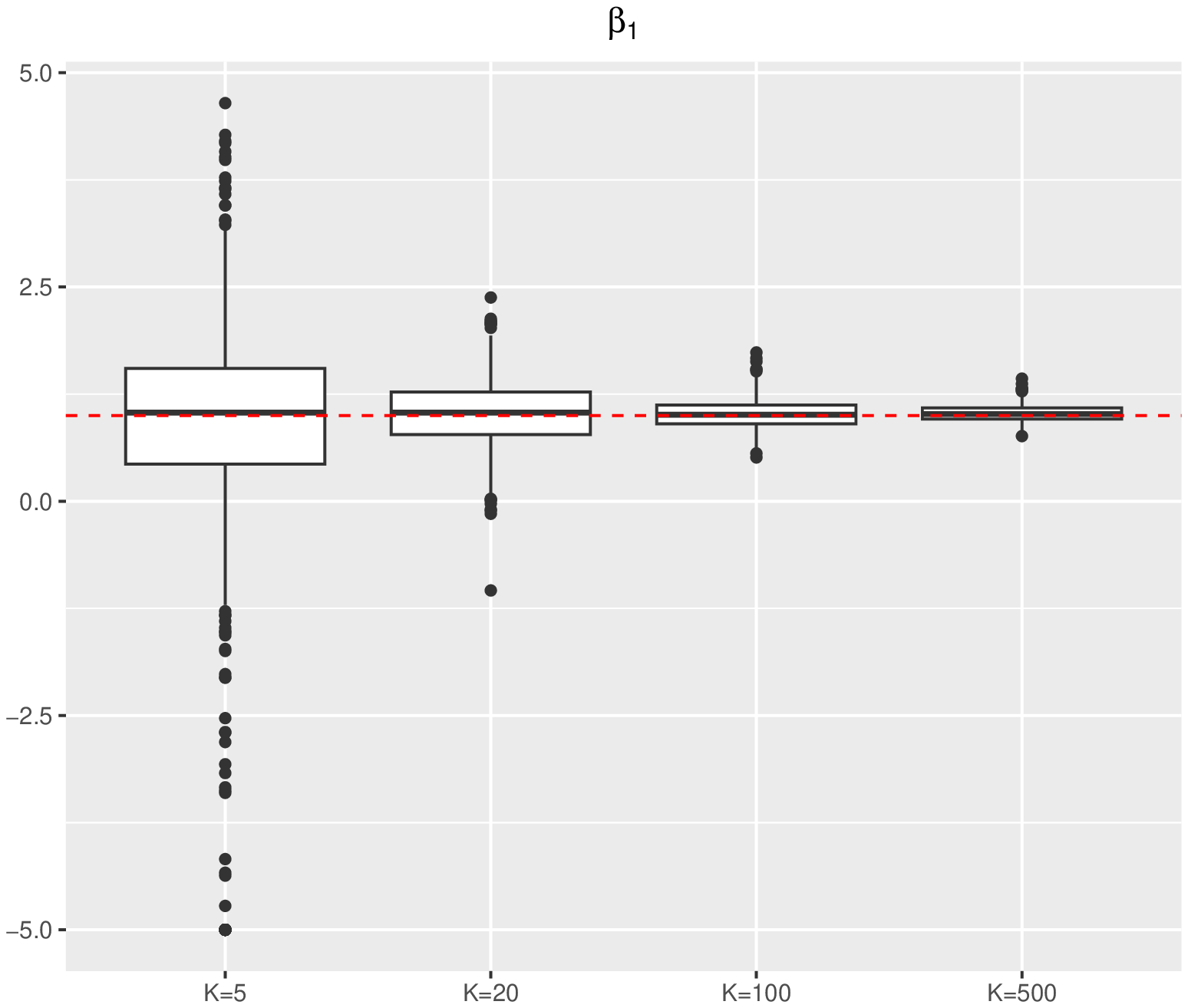}   \includegraphics[scale=0.3]{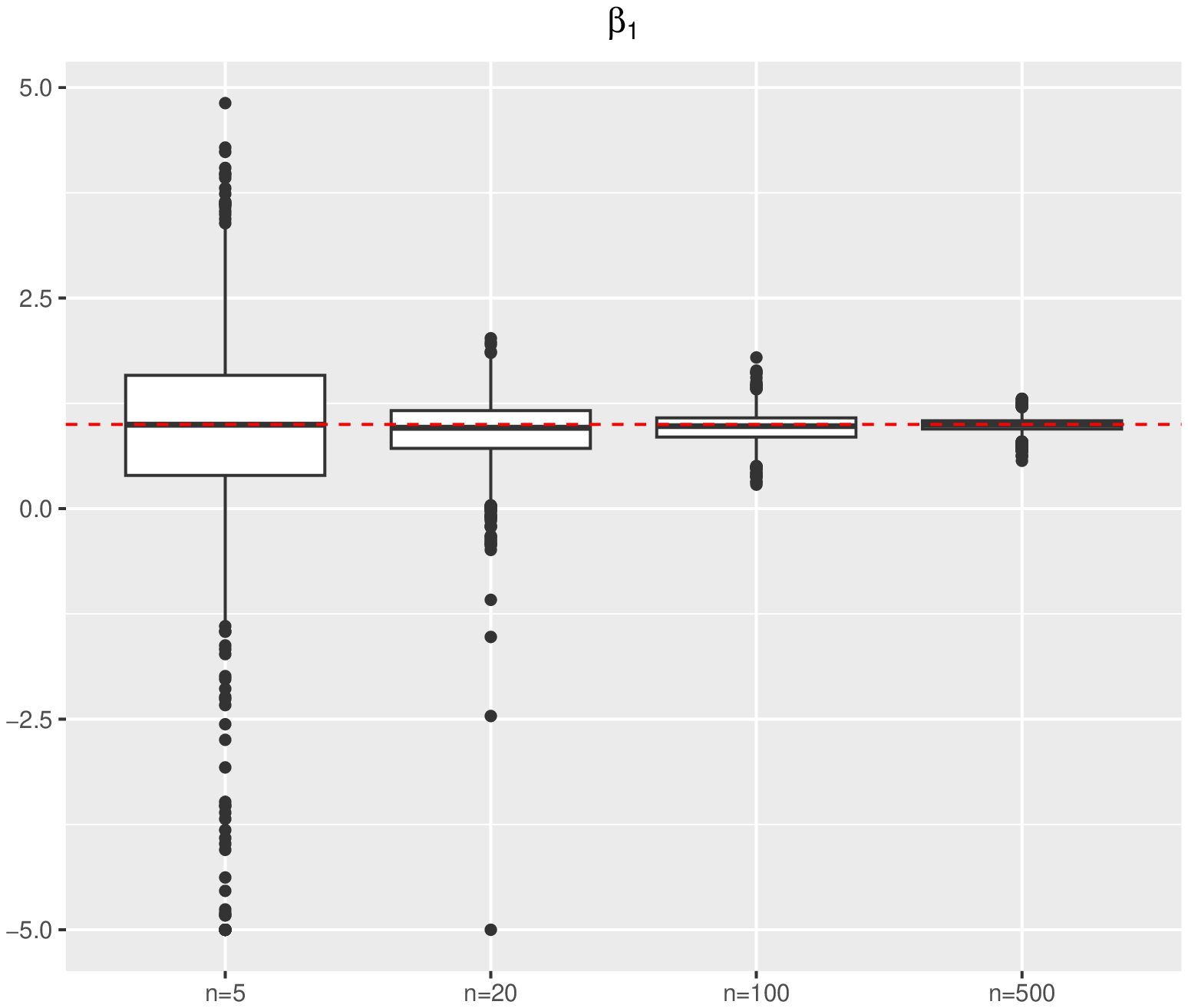}
     \includegraphics[scale=0.3]{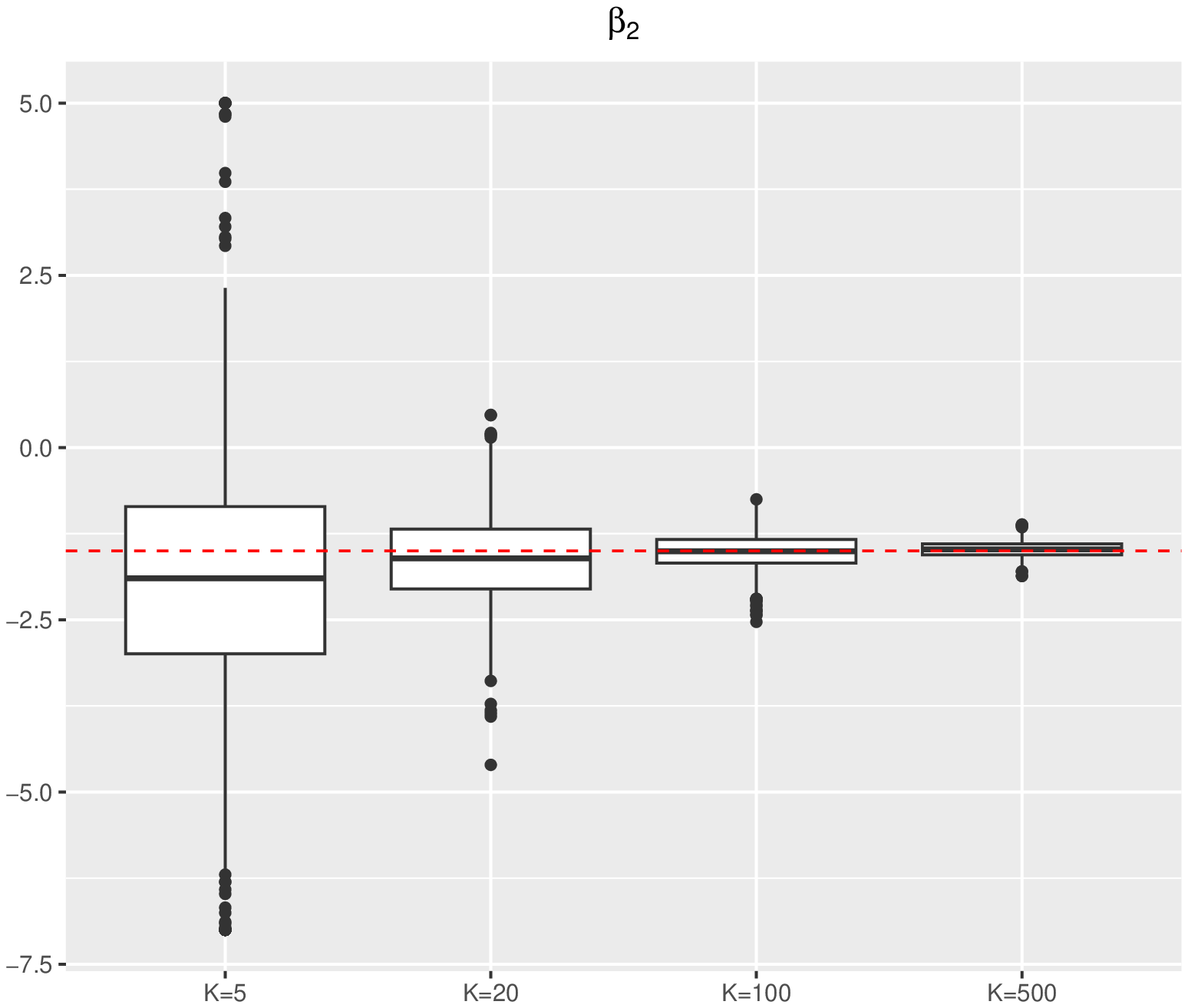}   \includegraphics[scale=0.3]{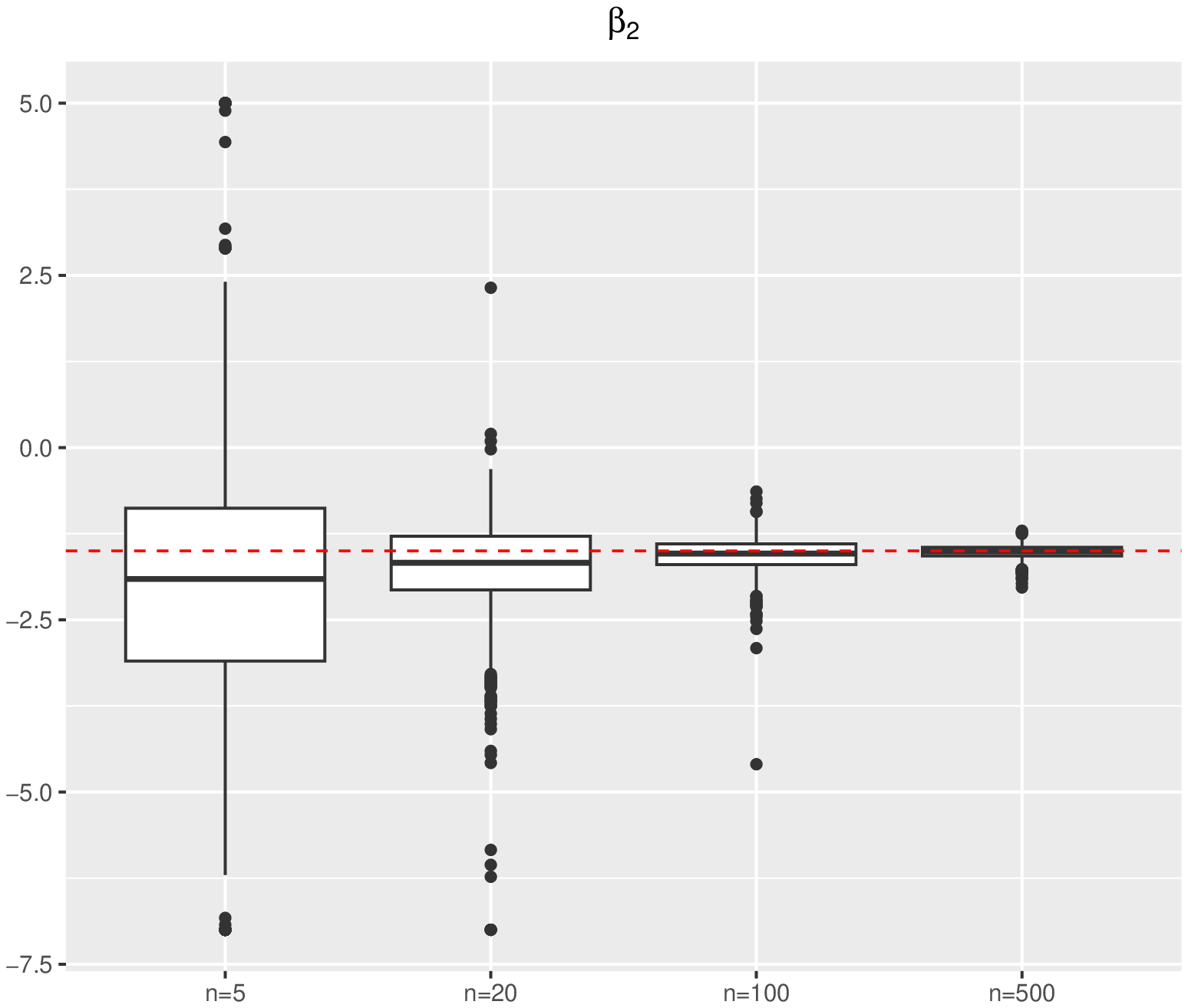}
     \includegraphics[scale=0.3]{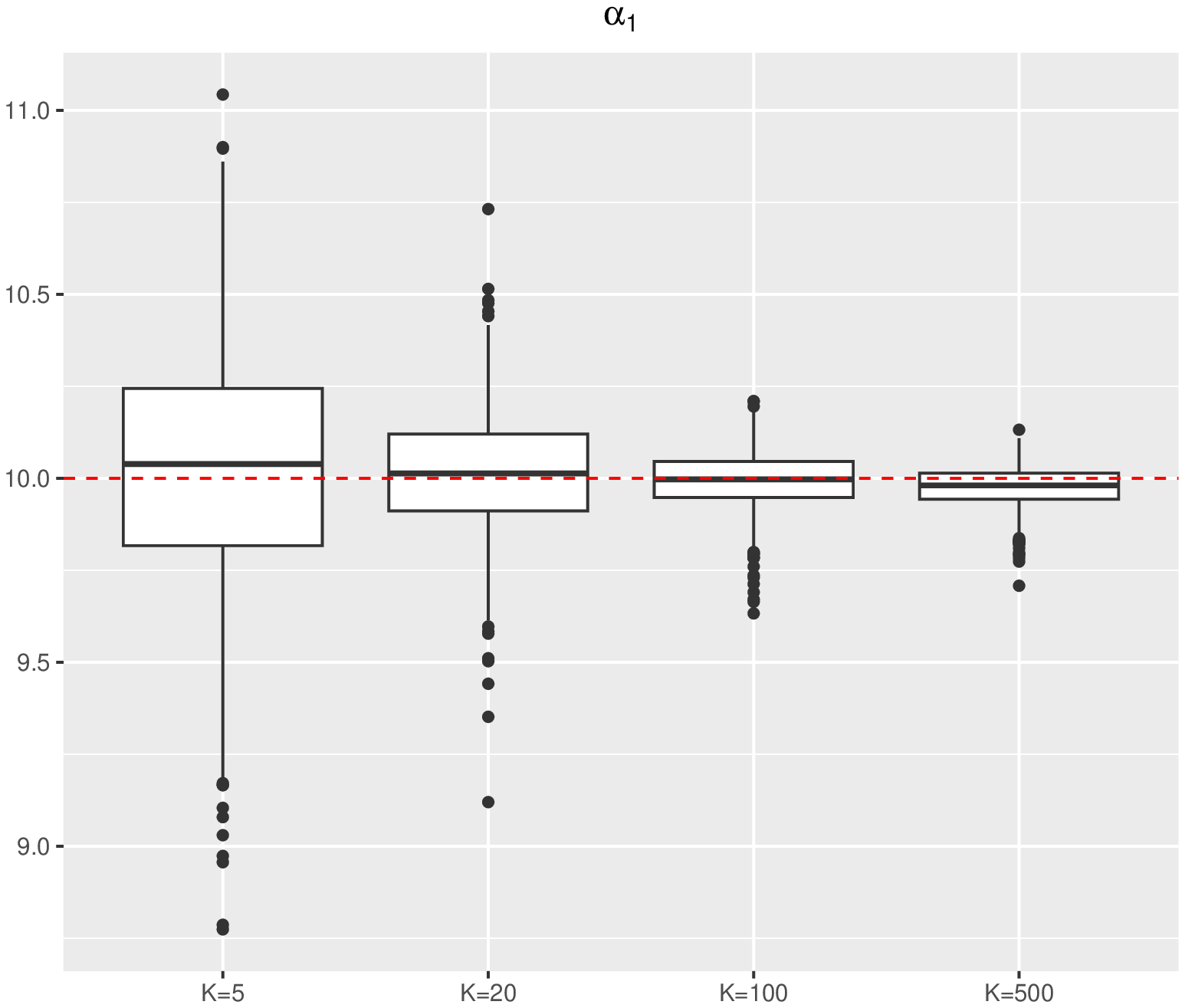}   \includegraphics[scale=0.3]{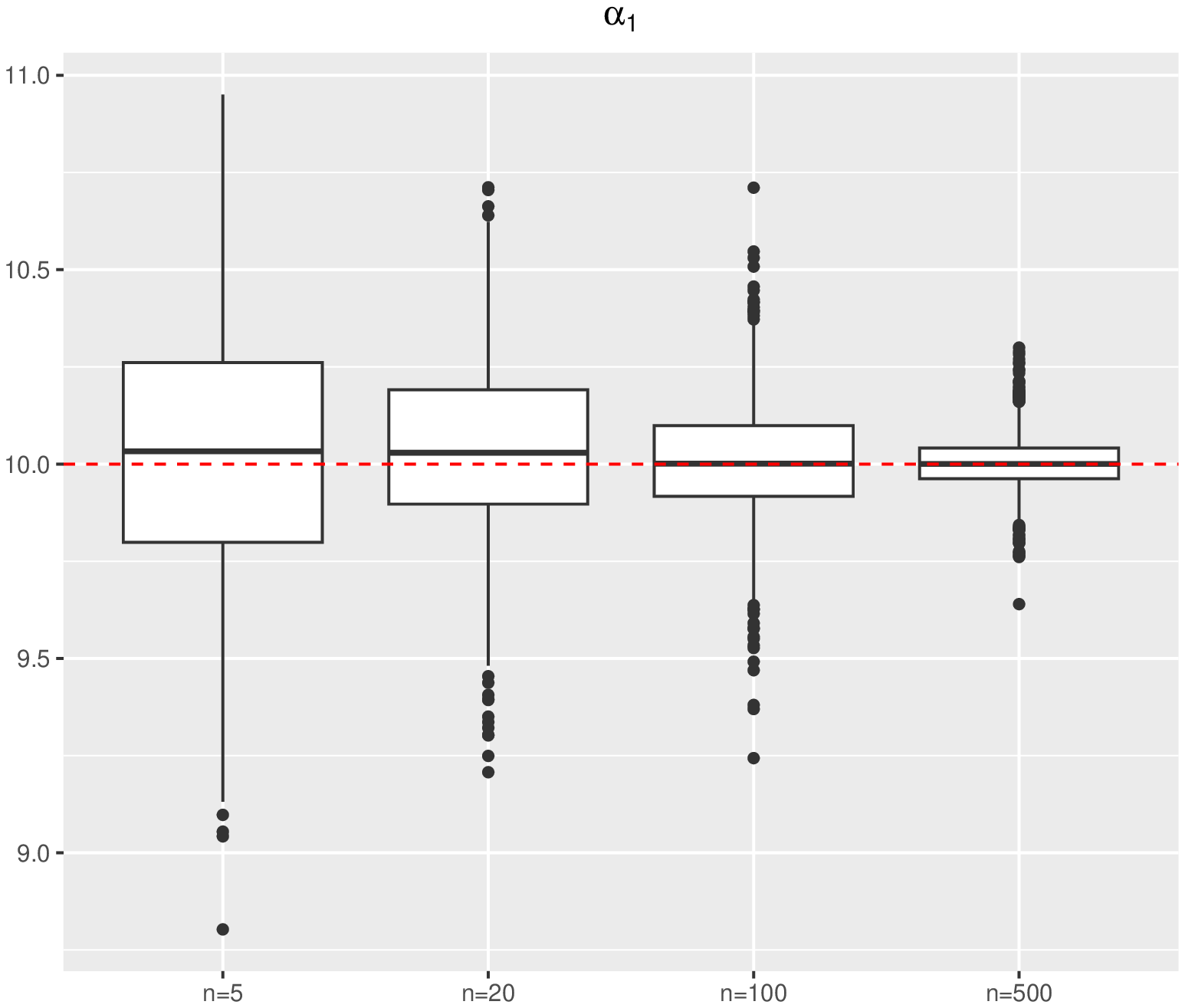}
     \includegraphics[scale=0.3]{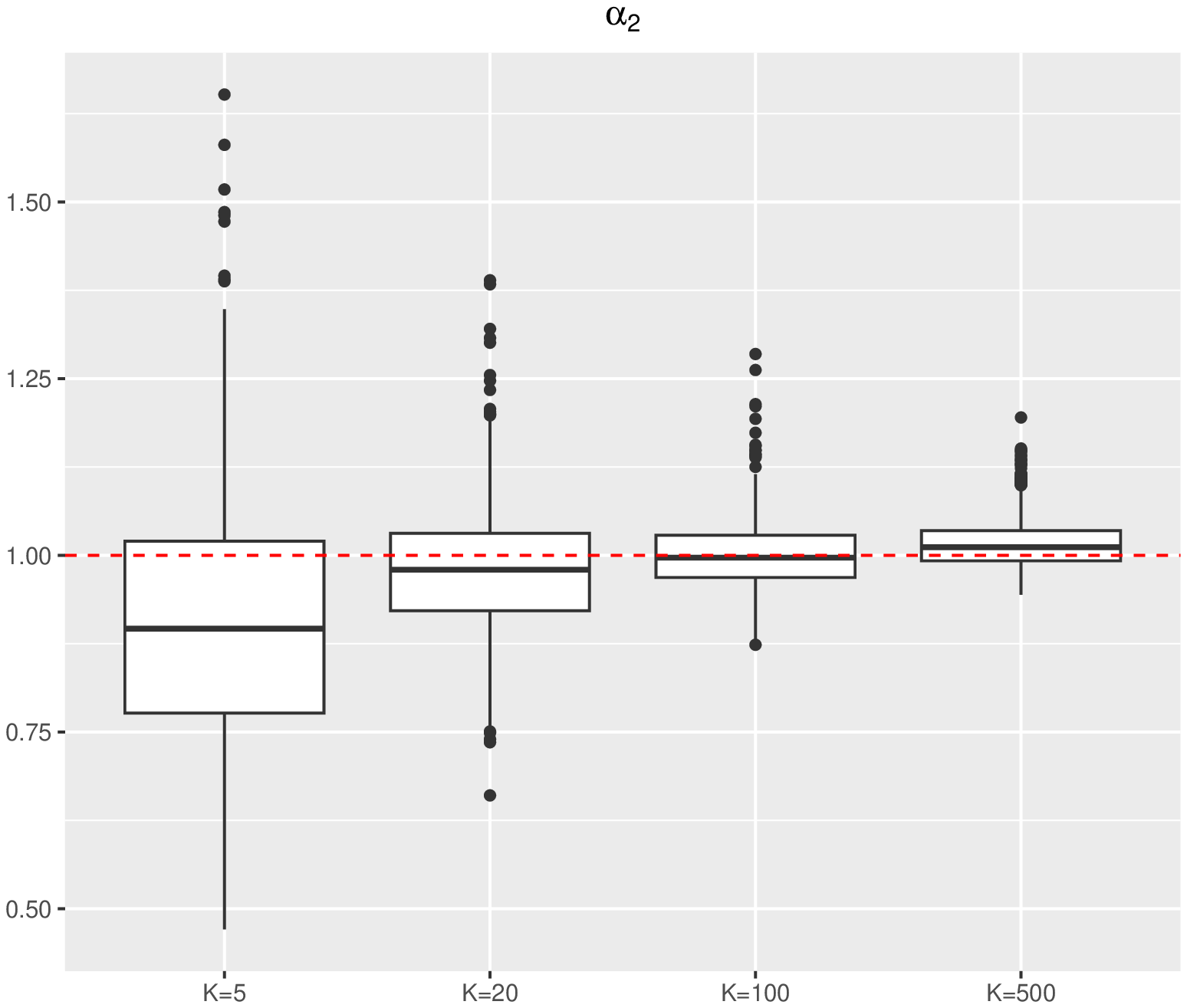}   \includegraphics[scale=0.3]{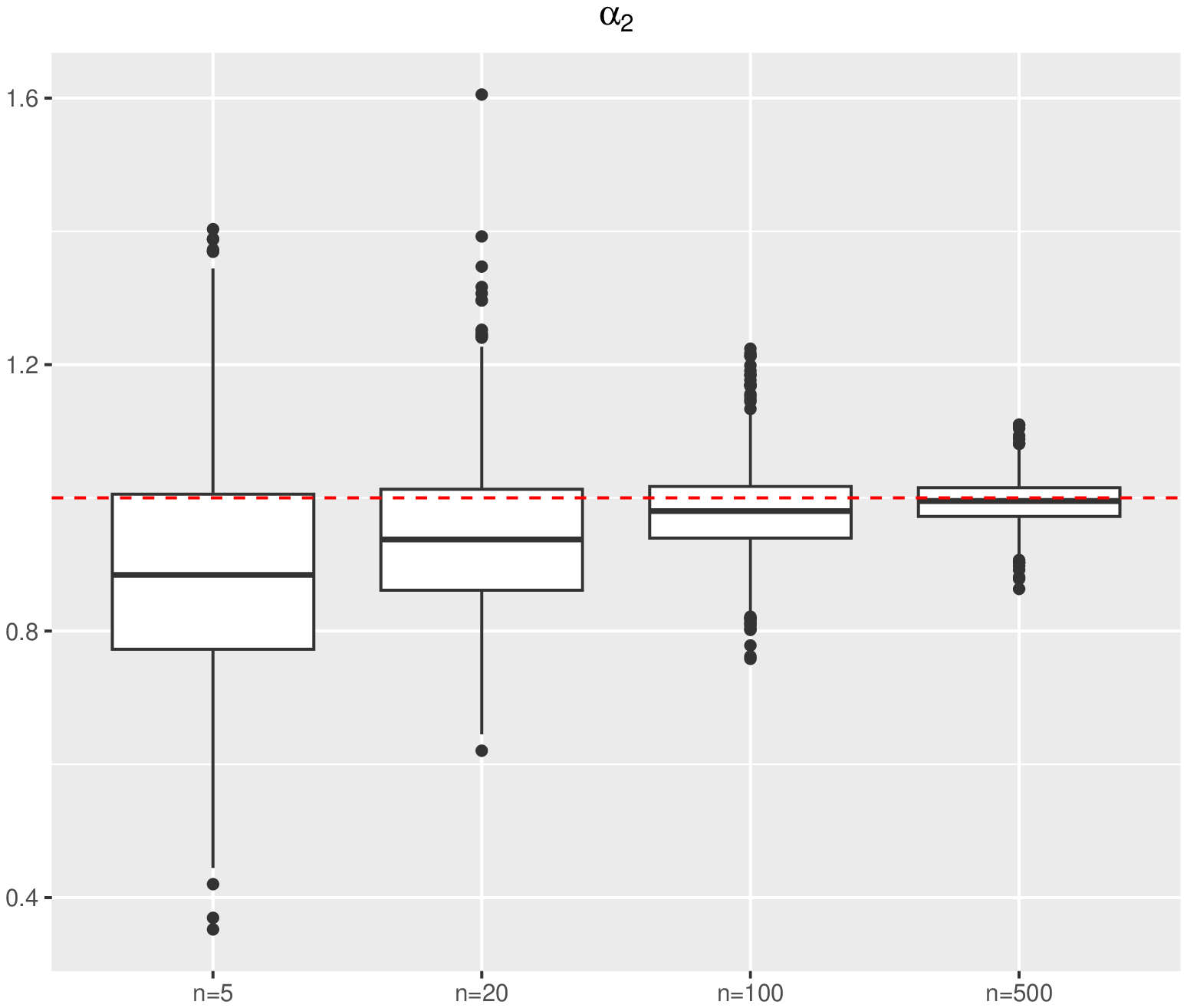}

    \caption{Exp2: Boxplots of parameter estimates for scenario 1 (left) and scenario 2 (right), based on 1000 samples from Clayton bivariate copulas with parameters $2e^{1 - 1.5 U_{ki}}$ and $N(10,1)$ margins.}
    \label{fig:exp2}
\end{figure}
For the third numerical experiment,
we used  Gumbel copulas  with parameters $1+e^{\phi(\bbbeta, \bX_{ki})}$, where $\phi(\bbbeta, \bX_{ki}) = 1-1.5 U_{ki}$, and the marginal distribution $G_{\bh(\balpha,\bX_{ki})}$ is  normal with  mean $5+5Z_{ki}$ and variance $1$, where $U_{ki}$ are iid $ \unif$, and $Z_{ki}$ are iid $ {\rm Exp}(1)$, $k\in\setK$, $i\in \{1,\ldots, n_k\}$.
 The results of this experiment for the two scenarios are displayed in  Figure \ref{fig:exp3}.
 \begin{figure}[ht!]
    \centering
    \includegraphics[scale=0.3]{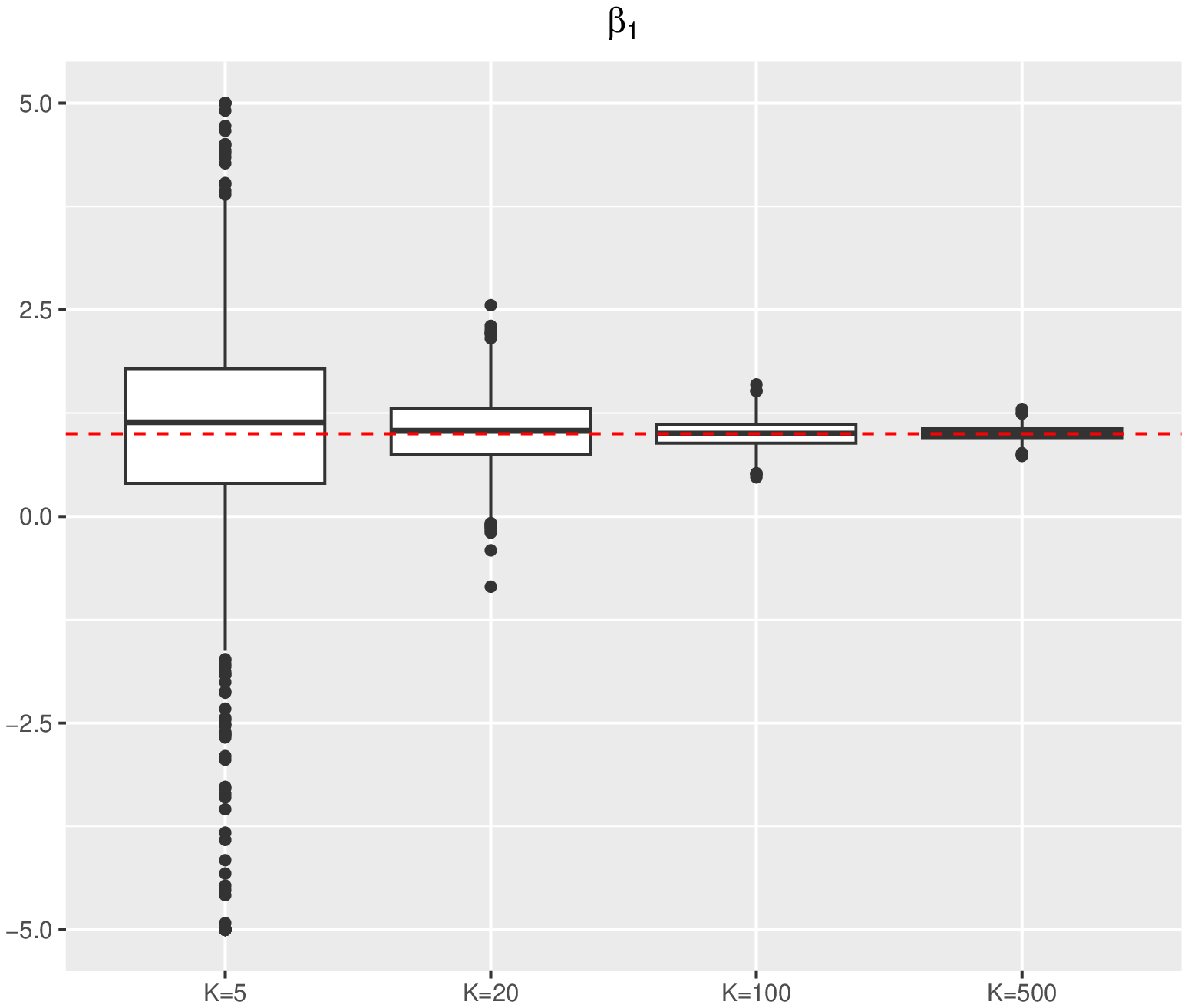}   \includegraphics[scale=0.3]{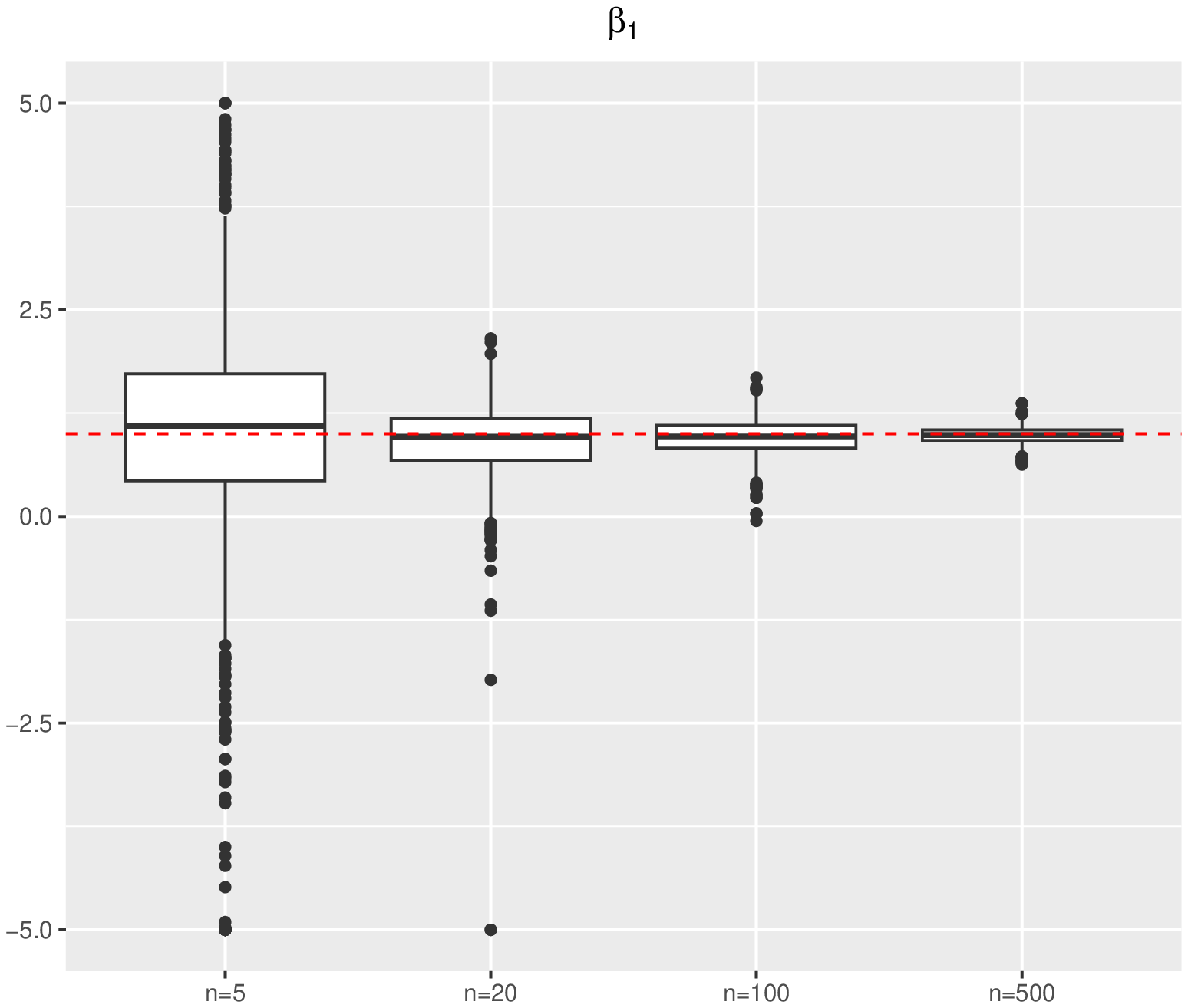}
    \includegraphics[scale=0.3]{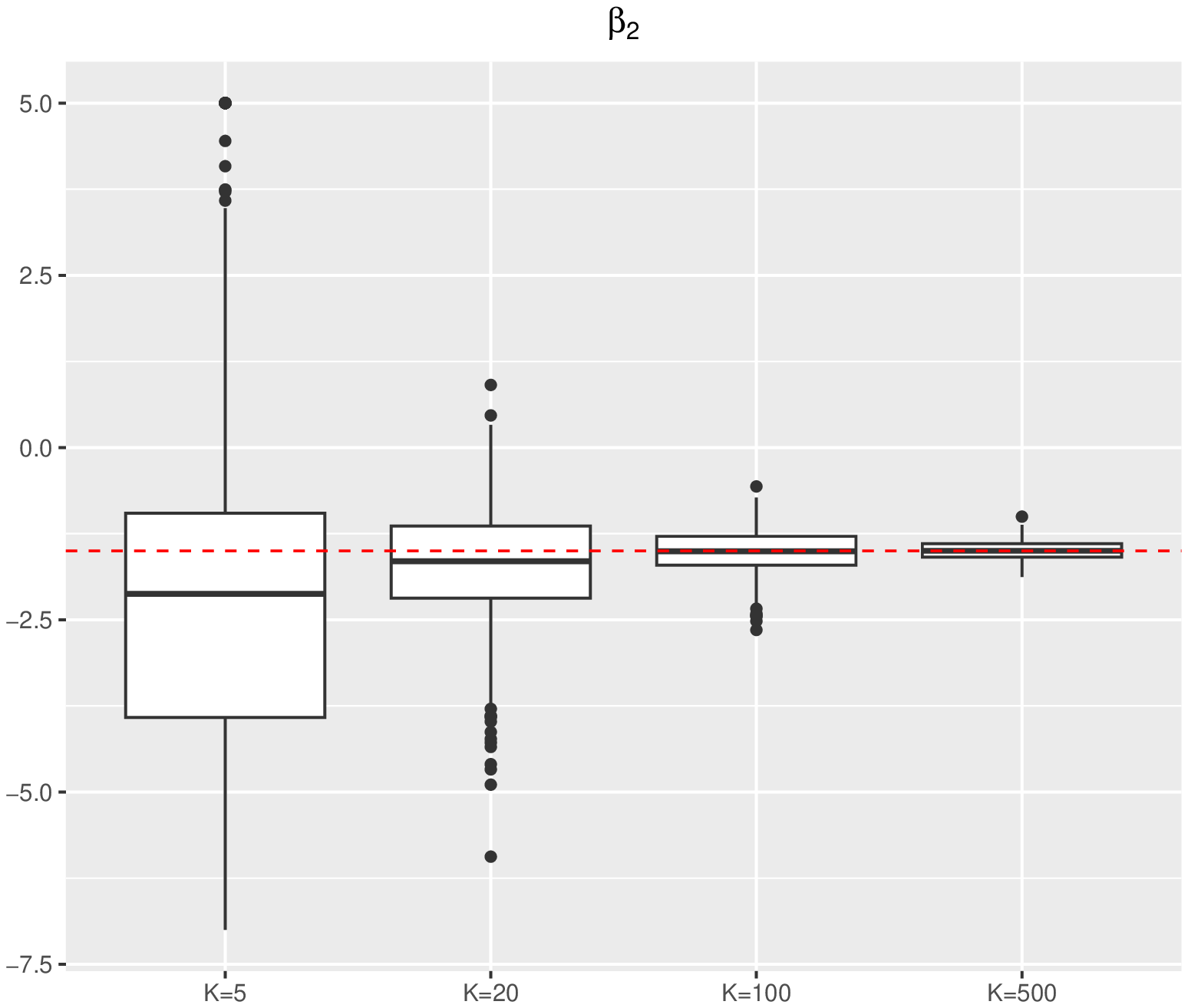}   \includegraphics[scale=0.3]{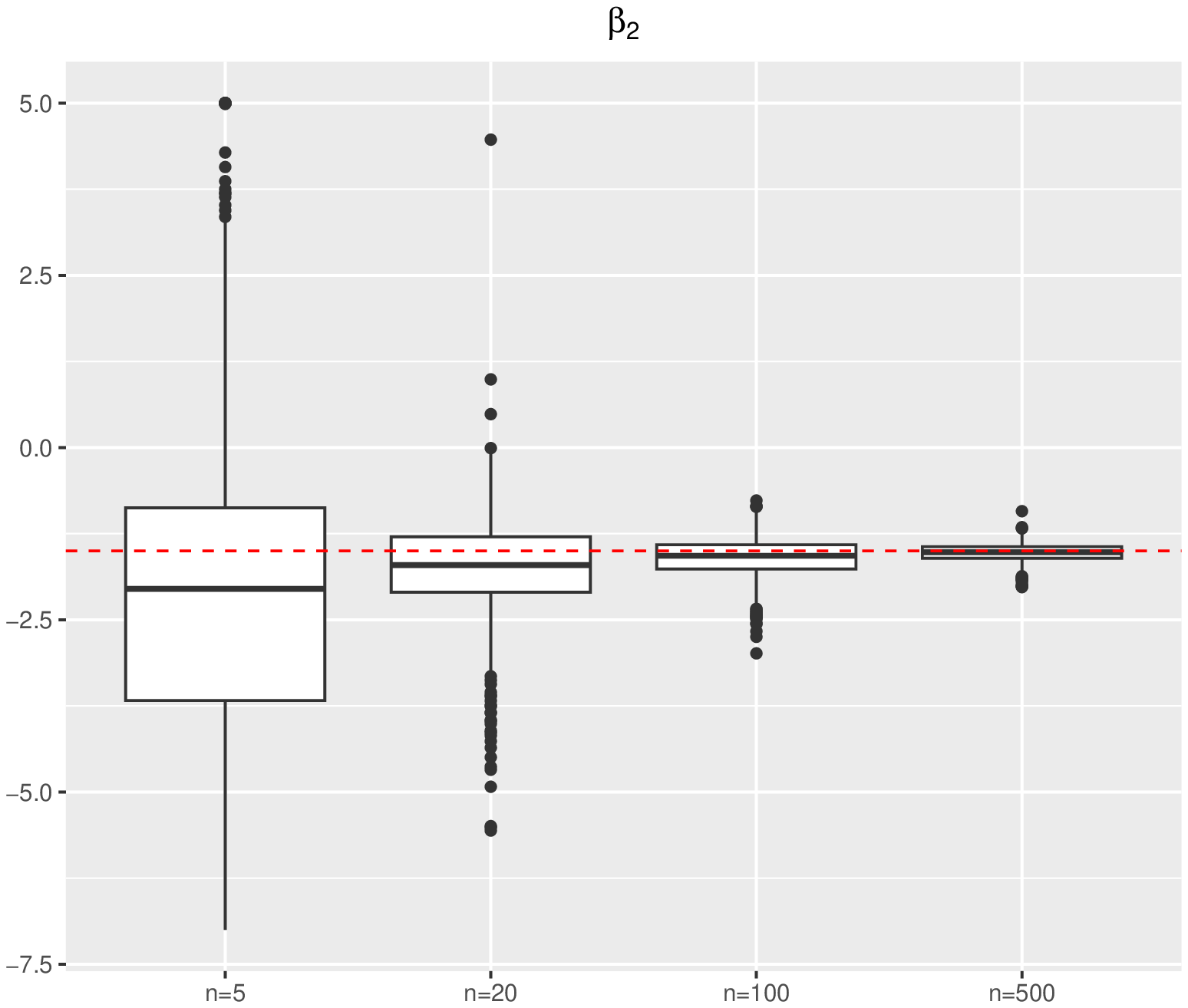}
    \includegraphics[scale=0.3]{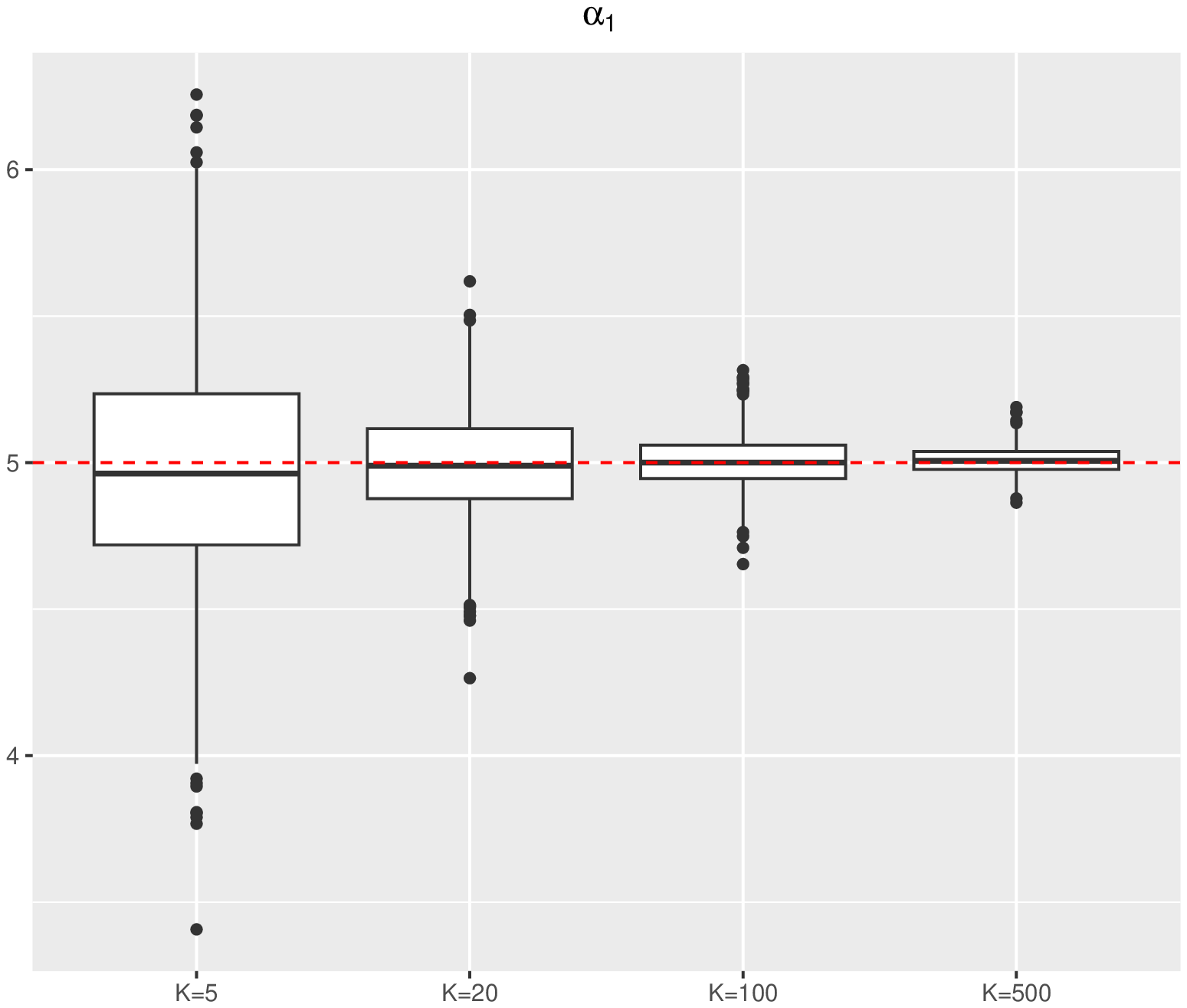}  \includegraphics[scale=0.3]{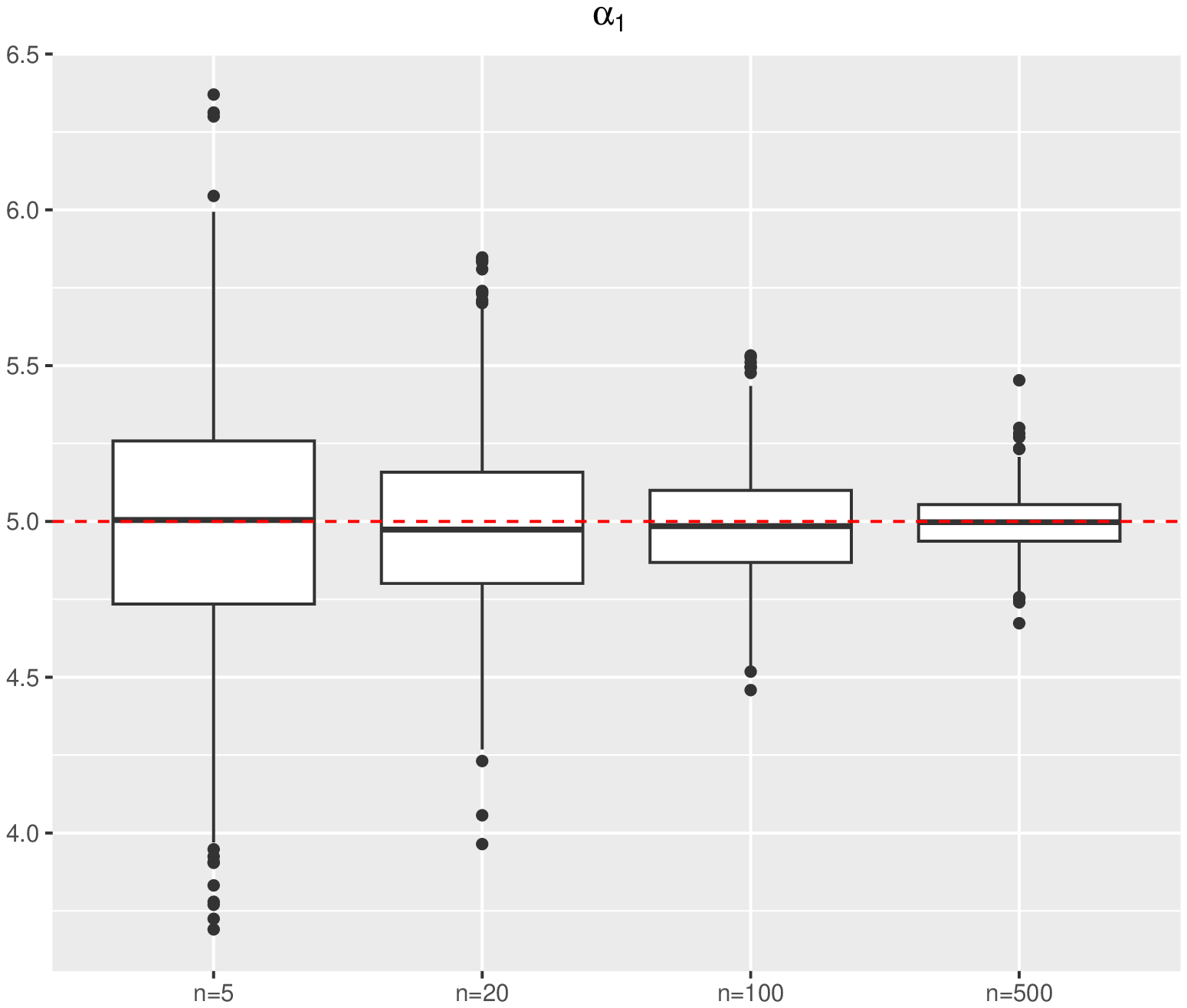}
    \includegraphics[scale=0.3]{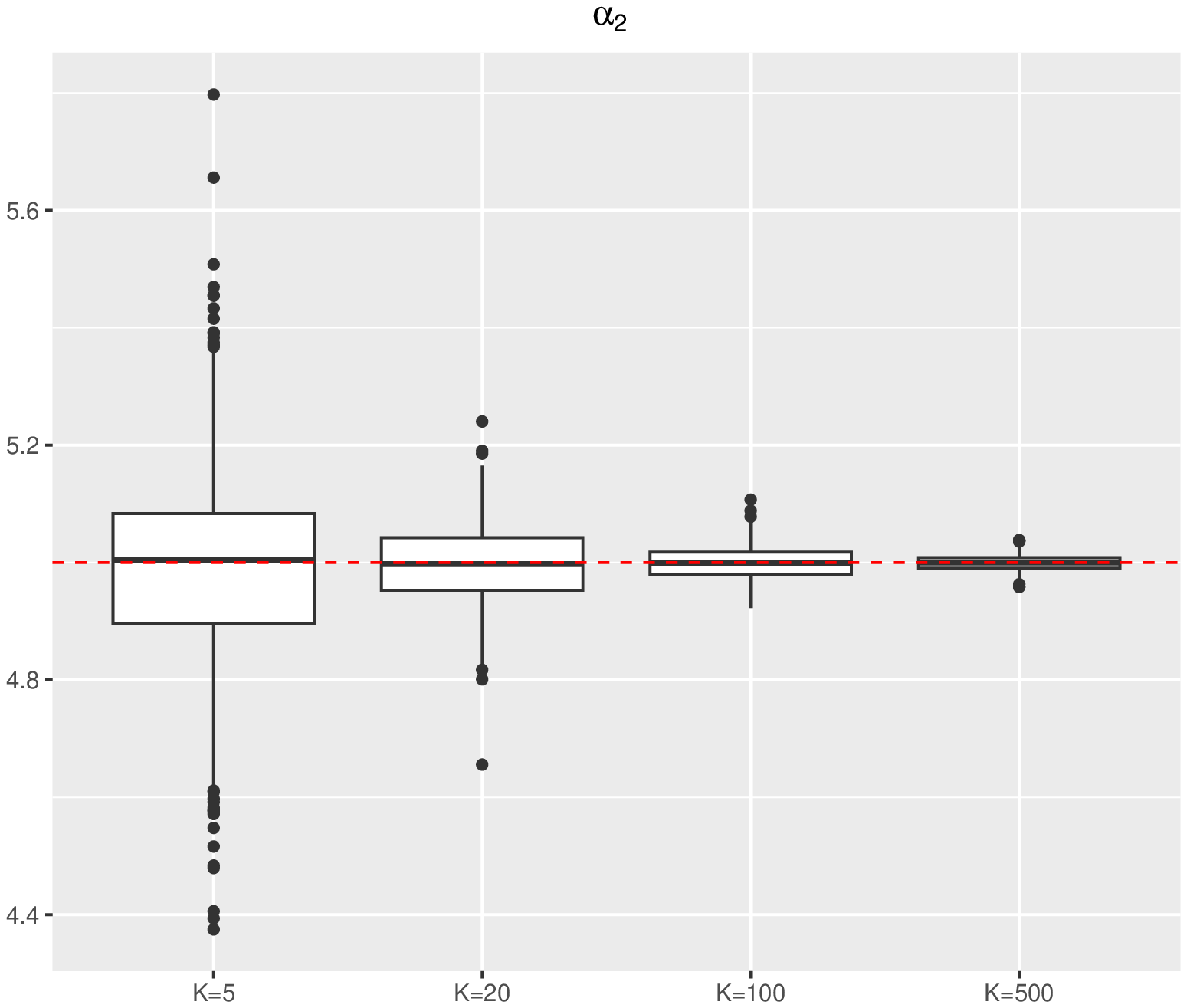}  \includegraphics[scale=0.3]{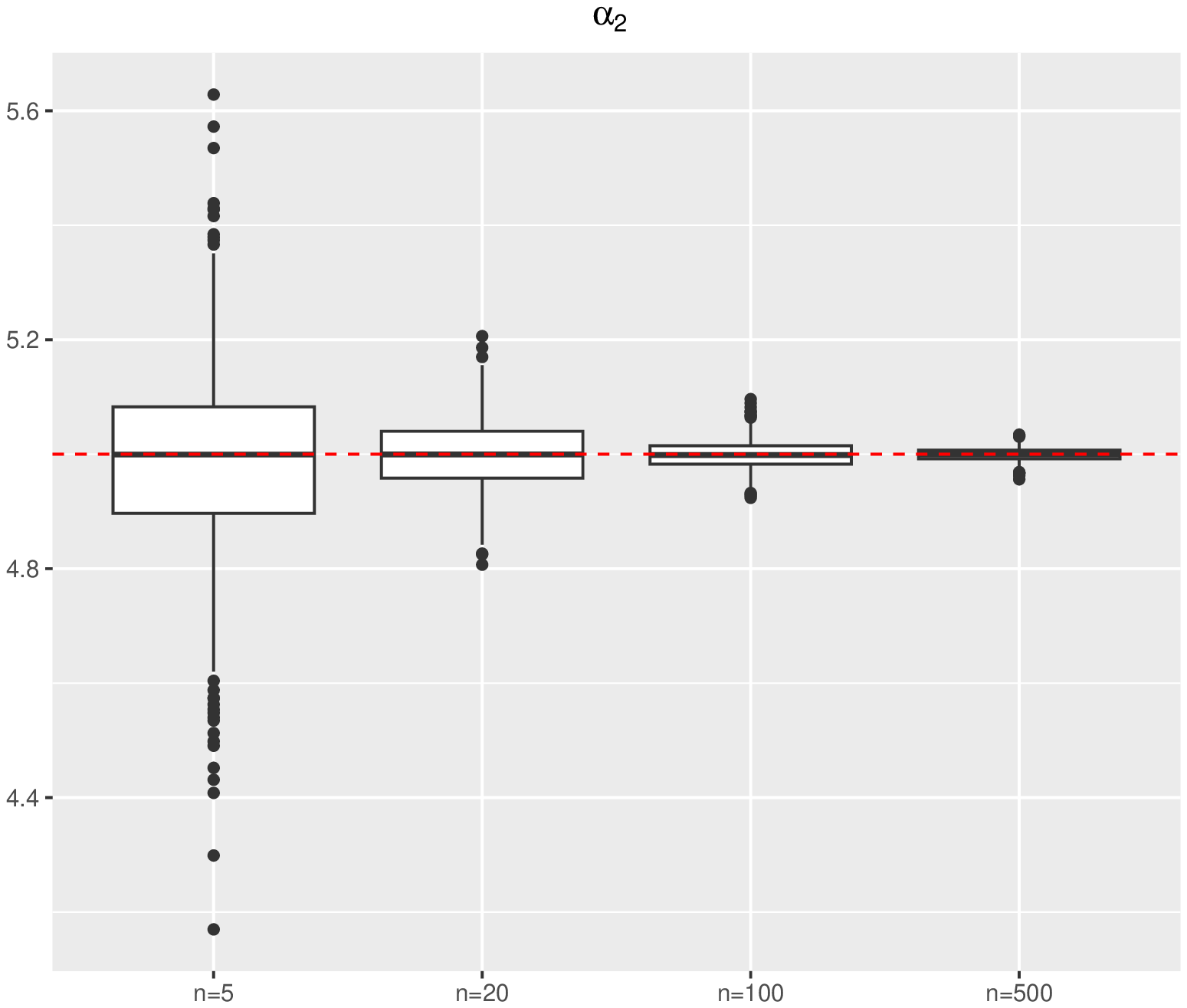}
    \includegraphics[scale=0.3]{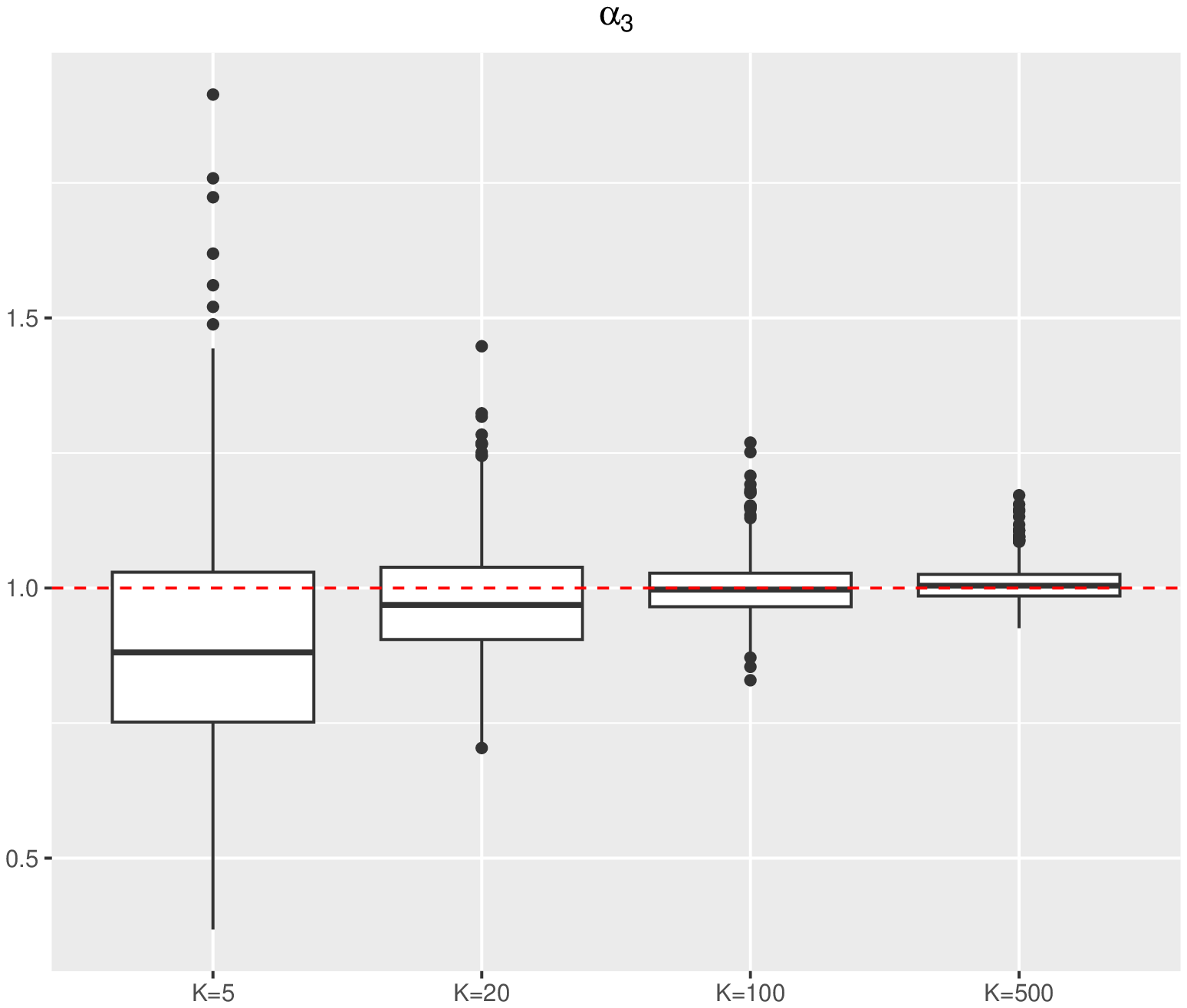}   \includegraphics[scale=0.3]{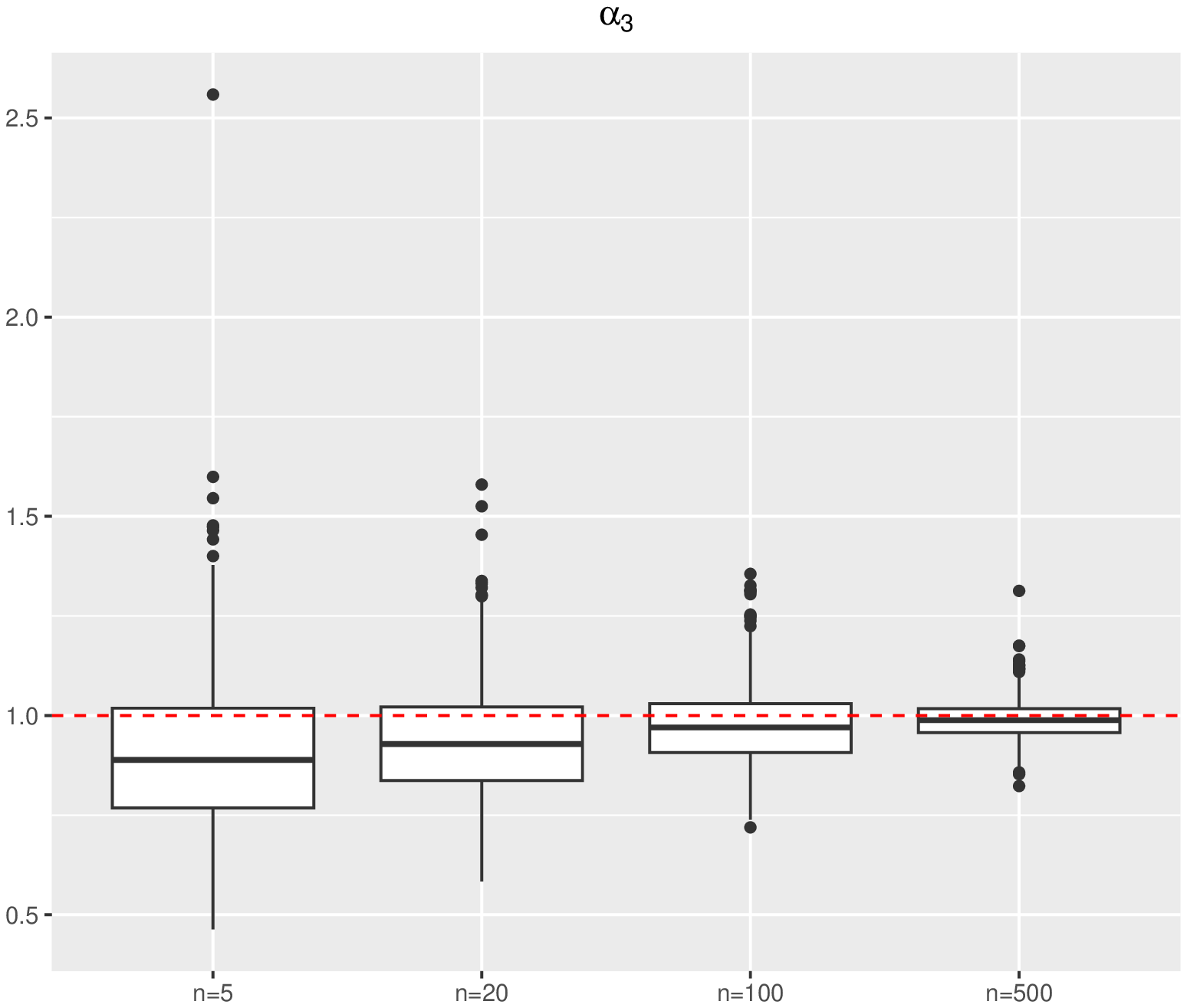}

    \caption{Exp3: Boxplots of parameter estimates for scenario 1 (left) and scenario 2 (right), based on 1000 samples from Gumbel bivariate copula with parameters $1+e^{\phi(\bbbeta,\bX_{ki})}$, where $\phi(\bbbeta,\bX_{ki}) =1 - 1.5U_{ki}$ and $N(5 + 5Z_{ki},1)$ margins.}
    \label{fig:exp3}
\end{figure}
\begin{figure}[ht!]
    \centering
    \includegraphics[scale=0.5]{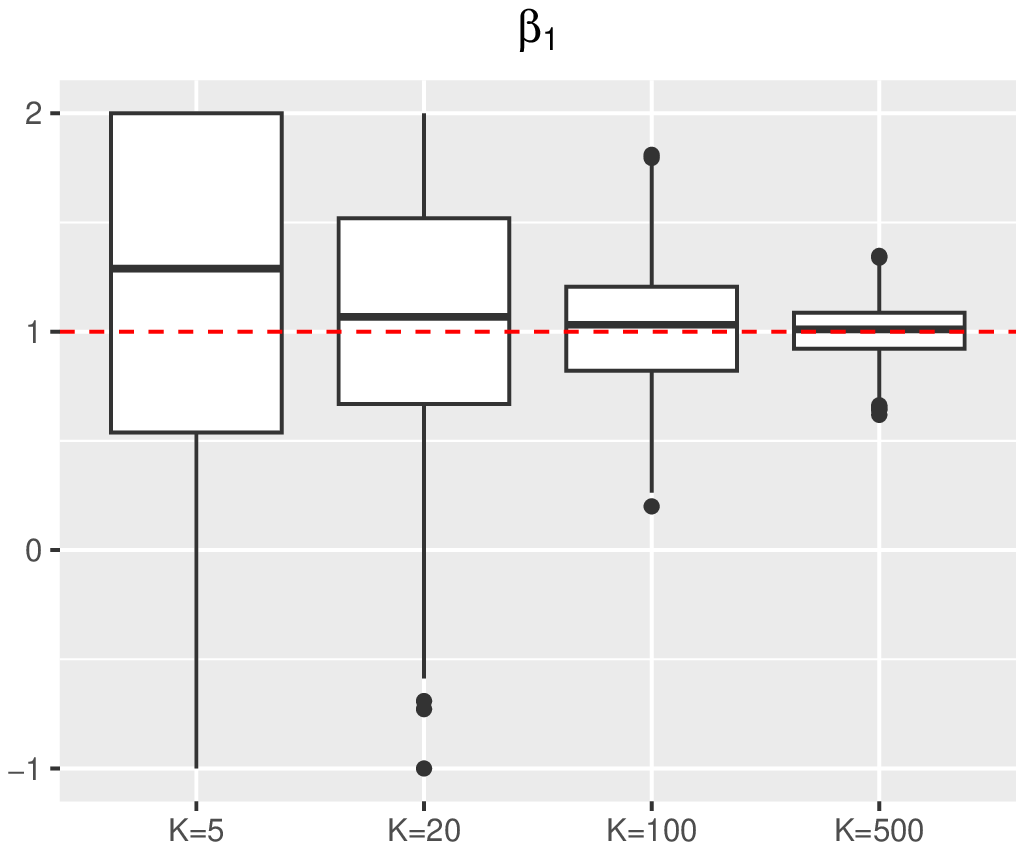}      \includegraphics[scale=0.5]{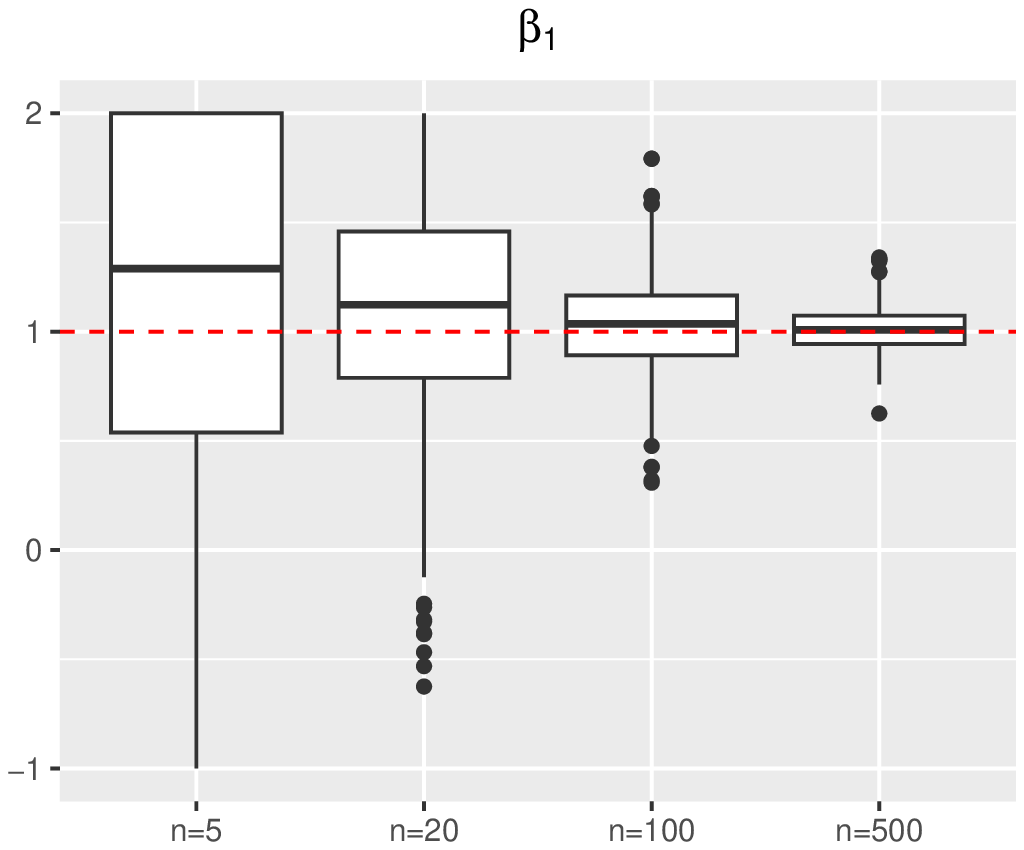}
    \includegraphics[scale=0.5]{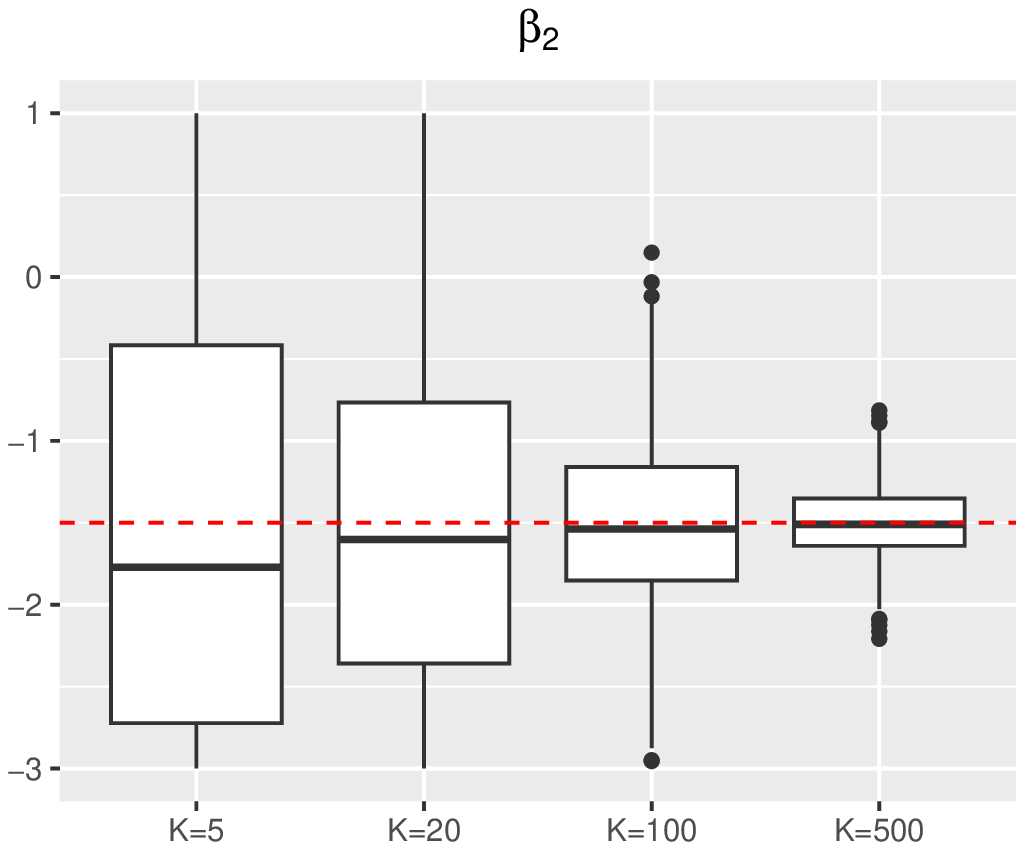}      \includegraphics[scale=0.5]{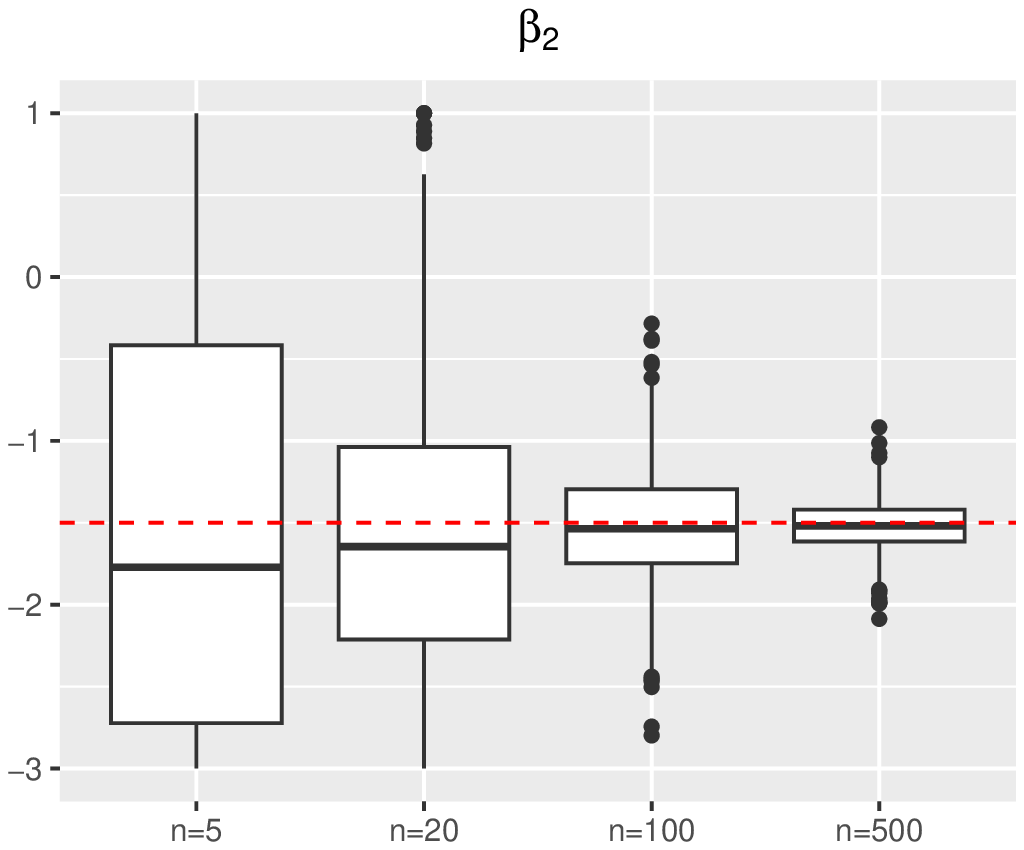}
    \includegraphics[scale=0.5]{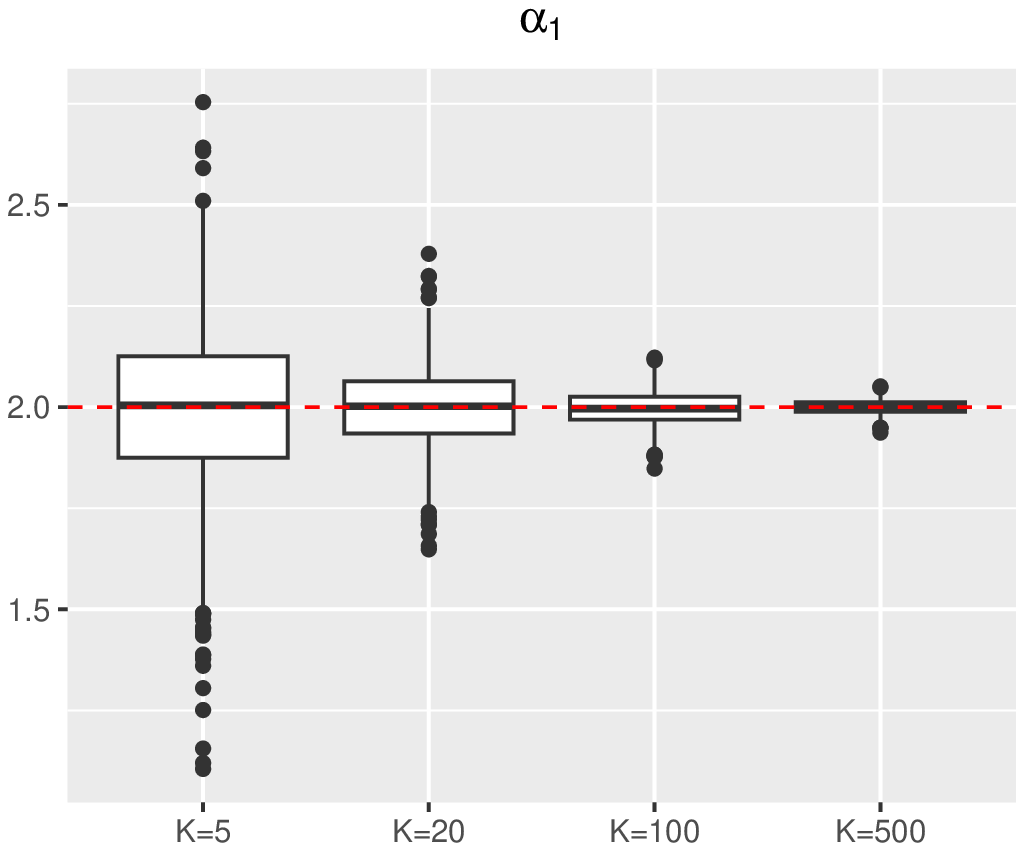}      \includegraphics[scale=0.5]{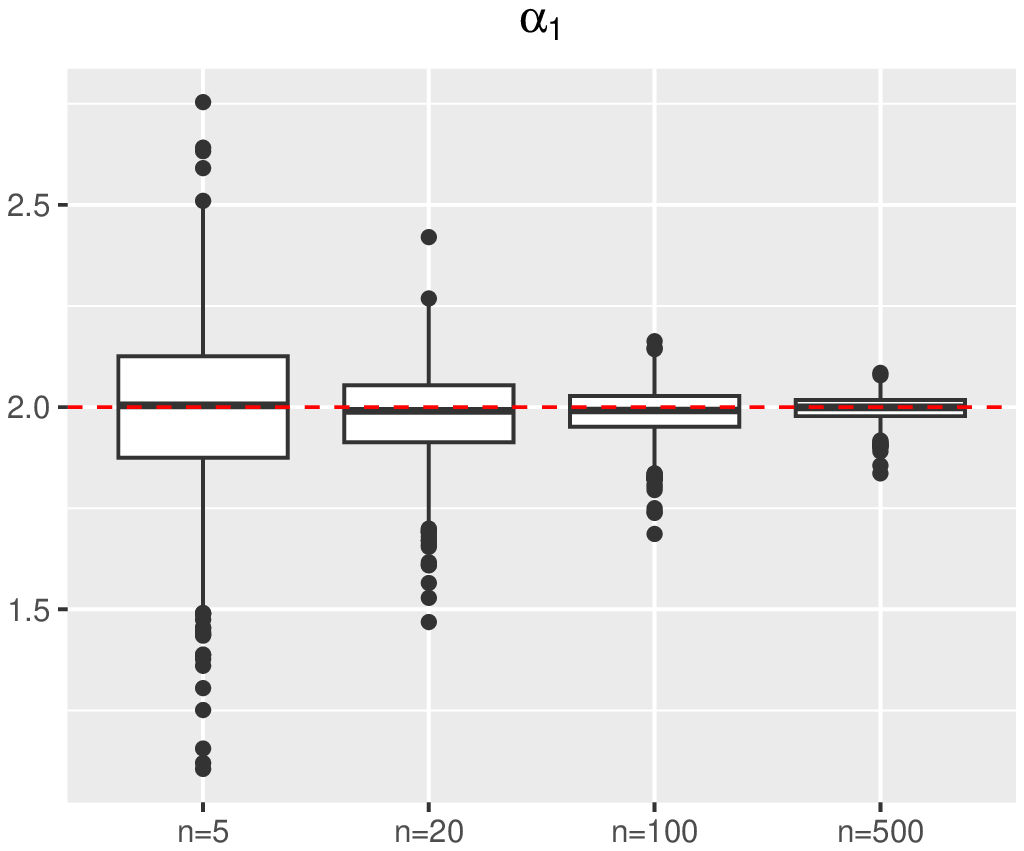}
    \includegraphics[scale=0.5]{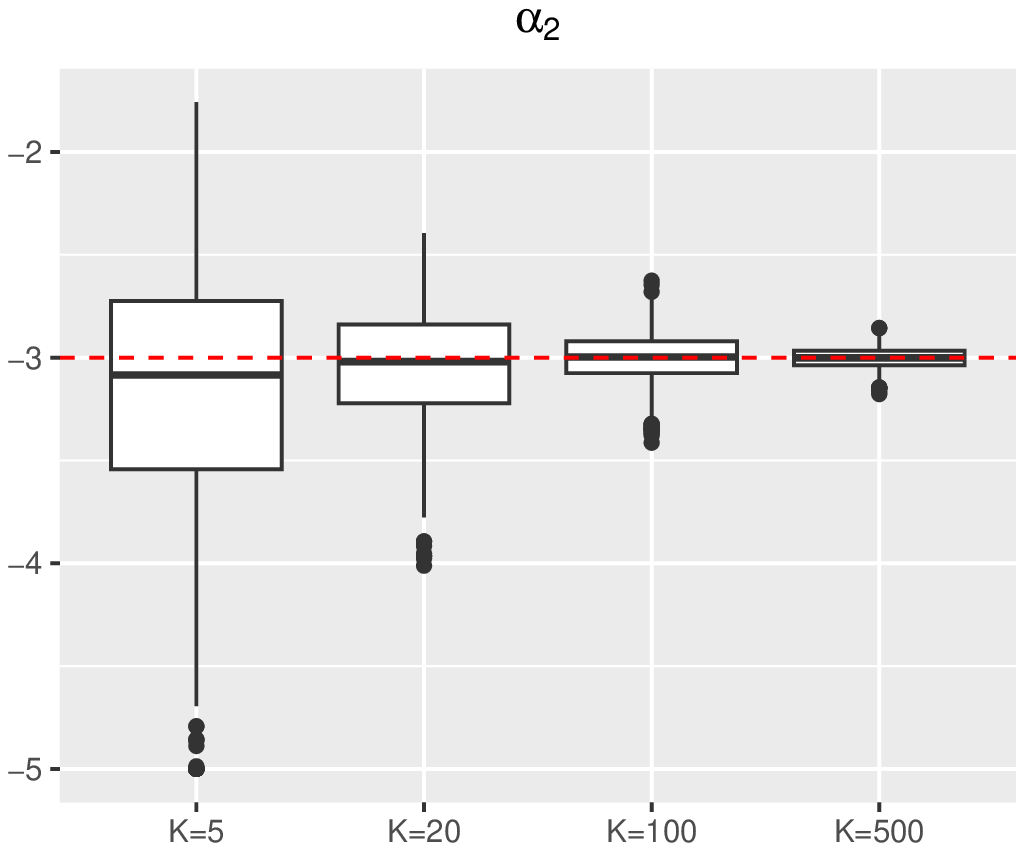}      \includegraphics[scale=0.5]{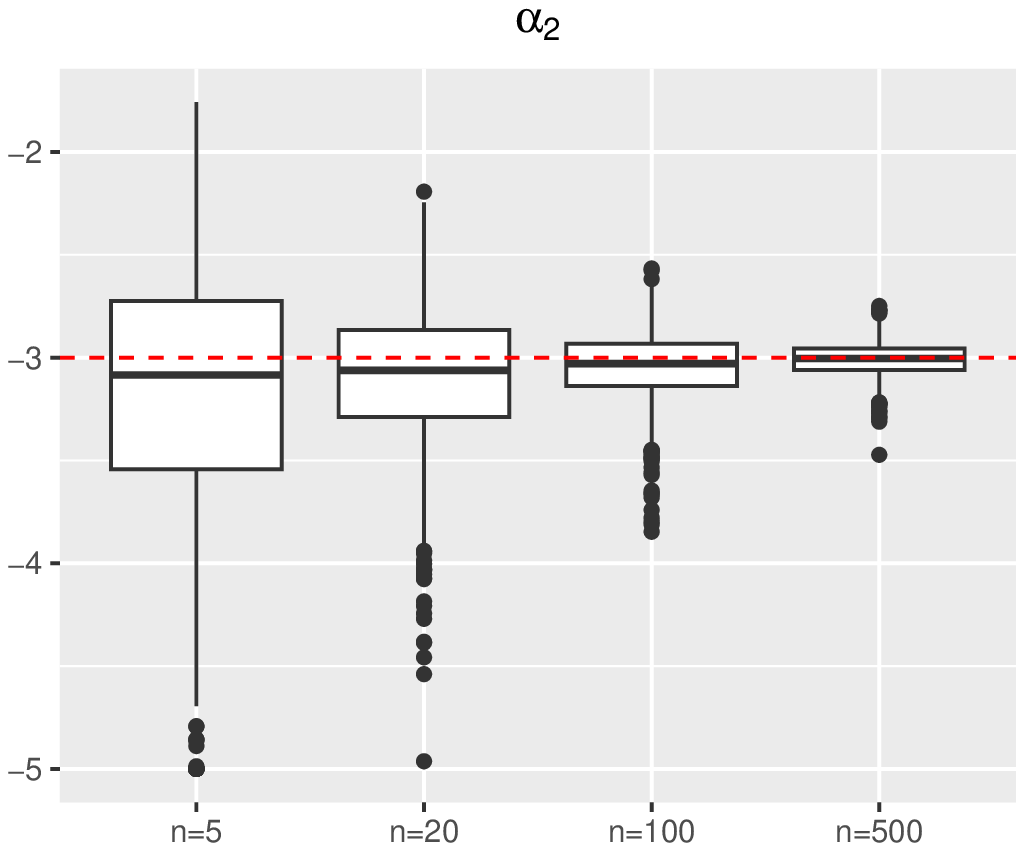}
\caption{Exp4: Boxplots of parameter estimates for scenario 1 (left) and scenario 2 (right), based on 1000 samples from Clayton bivariate copulas with parameters $2e^{\phi(\bbbeta,\bX_{ki})}$, where $\phi(\bbbeta,\bX_{ki})=1 - 1.5 U_{ki}$, and Poisson margins with rates $e^{h(\balpha, \bX_{ki})}$, where $h(\balpha, \bX_{ki}) = 2 - 3Z_{ki}$, $k\in\setK$, $i\in \{1,\ldots, n_k\}$.}
    \label{fig:exp4}
\end{figure}
\begin{figure}[ht!]
    \centering
    \includegraphics[scale=0.5]{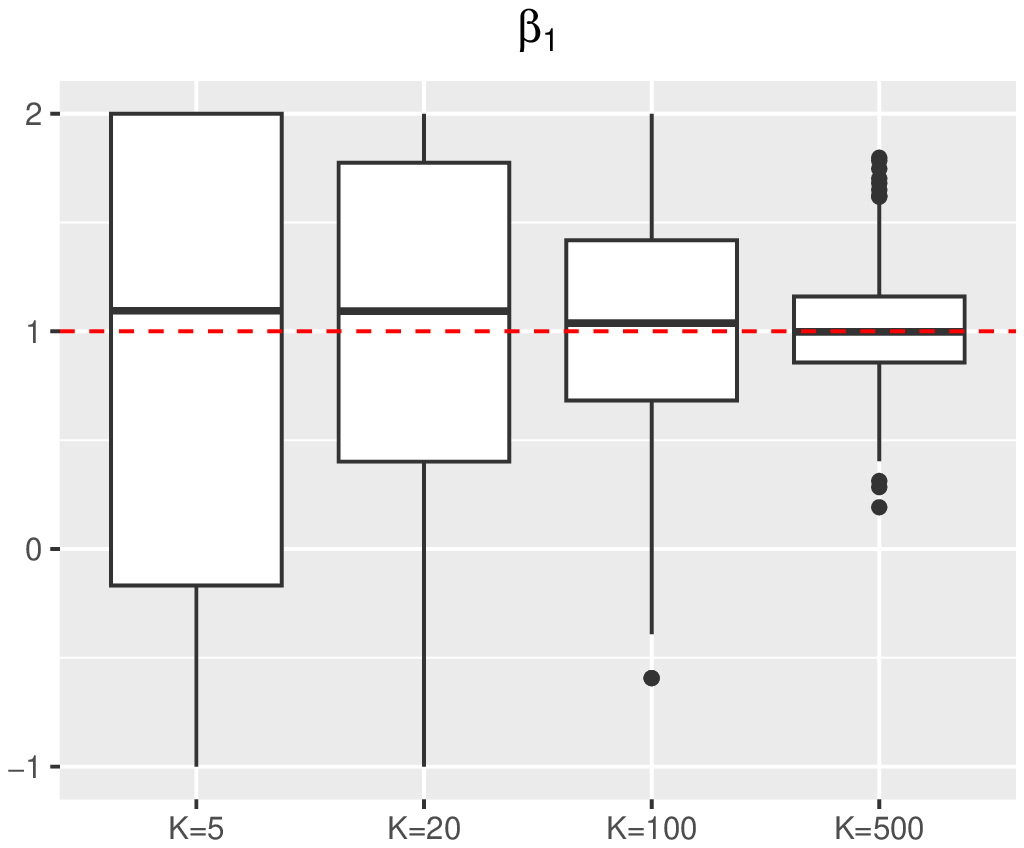}     \includegraphics[scale=0.5]{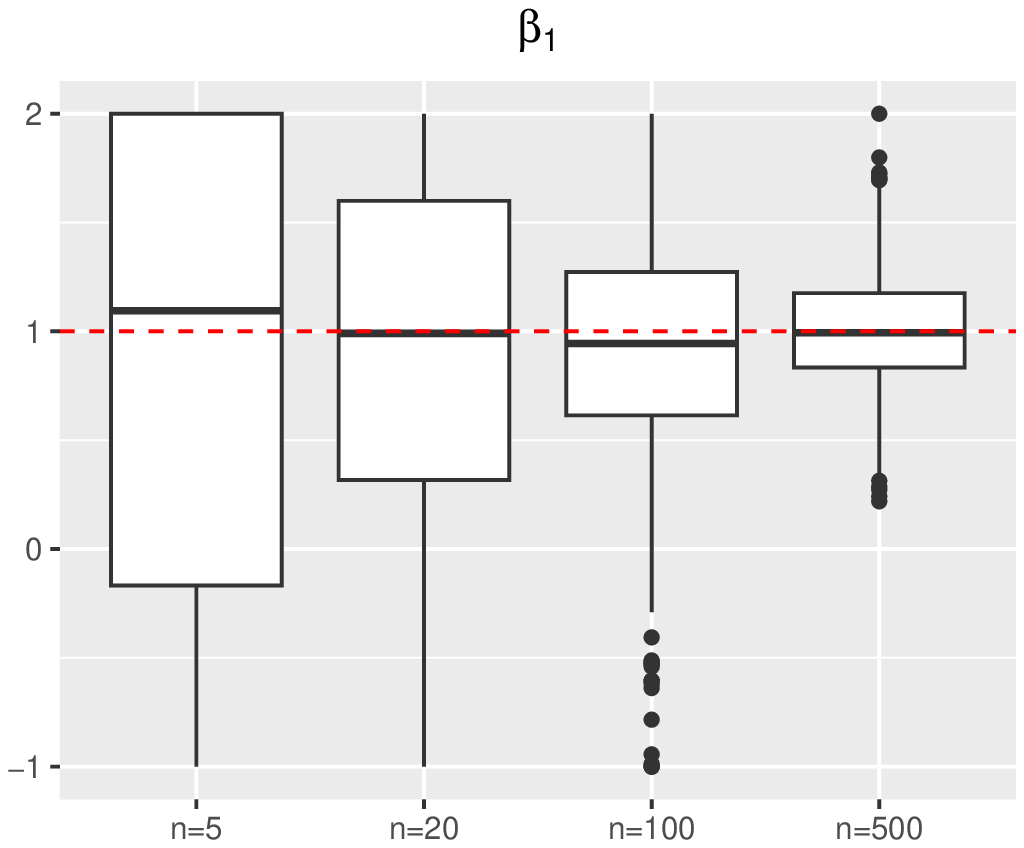}
    \includegraphics[scale=0.5]{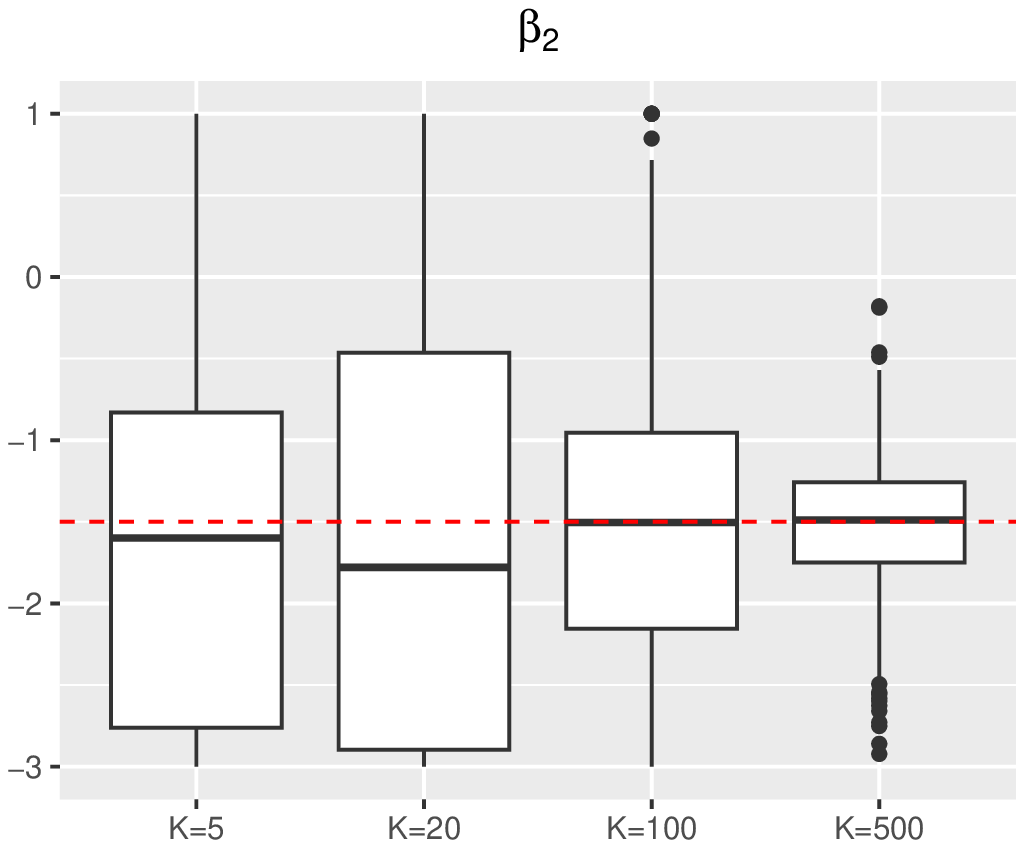}     \includegraphics[scale=0.5]{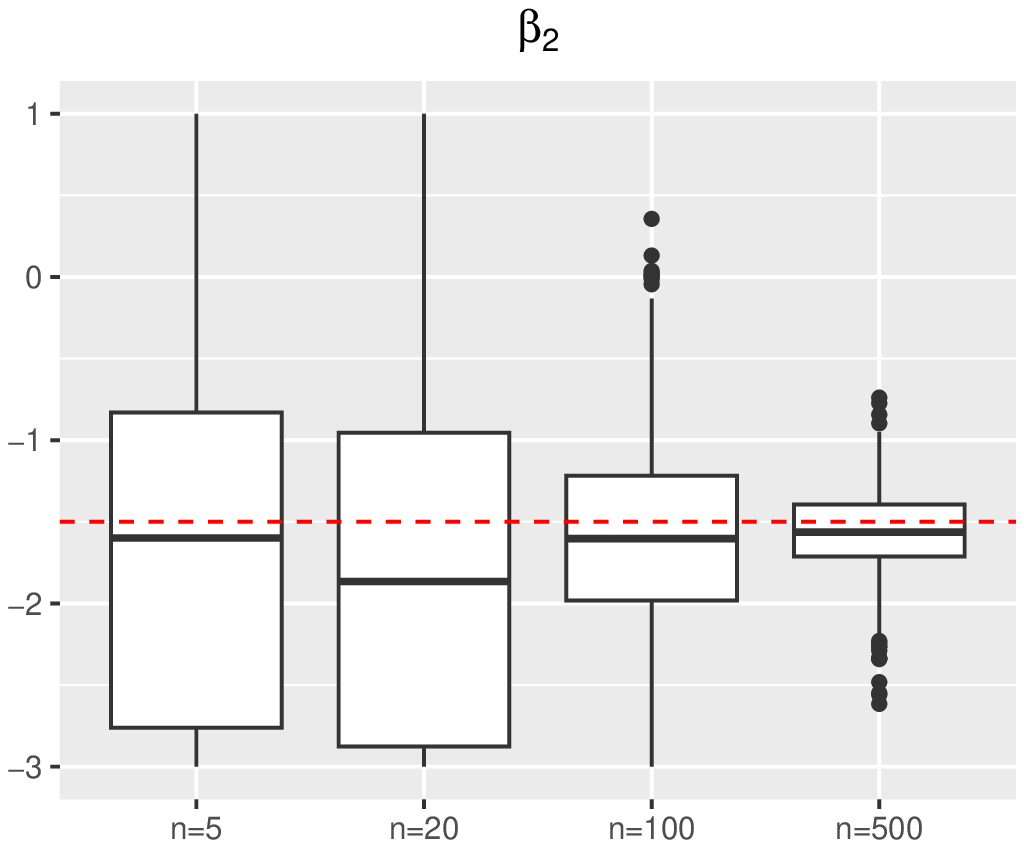}
    \includegraphics[scale=0.5]{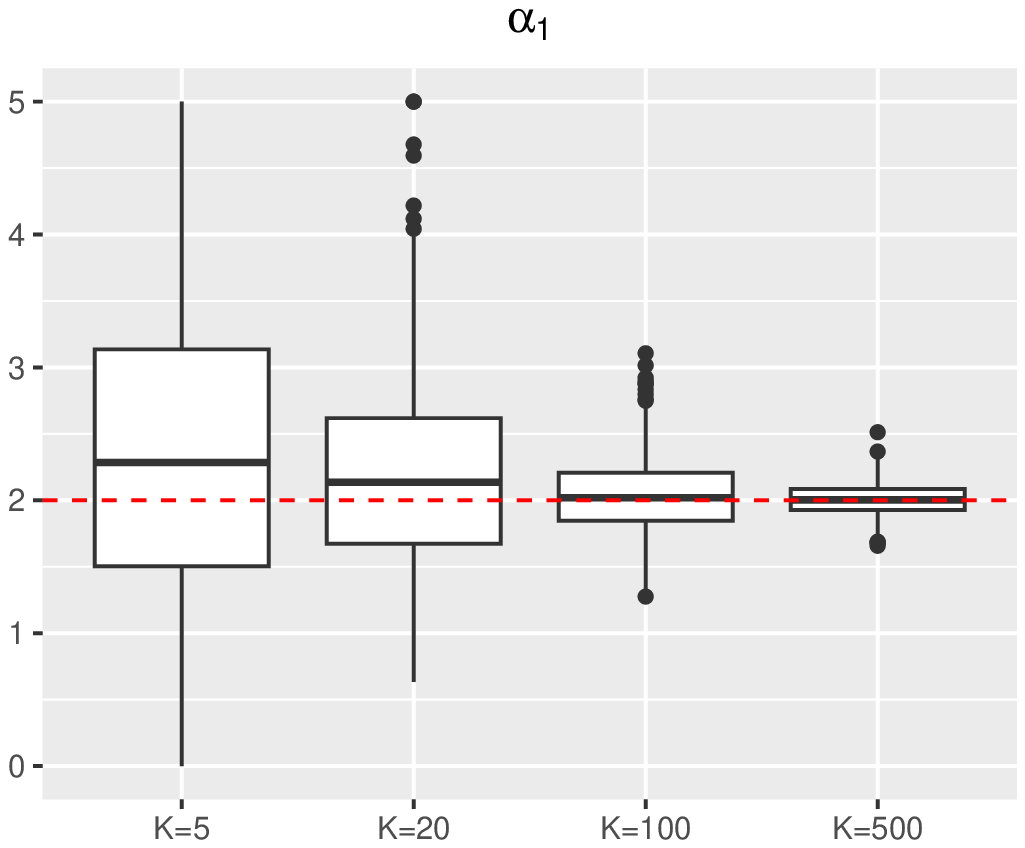}     \includegraphics[scale=0.5]{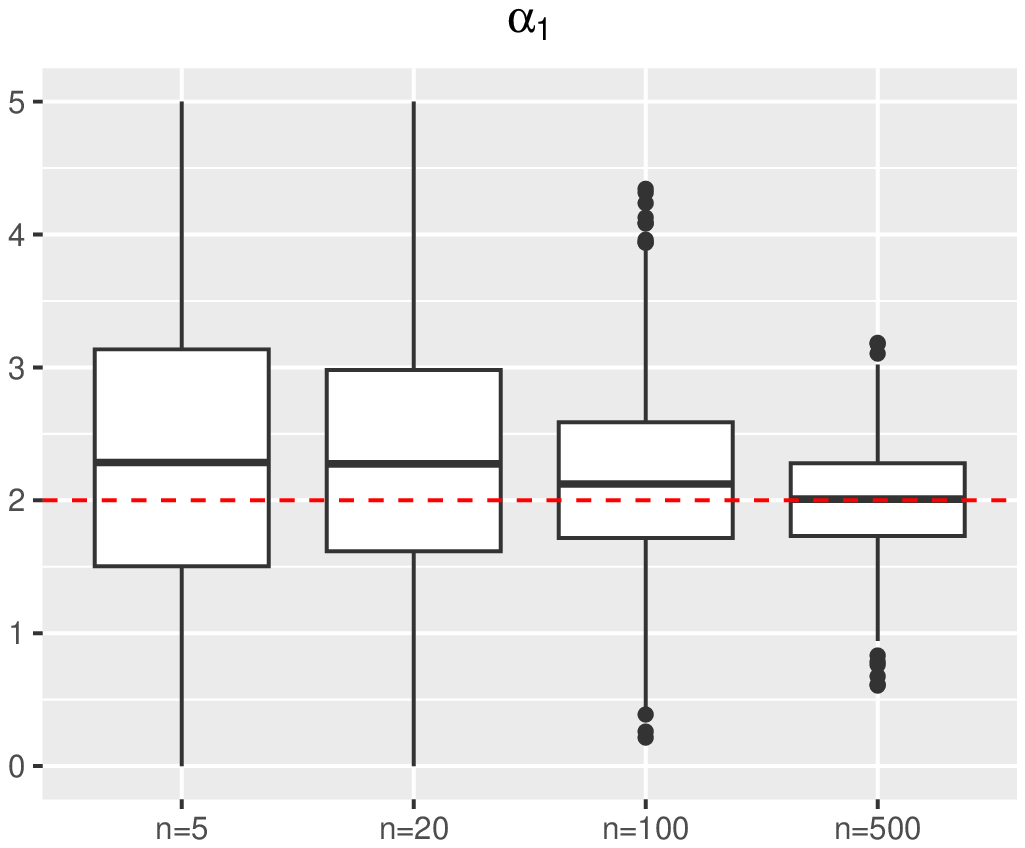}
    \includegraphics[scale=0.5]{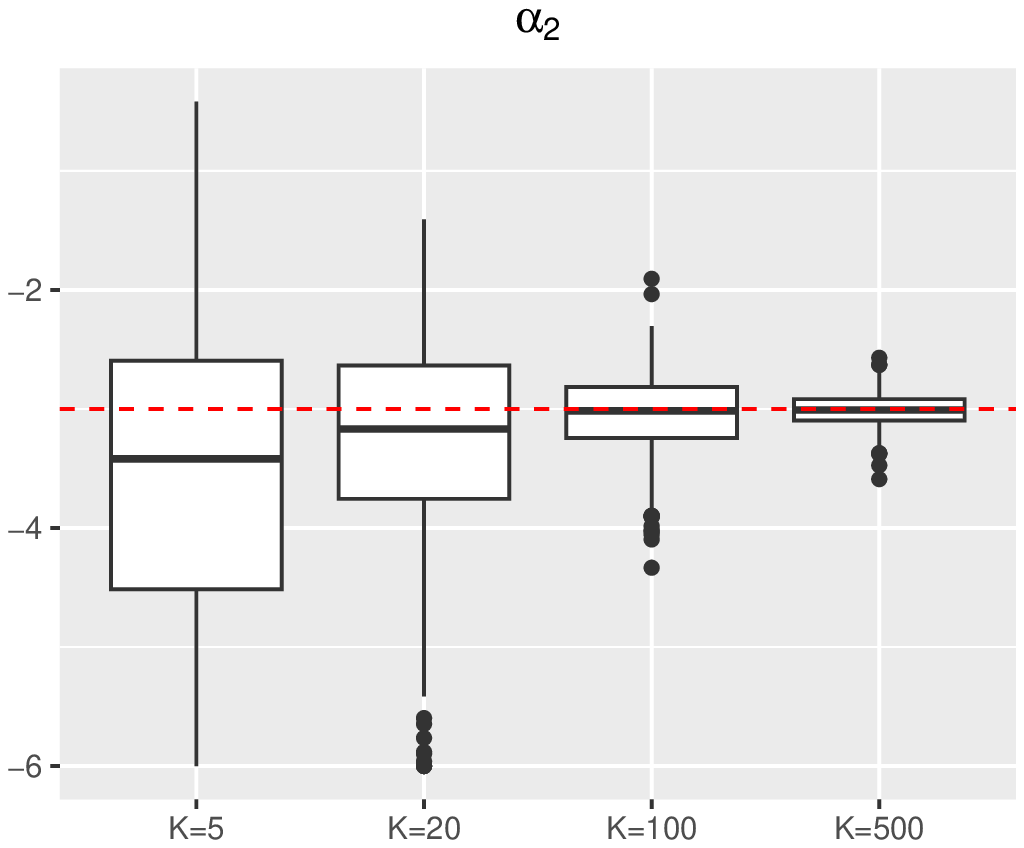}     \includegraphics[scale=0.5]{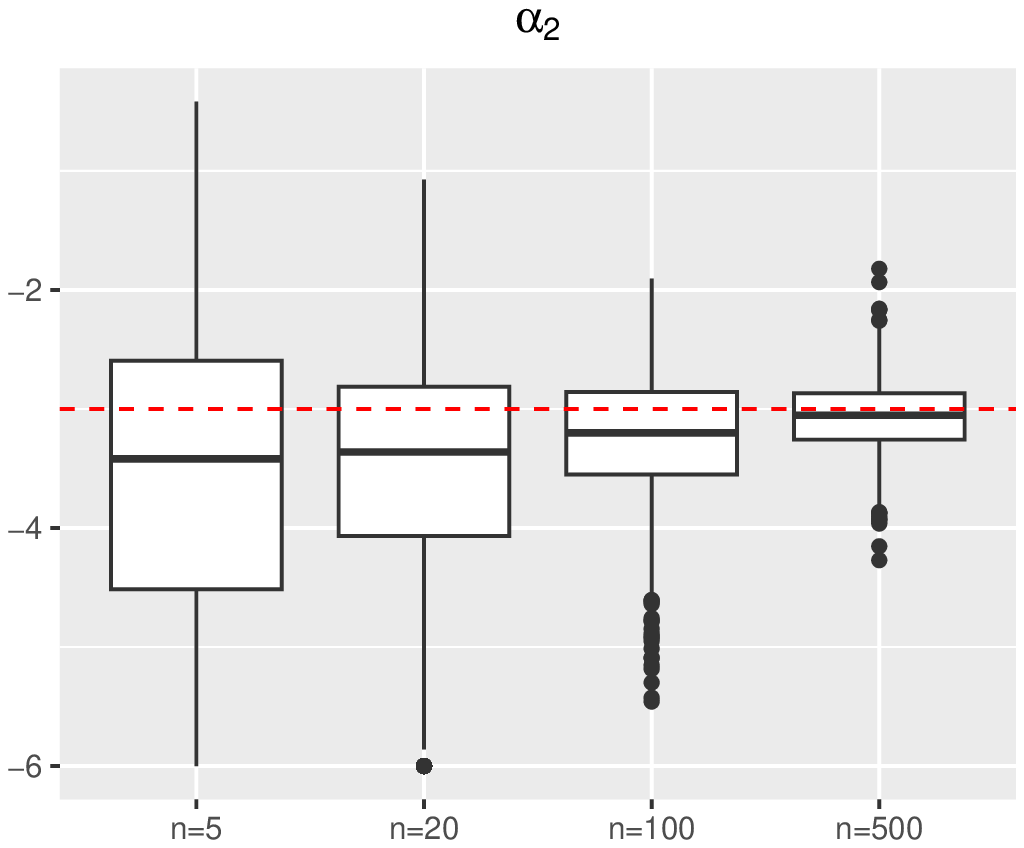}

    \caption{Exp5: Boxplots of parameter estimates for scenario 1 (left) and scenario 2 (right), based on 1000 samples from Clayton bivariate copulas with parameters $2e^{\phi(\bbbeta,\bX_{ki})}$, where $\phi(\bbbeta,\bX_{ki})=1 - 1.5 U_{ki}$,  and Bernoulli margins with parameters $
    \left\{1+e^{-h(\balpha,\bX_{ki})}\right\}^{-1}$, where $h(\balpha,\bX_{ki})=2-3Z_{ki}$, $k\in\setK$, $i\in \{1,\ldots, n_k\}$.}
    \label{fig:exp5}
\end{figure}
For the fourth and fifth numerical experiments, we used the same copula families as in the second experiment, but we considered discrete margins: Poisson margins with rate parameters $e^{2-3Z_{ki}}$, and Bernoulli margins with parameters
$\left\{1+e^{-(2-3Z_{ki})}\right\}^{-1}$, where  $Z_{ki}$ are iid $ {\rm Exp}(1)$, $k\in\setK$, $i\in \{1,\ldots, n_k\}$.
 The results of these experiments are displayed in  Figures \ref{fig:exp4}--\ref{fig:exp5}.
We also included three numerical experiments with Frank copula and four covariates. For the sixth experiment, we used Frank copula with parameters $\phi(\bbbeta,\bX_{ki}) = 2+ 8 Z_{1,ki} + 3 U_{2,ki}$ and normal margins with a variance of $1$ and means $h(\balpha,\bX_{ki})= 5+5Z_{3,ki} +3U_{4ki}$, where $Z_{1,ki}, Z_{3ki}$ are iid $N(0,1)$ and $U_{2,ki},U_{4,ki}$ are iid ${\rm U}(-1,1)$, $k\in\setK$, $i\in \{1,\ldots, n_k\}$.  The results are displayed in Figures \ref{fig:franknor-cop}-\ref{fig:franknor-mar}. For the seventh and eighth experiments, we used Frank copula with parameters $\phi(\bbbeta,\bX_{ki}) = 2+ 8 U_{1,ki} + 3 U_{2,ki}$, where $U_{1,ki}$ are iid $\unif$ and $U_{2,ki}$ are iid ${\rm U}(-1,1)$, $k\in\setK$, $i\in \{1,\ldots, n_k\}$. For the seventh experiment, Poisson margins were employed with rates $e^{3-U_{3,ki}-0.5U_{4,ki}}$, where $U_{3,ki}$ are iid $\unif$ and $U_{4,ki}$ are iid ${\rm U}(-1,1)$, $k\in\setK$, $i\in \{1,\ldots, n_k\}$, while for the eighth experiment, we used Bernoulli margins with parameters $\left\{1+e^{-(1.5-2U_{3,ki}-0.5U_{4,ki})}\right\}^{-1}$, where  $U_{3,ki},U_{4,ki} $ are iid $\unif$, $k\in\setK$, $i\in \{1,\ldots, n_k\}$. The results are displayed in Figures \ref{fig:frankpoi-cop}-\ref{fig:frankber-mar}.

\begin{figure}[ht!]
    \centering
 \includegraphics[scale=0.3]{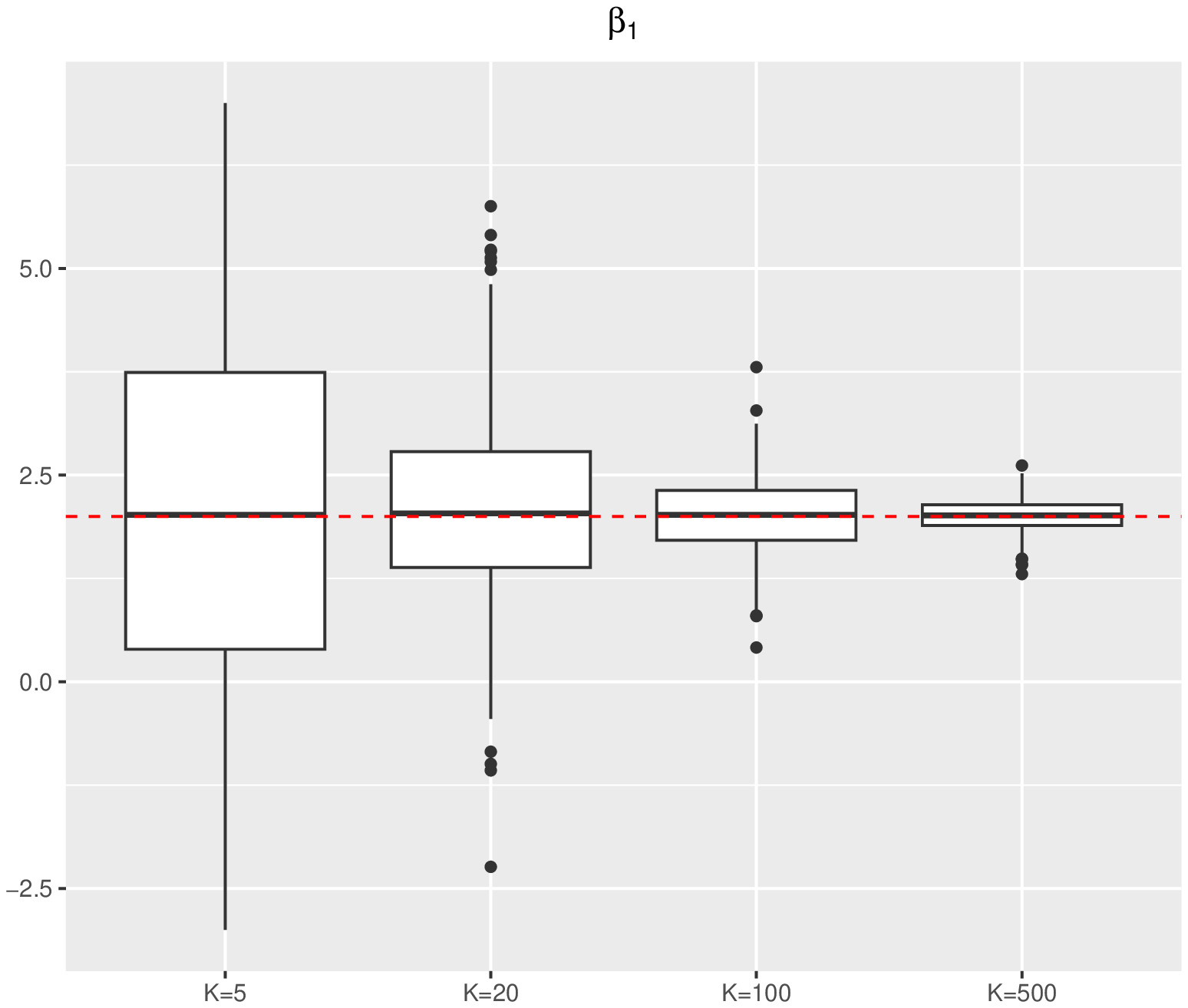}    \includegraphics[scale=0.3]{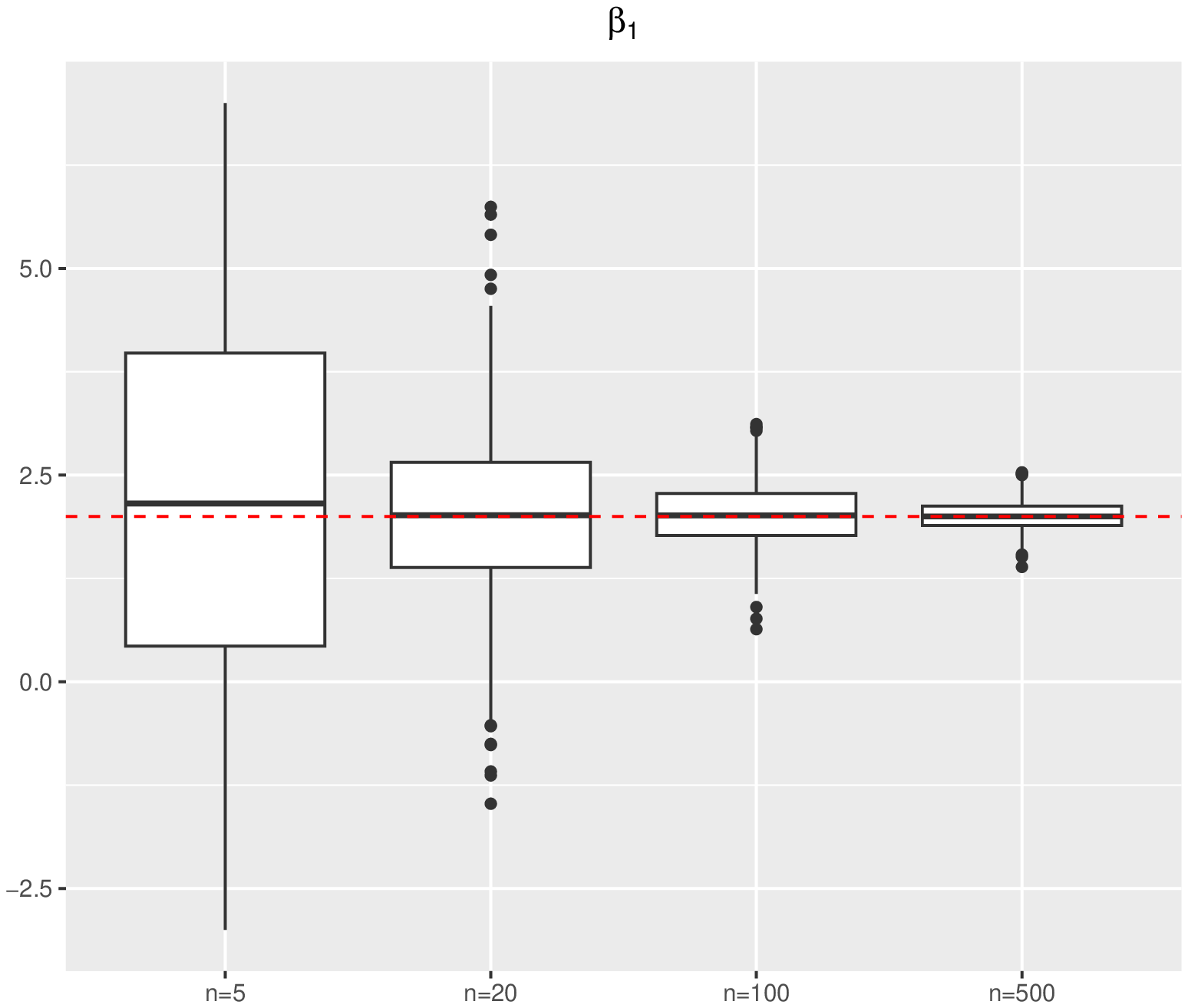}
 \includegraphics[scale=0.3]{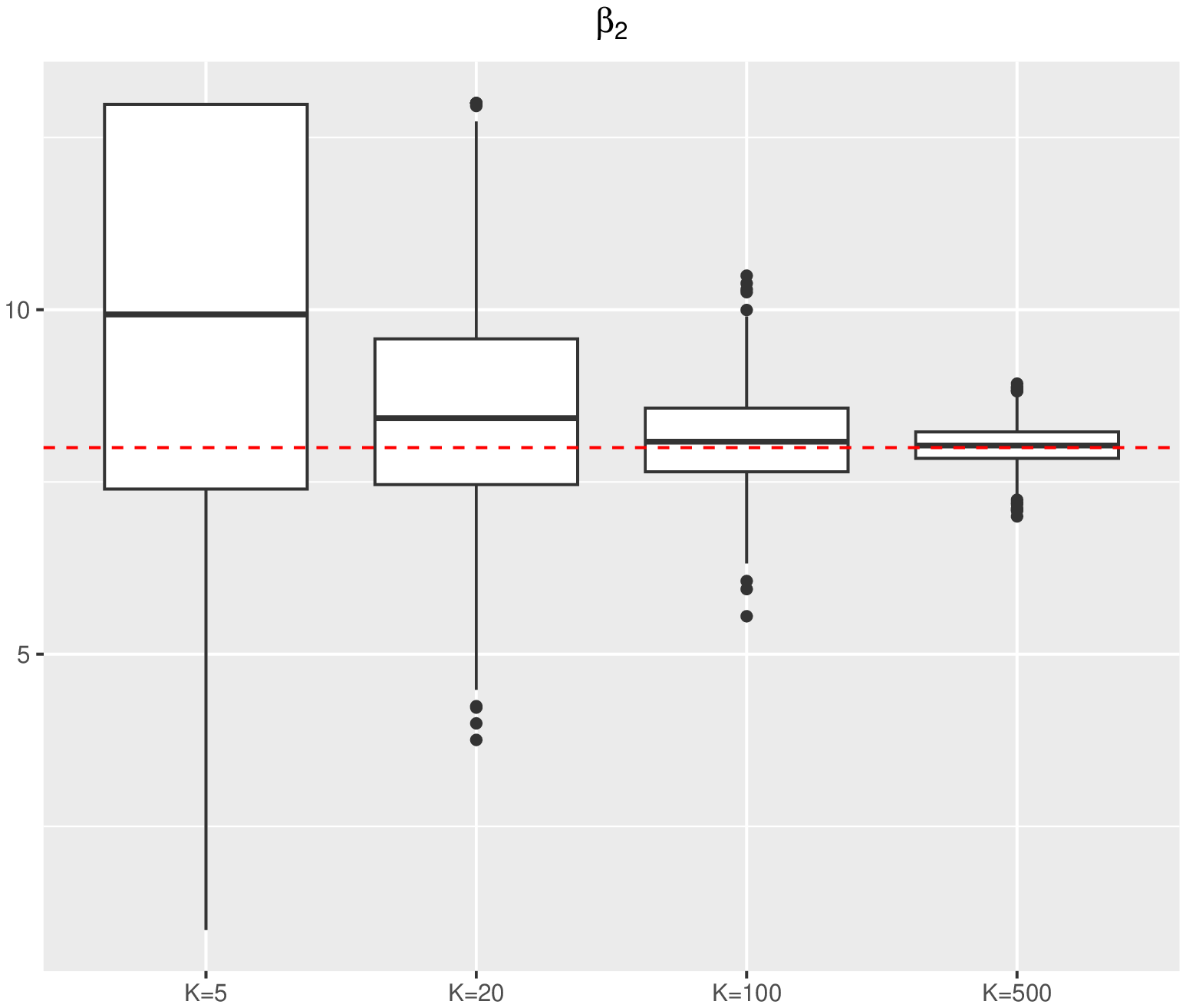}    \includegraphics[scale=0.3]{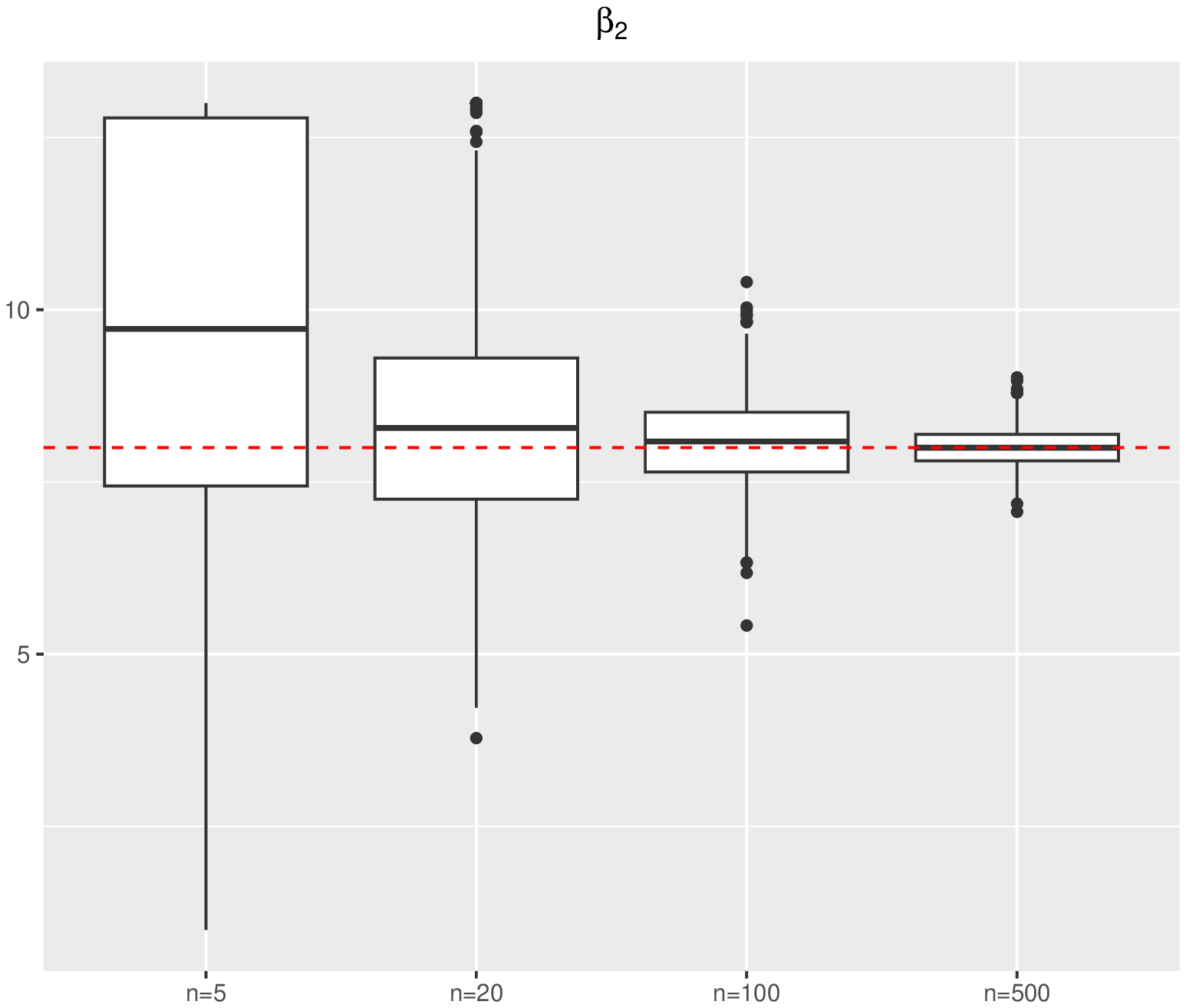}
 \includegraphics[scale=0.3]{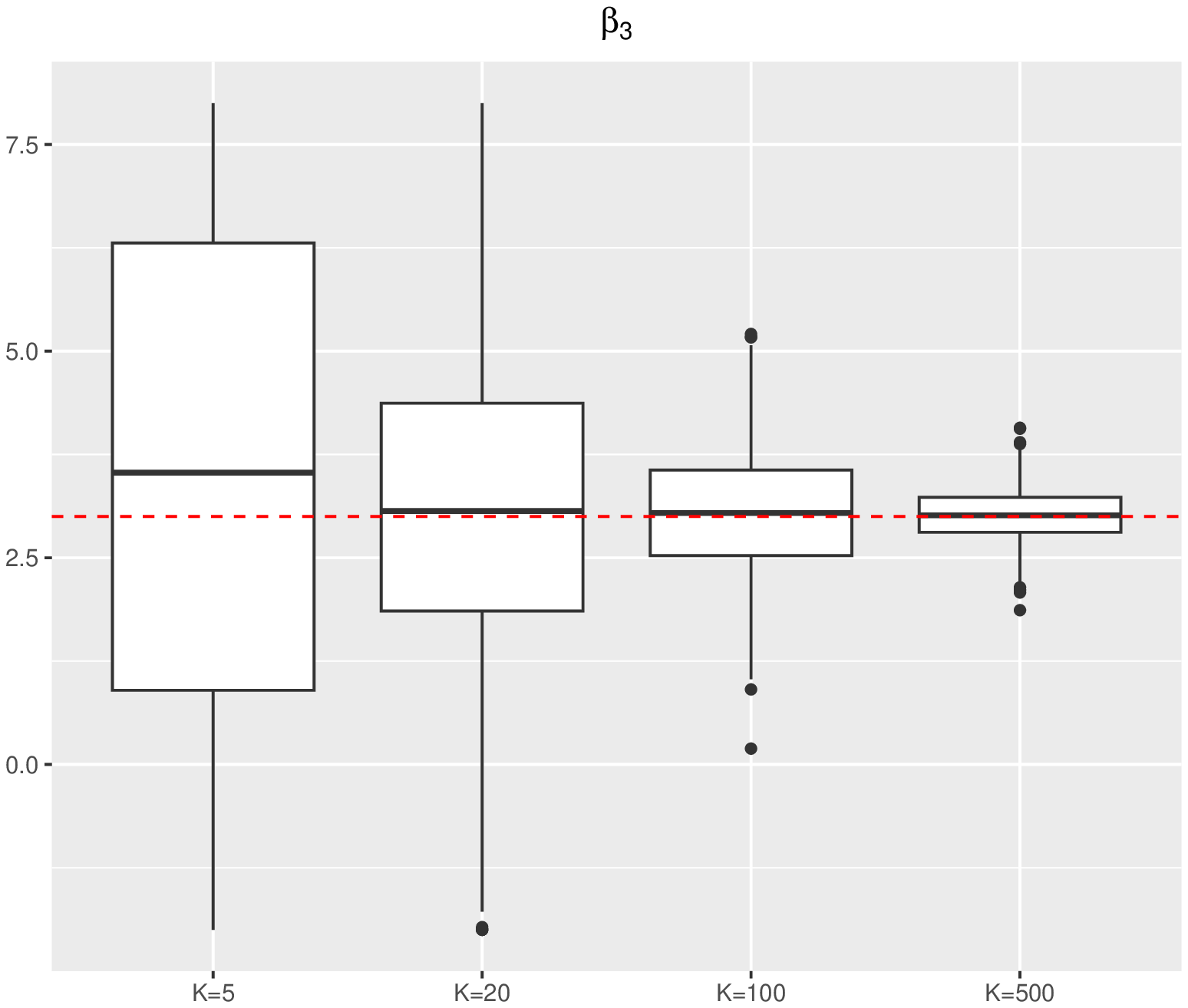}    \includegraphics[scale=0.3]{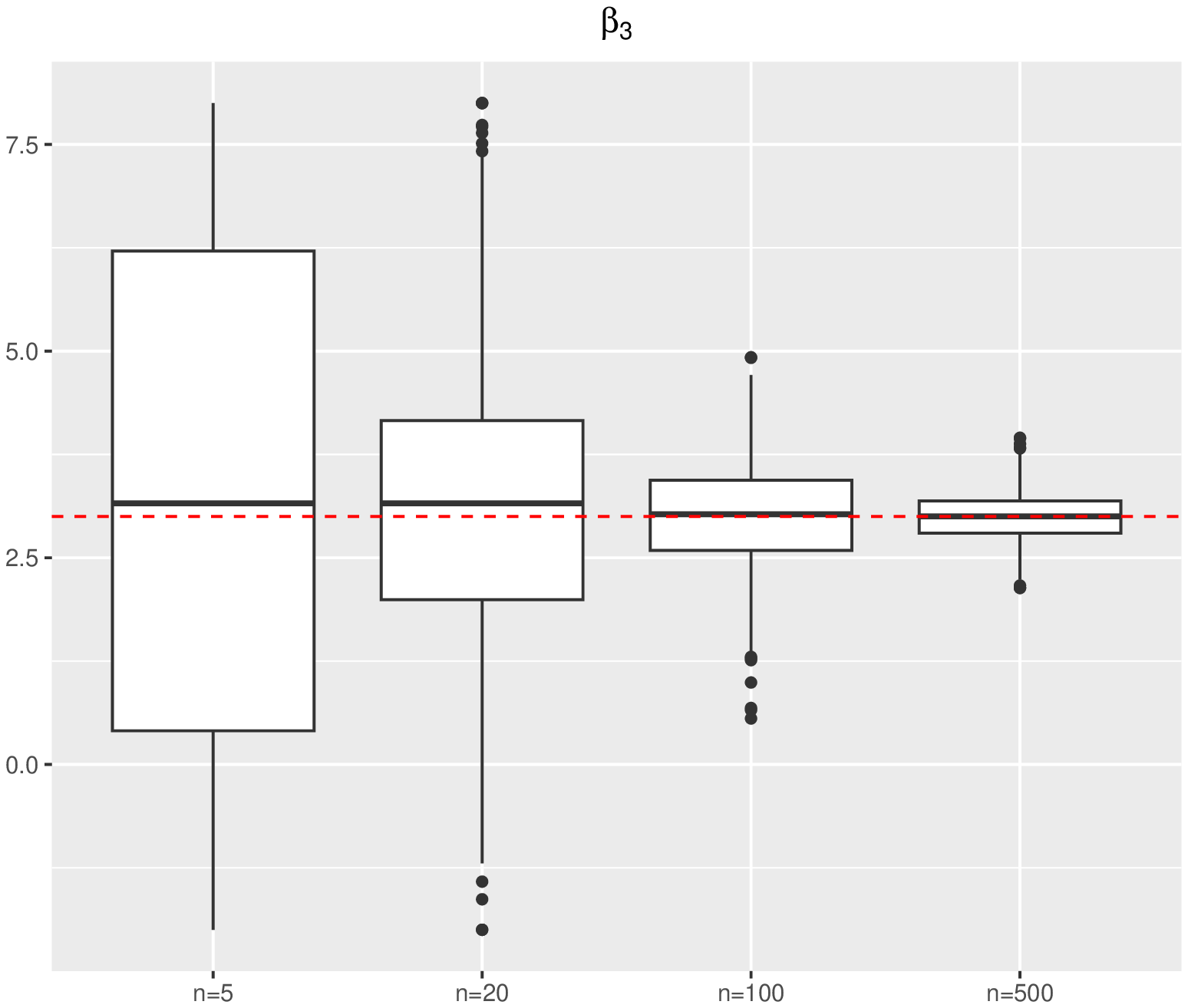}
     \caption{Exp6a: Boxplots of the estimation of $(\beta_1,\beta_2,\beta_3)$ for scenario 1 (left) and scenario 2 (right), based on 1000 samples from Frank bivariate copulas with parameters  $\phi(\bbbeta,\bX_{ki}) = 2+ 8Z_{1ki} + 3U_{2ki}$ and Gaussian margins with means $h(\balpha,\bX_{ki})= 5+5Z_{2ki}+3U_{4ki}$ and variance $\alpha_4=1$, $k\in\setK$, $i\in \{1,\ldots, n_k\}$.}
    \label{fig:franknor-cop}
\end{figure}
\begin{figure}[ht!]
    \centering
  \includegraphics[scale=0.3]{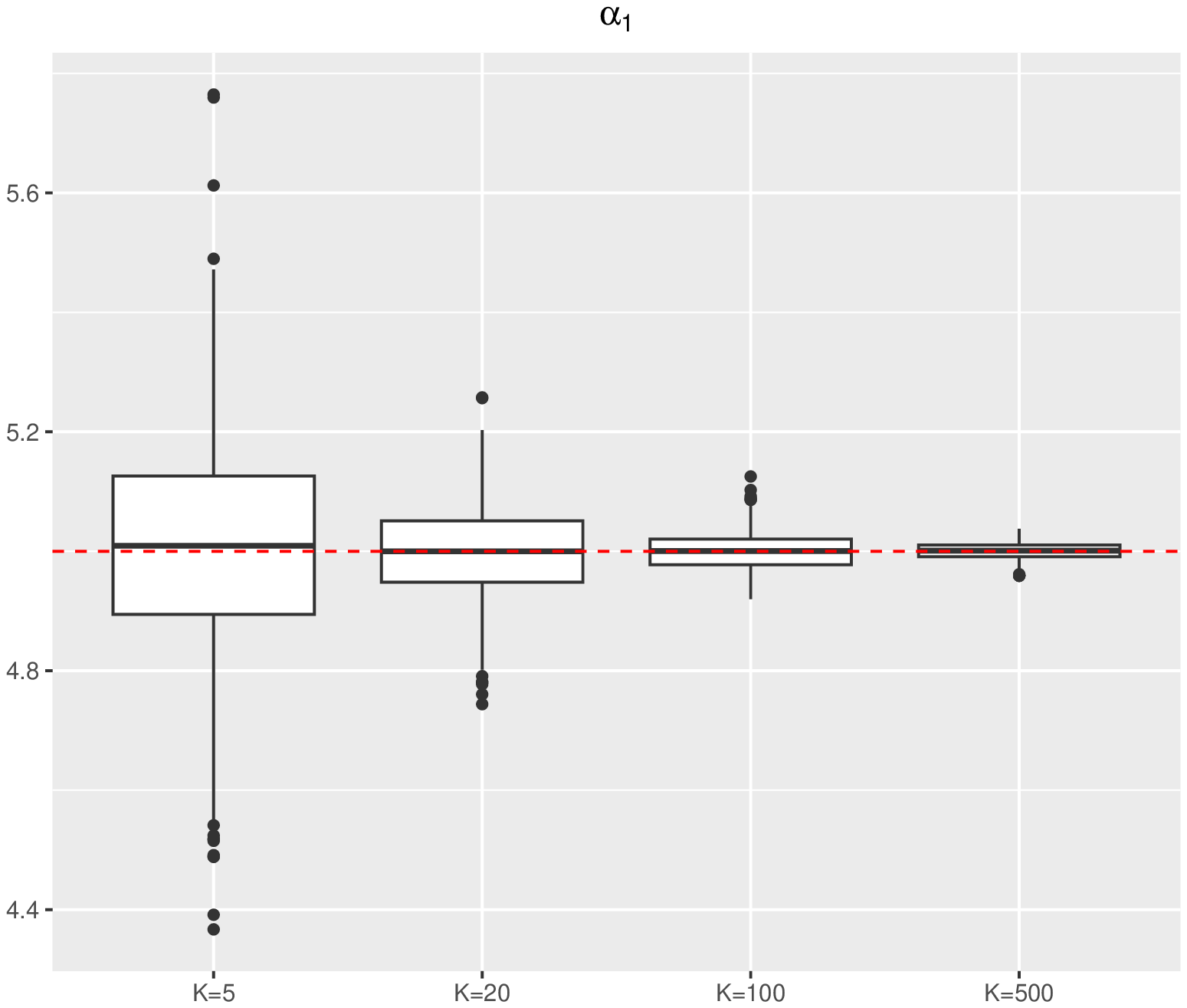}    \includegraphics[scale=0.3]{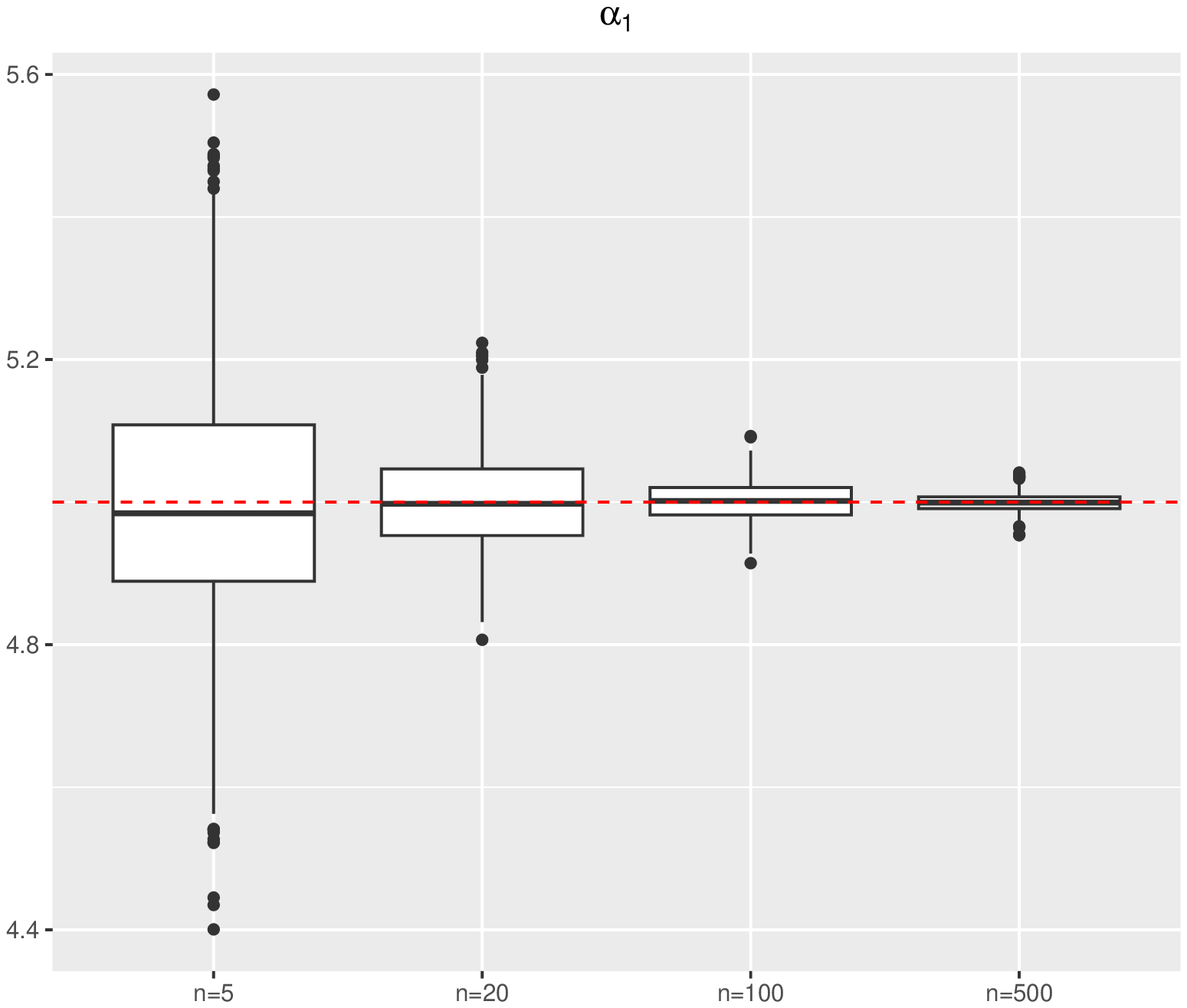}
 \includegraphics[scale=0.3]{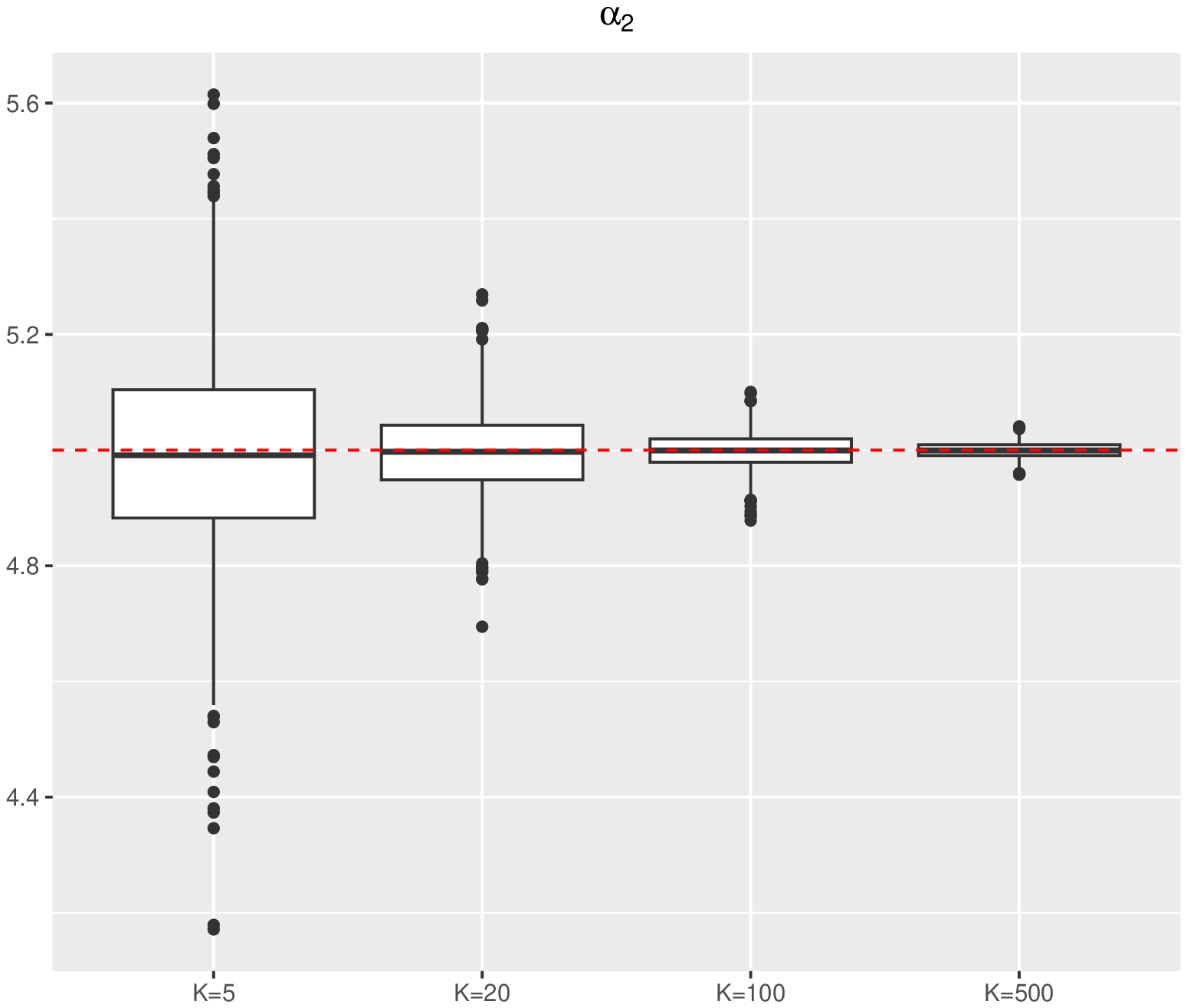}    \includegraphics[scale=0.3]{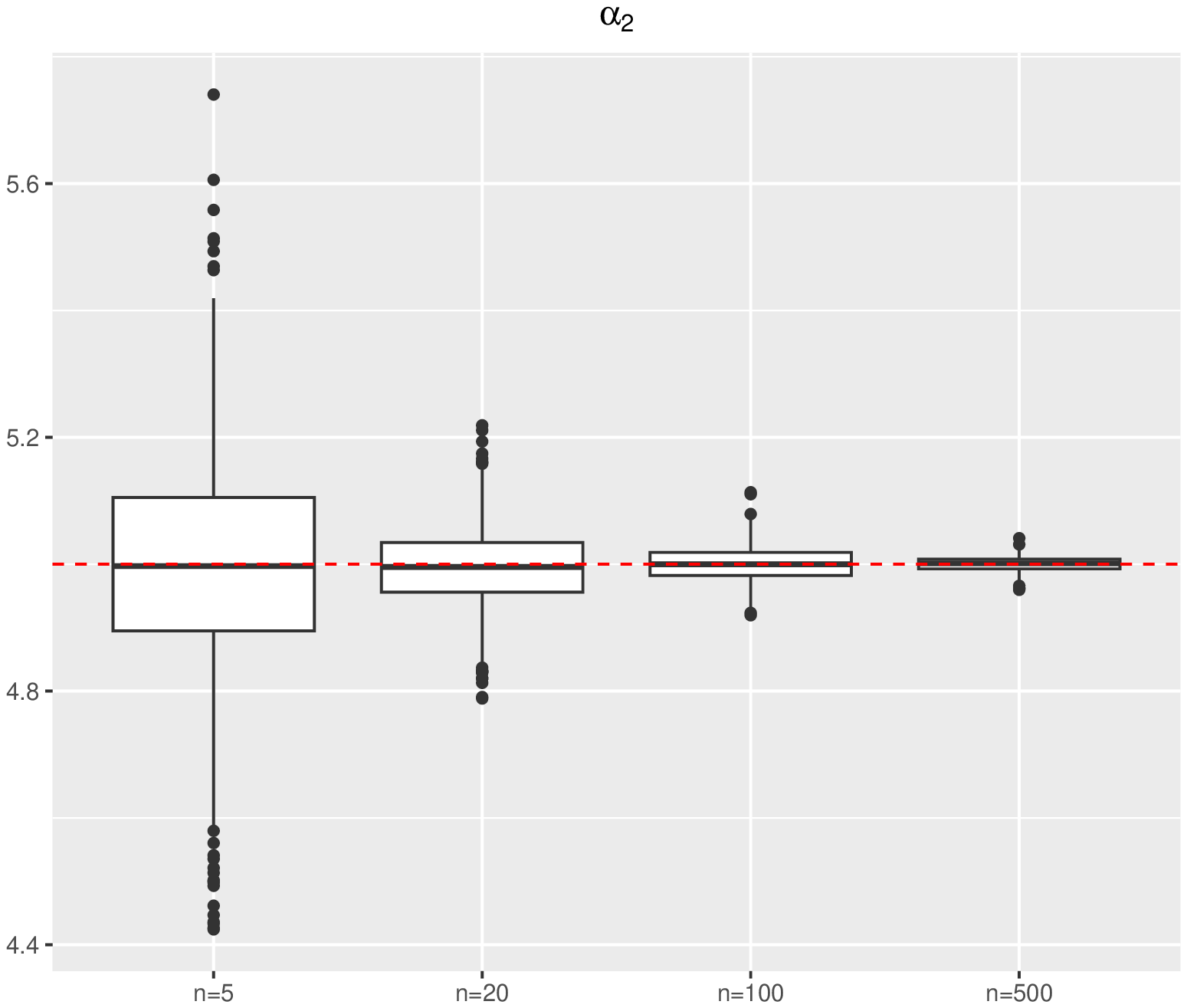}
 \includegraphics[scale=0.3]{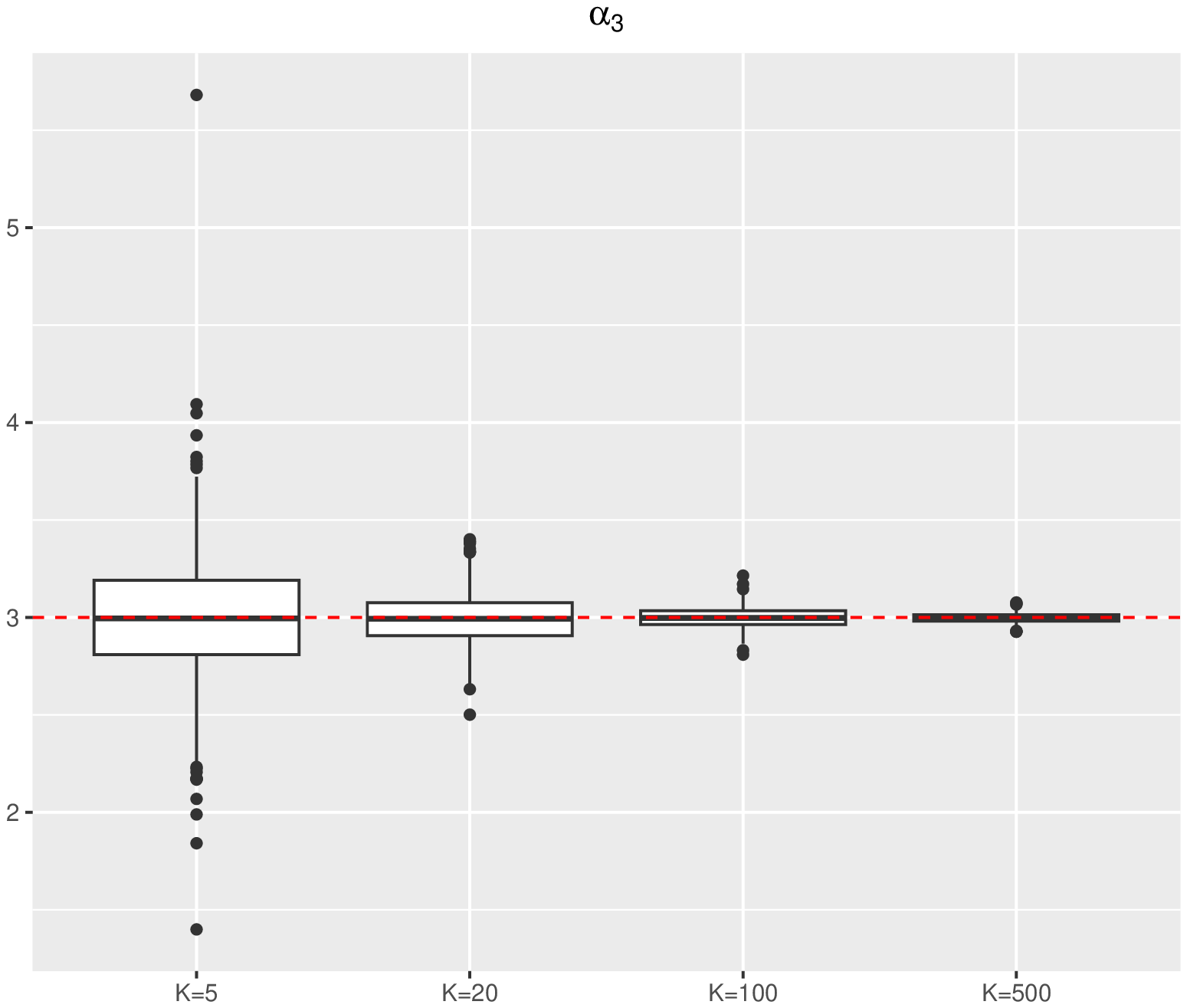}    \includegraphics[scale=0.3]{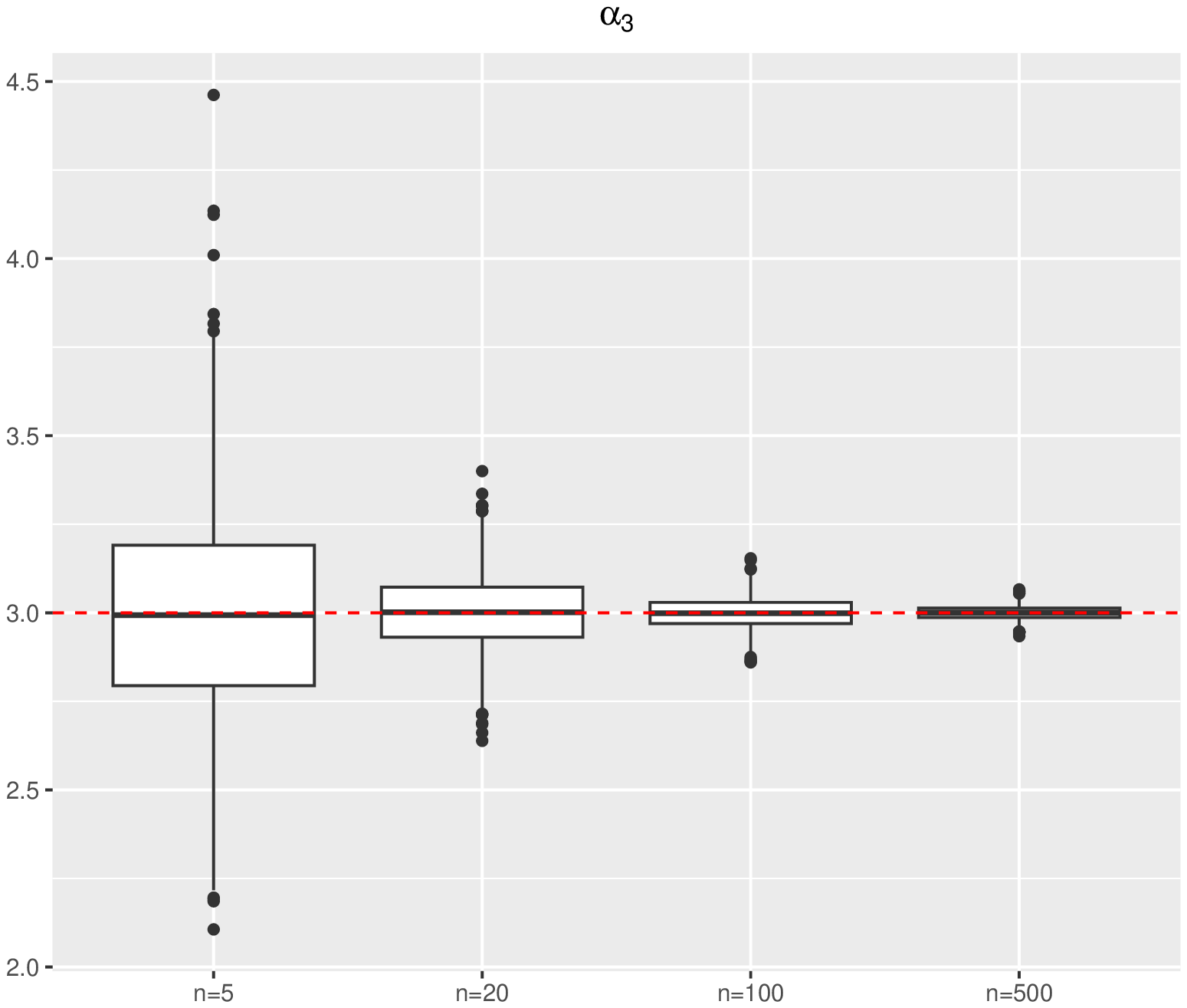}
 \includegraphics[scale=0.3]{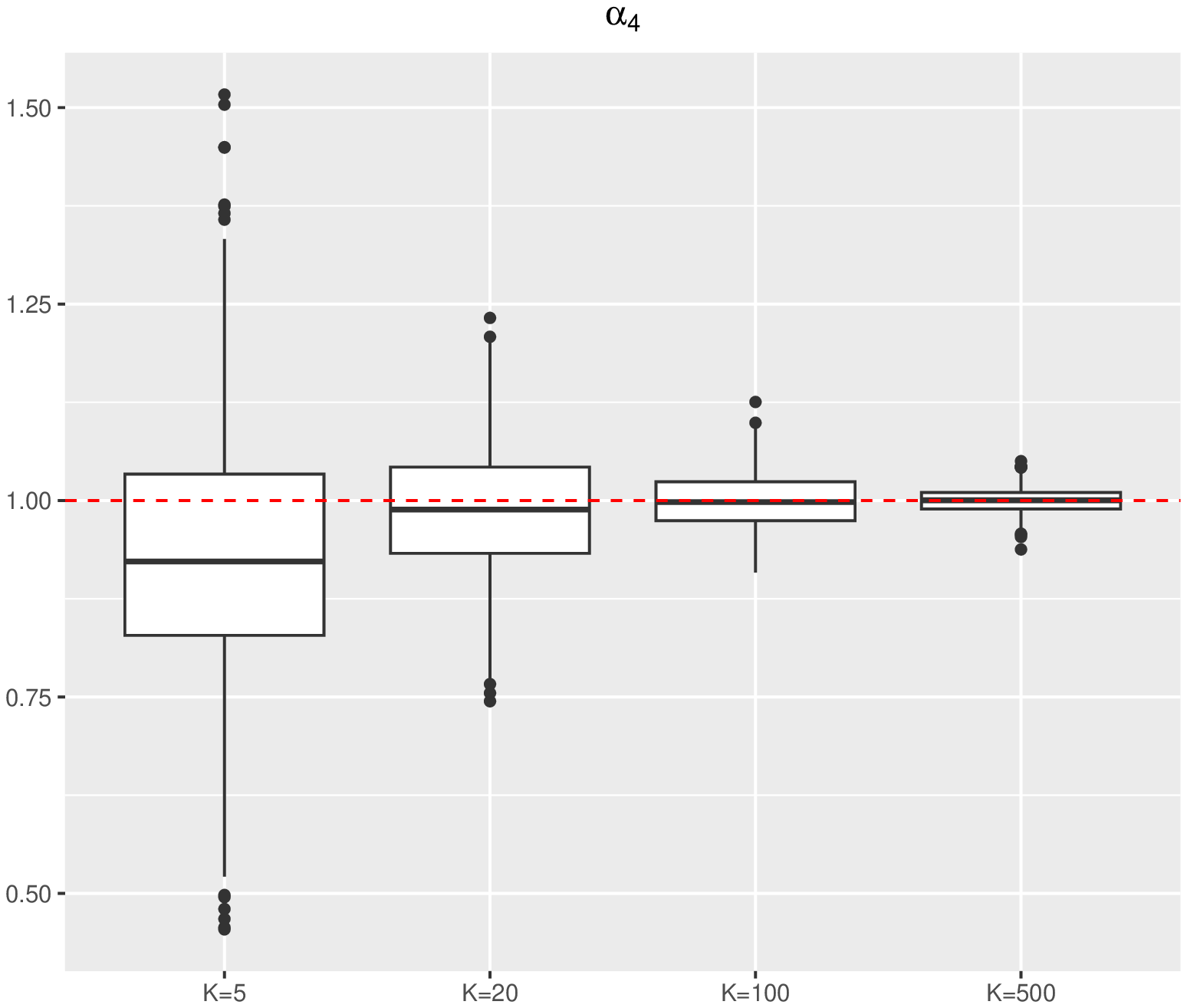}     \includegraphics[scale=0.3]{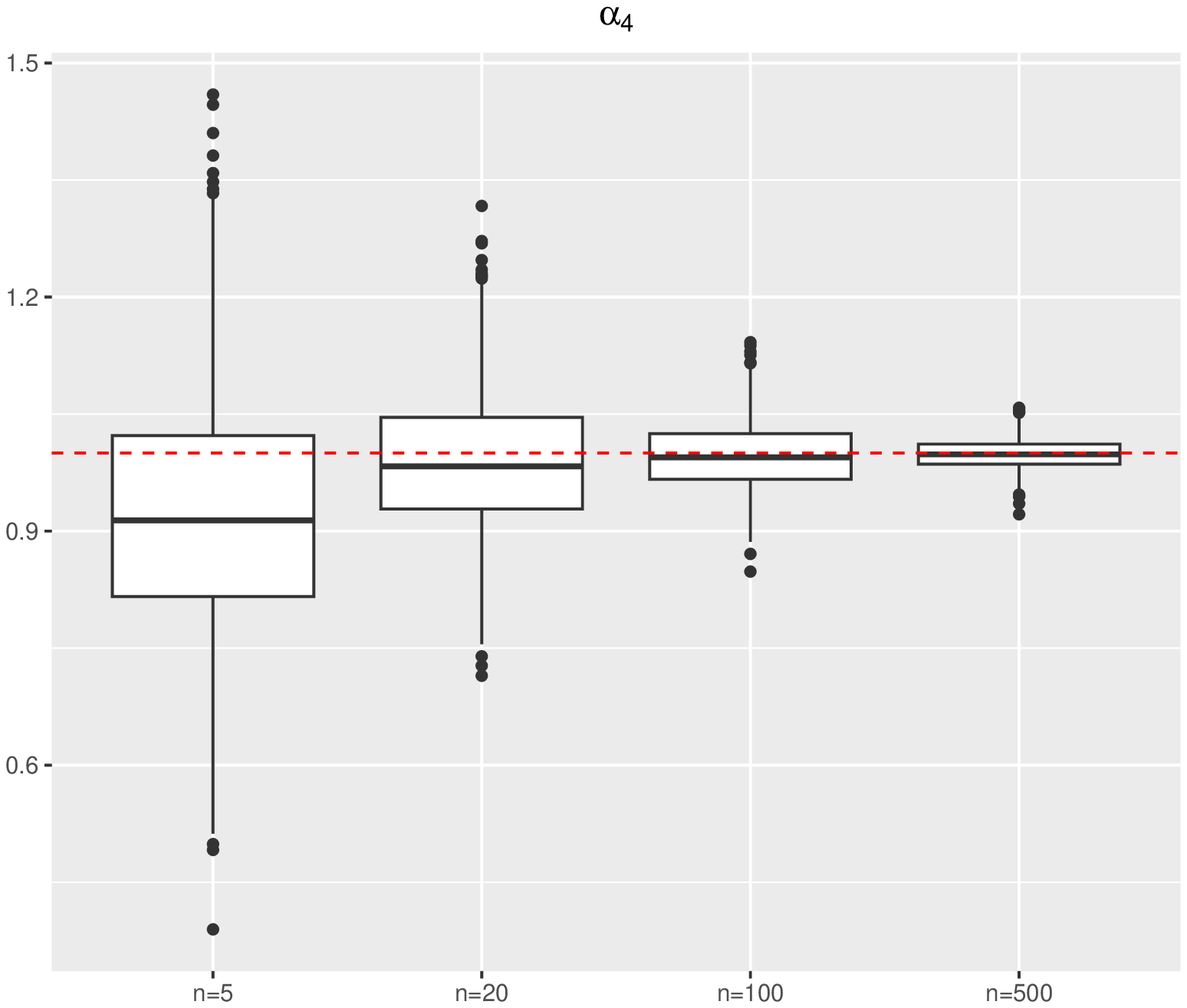}
    \caption{Exp6b: Boxplots of the estimation of $(\alpha_{1},\alpha_{2},\alpha_{3},\alpha_4)$ for scenario 1 (left) and scenario 2 (right), based on 1000 samples from Frank bivariate copulas with parameters $\phi(\bbbeta,\bX_{ki}) = 2+ 8Z_{1ki} + 3U_{2ki}$ and Gaussian margins with means $h(\balpha,\bX_{ki})= 5+5Z_{3ki}+3U_{4ki}$ and variance $\alpha_4=1$, $k\in\setK$, $i\in \{1,\ldots, n_k\}$.}
    \label{fig:franknor-mar}
\end{figure}

\begin{figure}[ht!]
    \centering
 \includegraphics[scale=0.5]{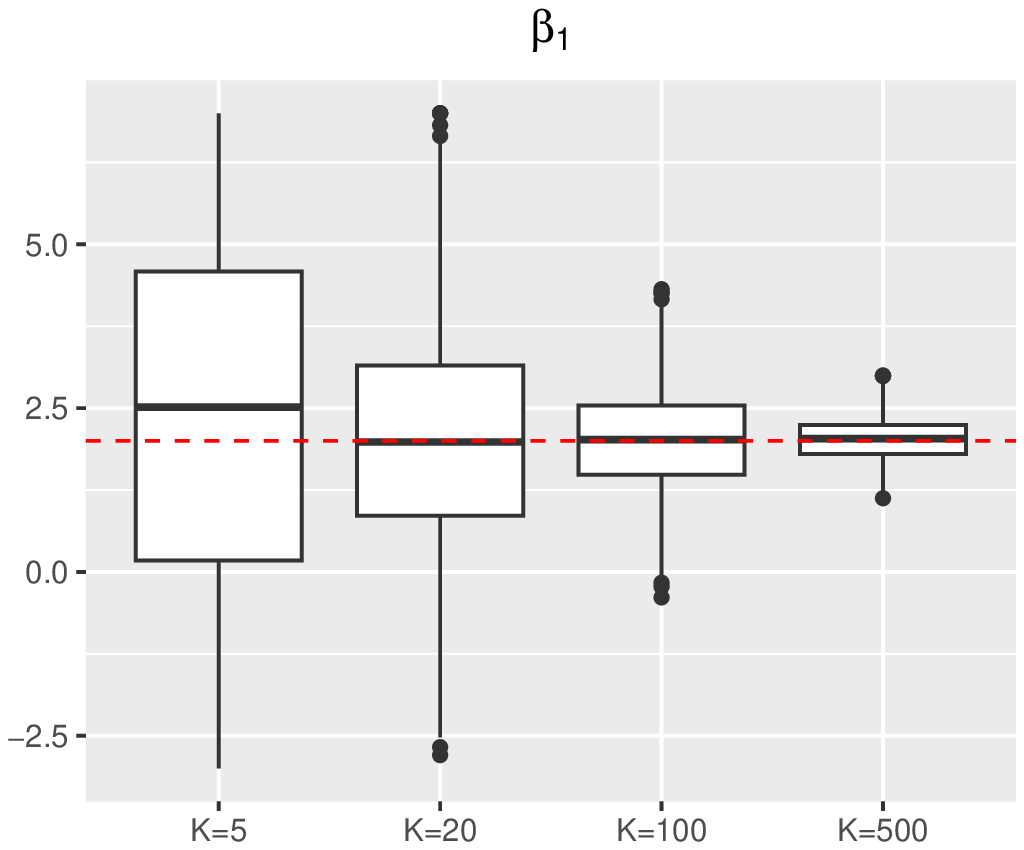}    \includegraphics[scale=0.5]{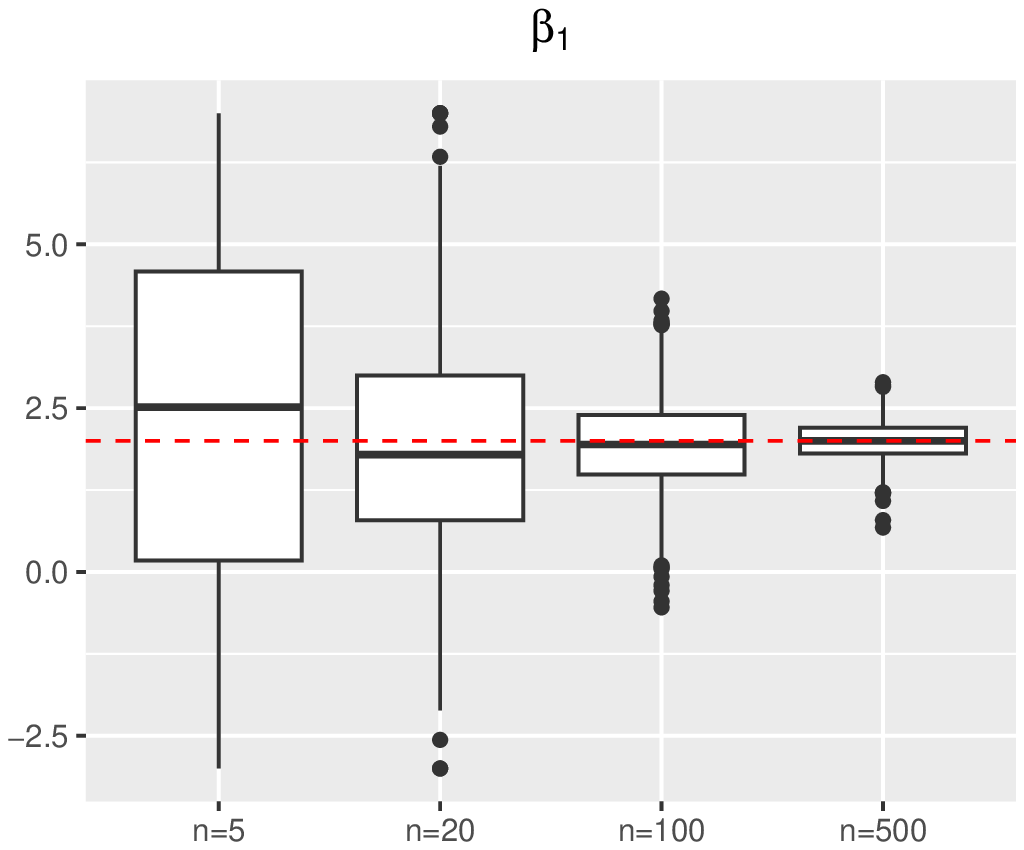}
 \includegraphics[scale=0.5]{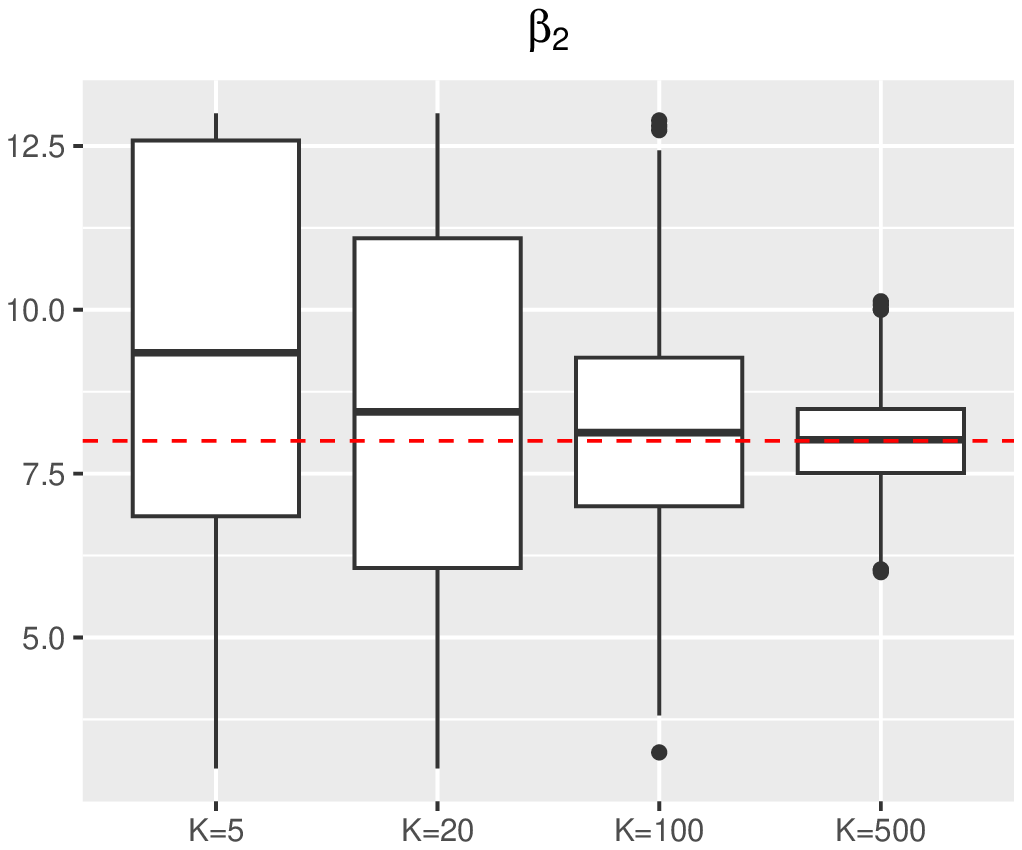}    \includegraphics[scale=0.5]{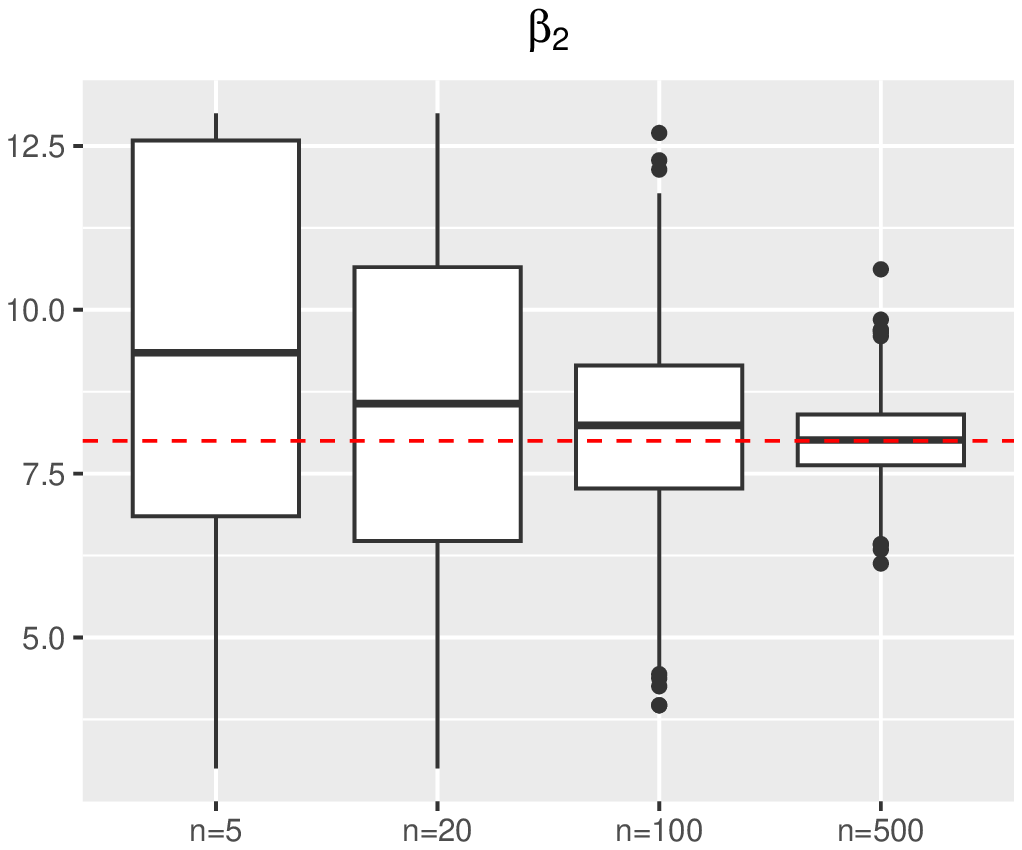}
 \includegraphics[scale=0.5]{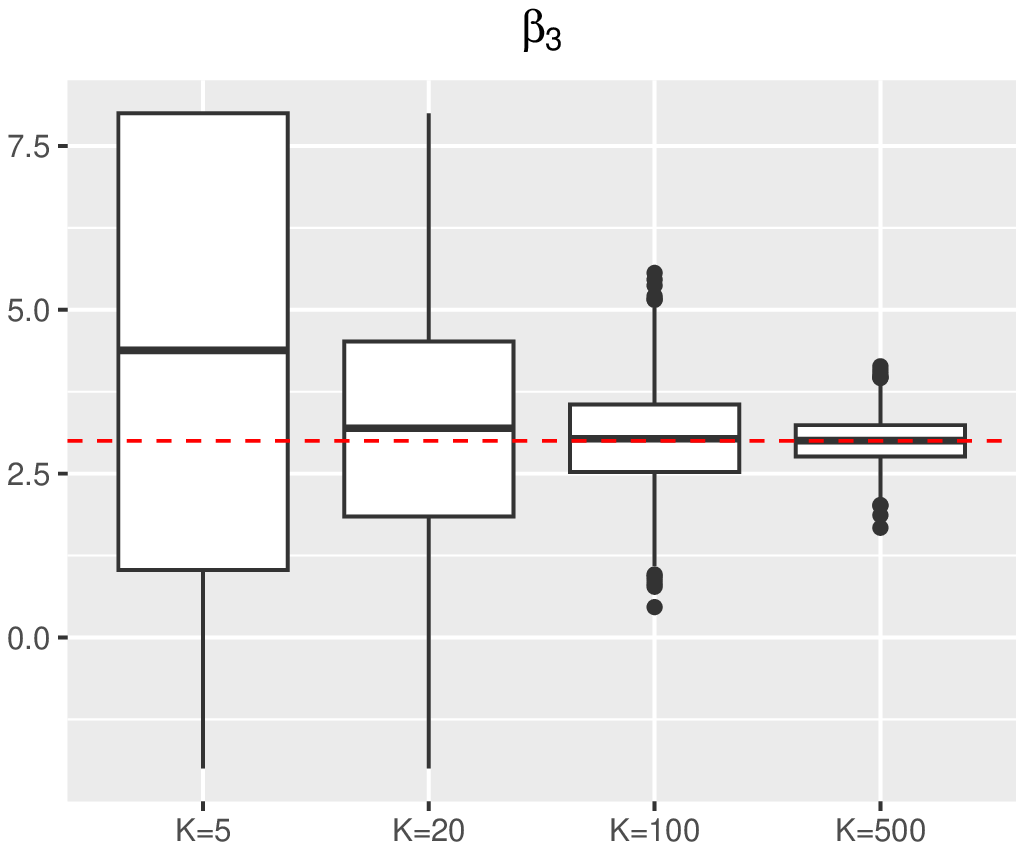}    \includegraphics[scale=0.5]{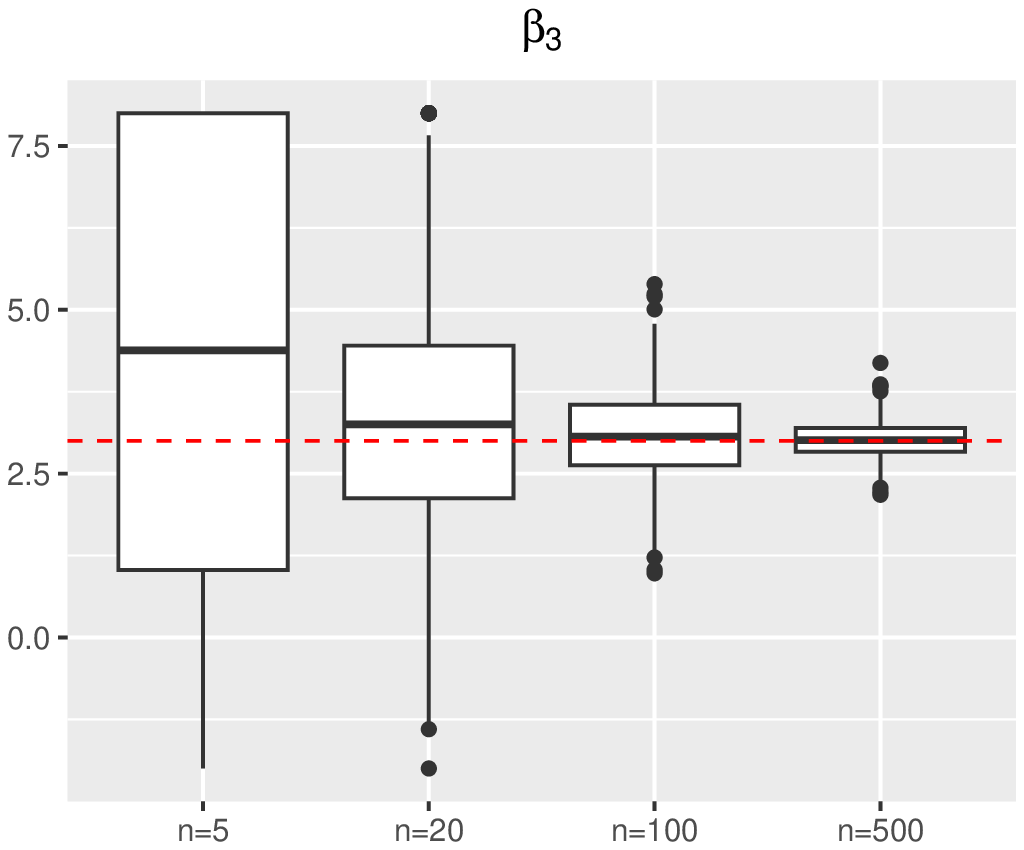}
     \caption{Exp7a: Boxplots of the estimation of $(\beta_1,\beta_2,\beta_3)$ for scenario 1 (left) and scenario 2 (right), based on 1000 samples from Frank bivariate copulas with parameters $\phi(\bbbeta,\bX_{ki}) = 2+ 8 U_{1,ki} + 3 U_{2,ki}$ and Poisson margins with means  $e^{h(\balpha,\bX_{ki})}$, where $h(\balpha,\bX_{ki})=3-U_{3,ki}-0.5U_{4,ki}$, $k\in\setK$, $i\in \{1,\ldots, n_k\}$. }
    \label{fig:frankpoi-cop}
\end{figure}

\begin{figure}[ht!]
    \centering
 \includegraphics[scale=0.5]{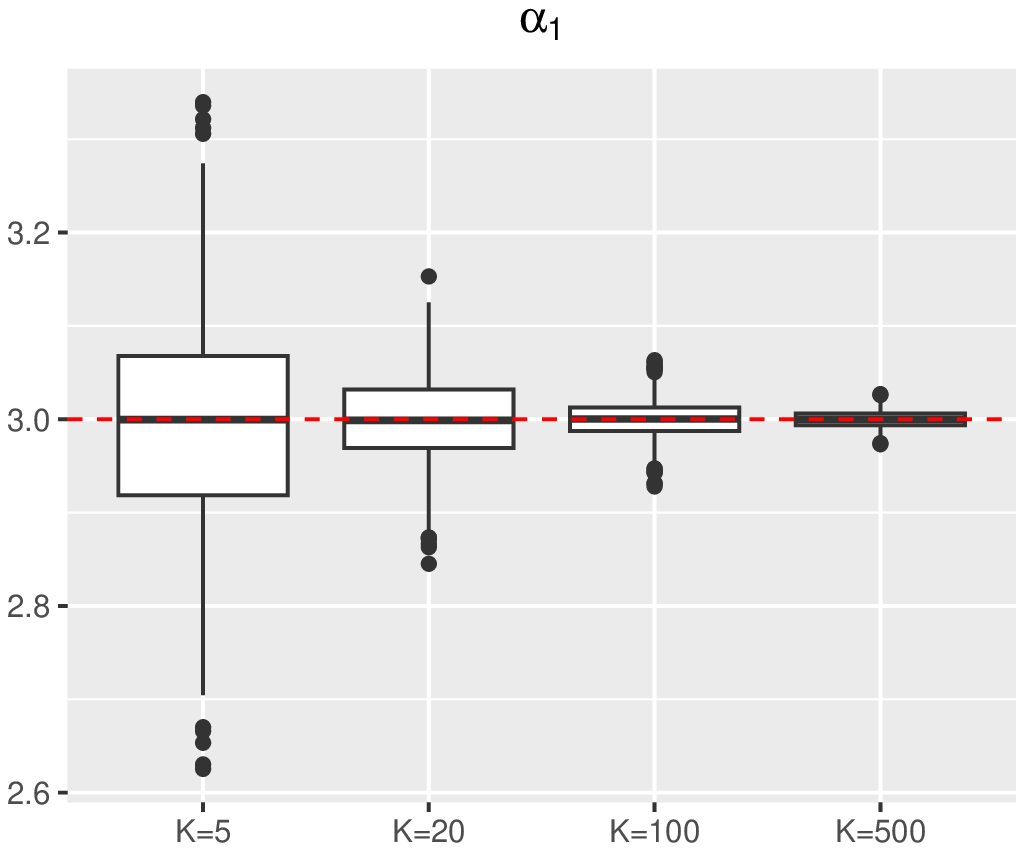}    \includegraphics[scale=0.5]{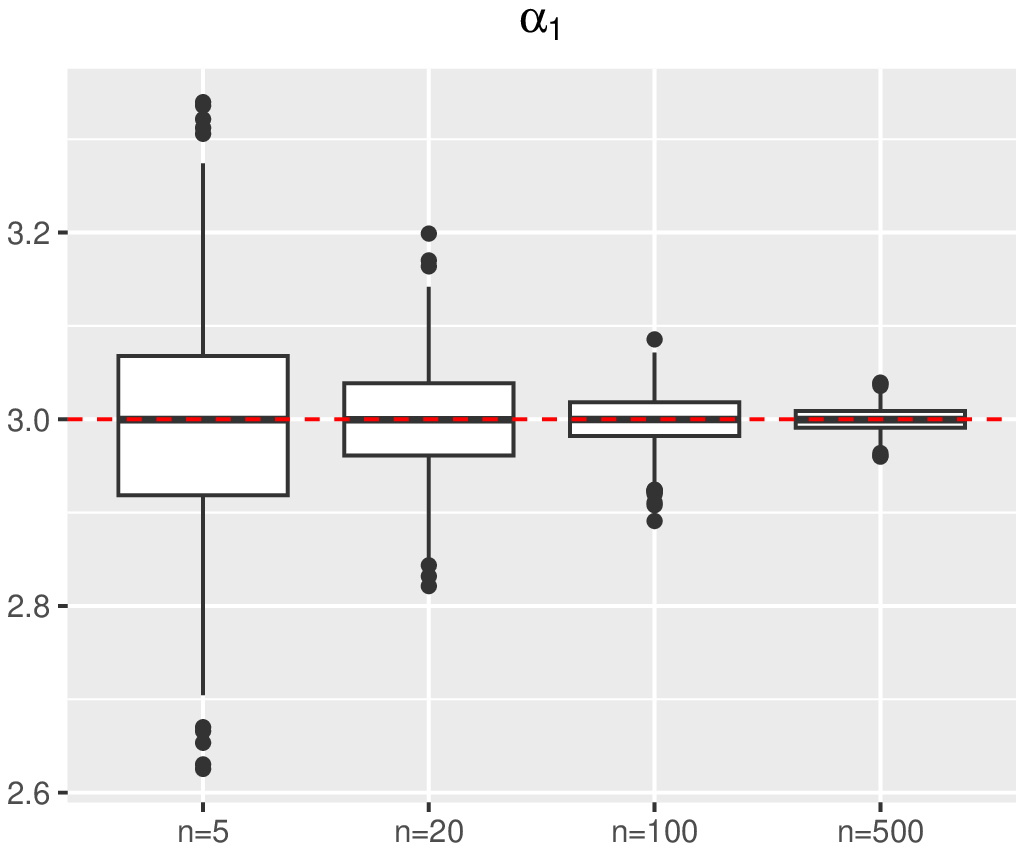}
 \includegraphics[scale=0.5]{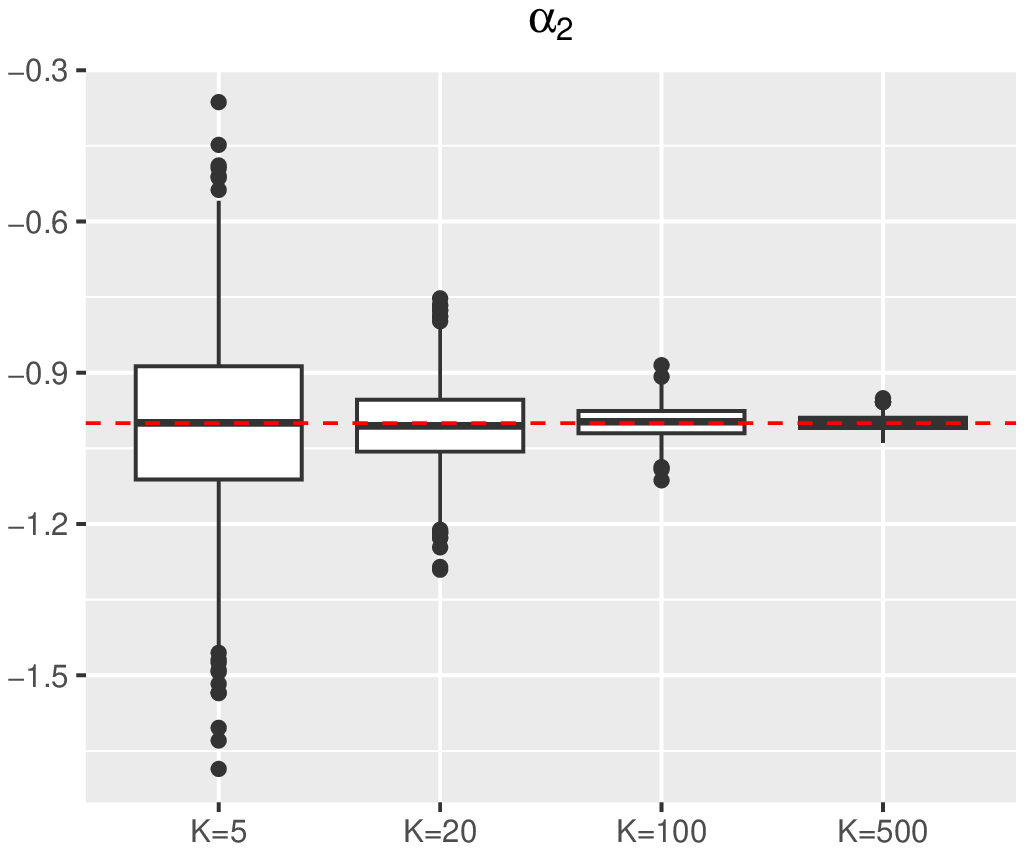}    \includegraphics[scale=0.5]{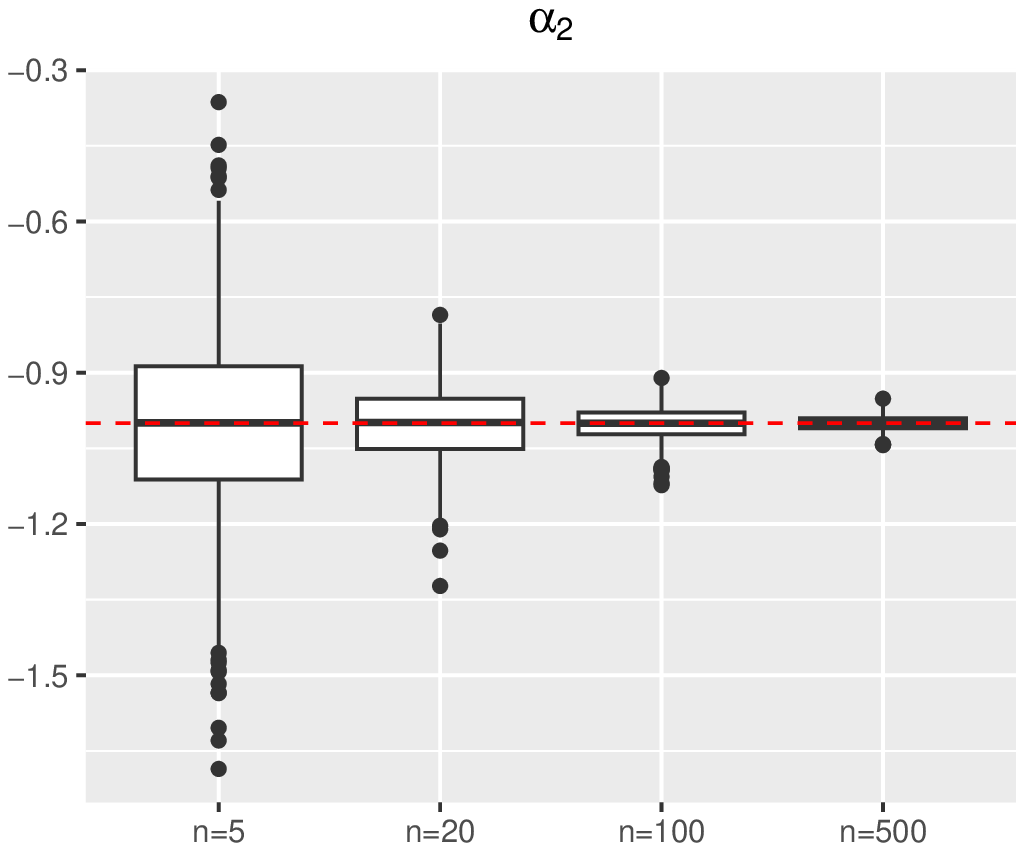}
 \includegraphics[scale=0.5]{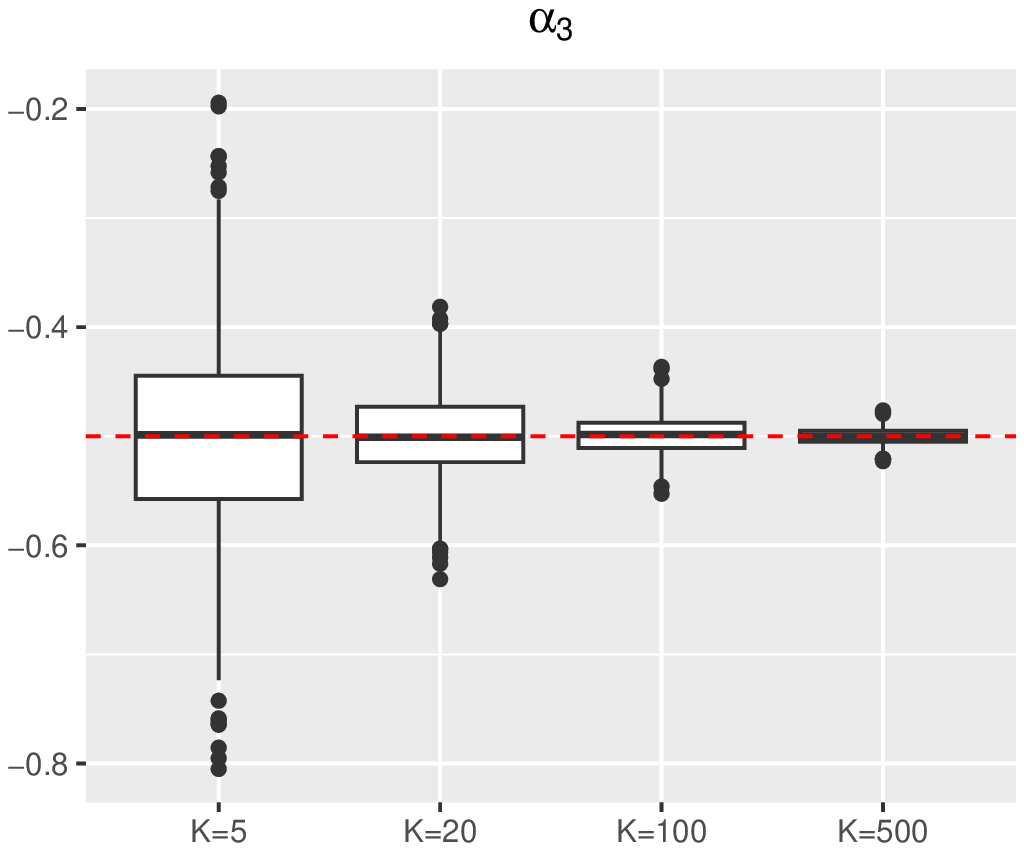}    \includegraphics[scale=0.5]{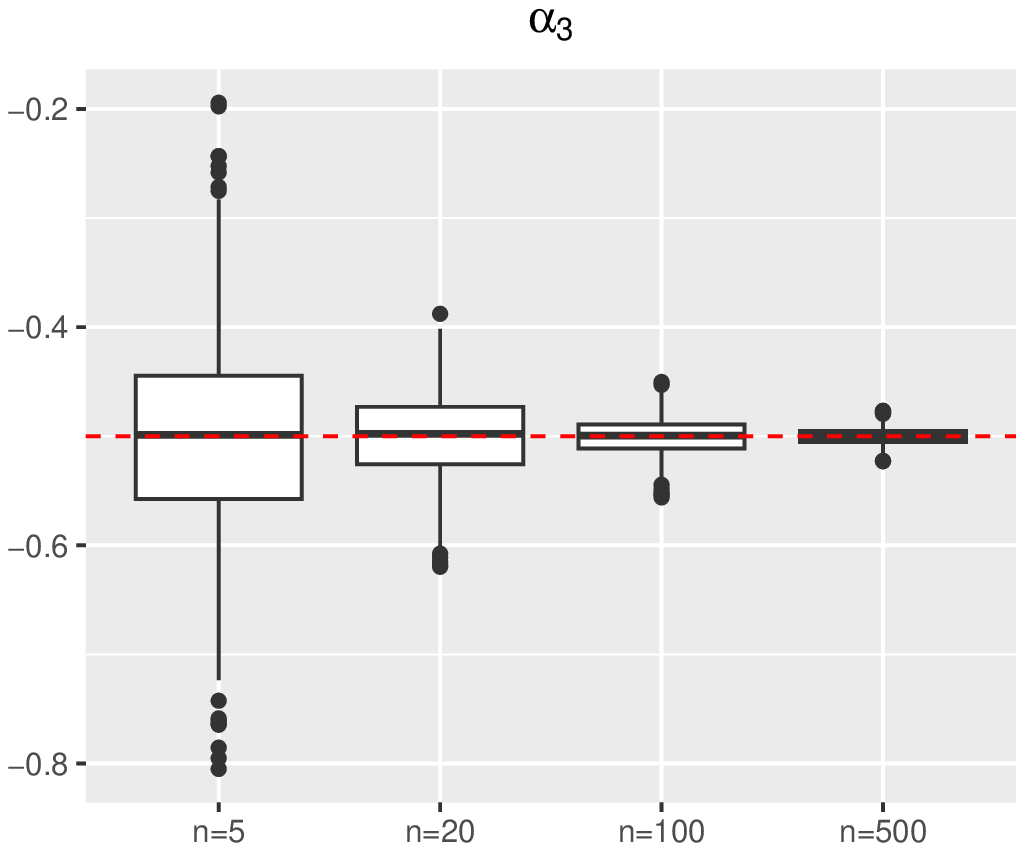}
     \caption{Exp7b: Boxplots of the estimation of $(\alpha_{1},\alpha_{2},\alpha_3)$ for scenario 1 (left) and scenario 2 (right), based on 1000 samples from Frank bivariate copulas with parameters $\phi(\bbbeta,\bX_{ki}) = 2+ 8 U_{1,ki} + 3 U_{2,ki}$ and Poisson margins with means  $e^{h(\balpha,\bX_{ki})}$, where $h(\balpha,\bX_{ki})=3-U_{3,ki}-0.5U_{4,ki}$, $k\in\setK$, $i\in \{1,\ldots, n_k\}$.}
    \label{fig:frankpoi-mar}
\end{figure}

\begin{figure}[ht!]
    \centering
 \includegraphics[scale=0.5]{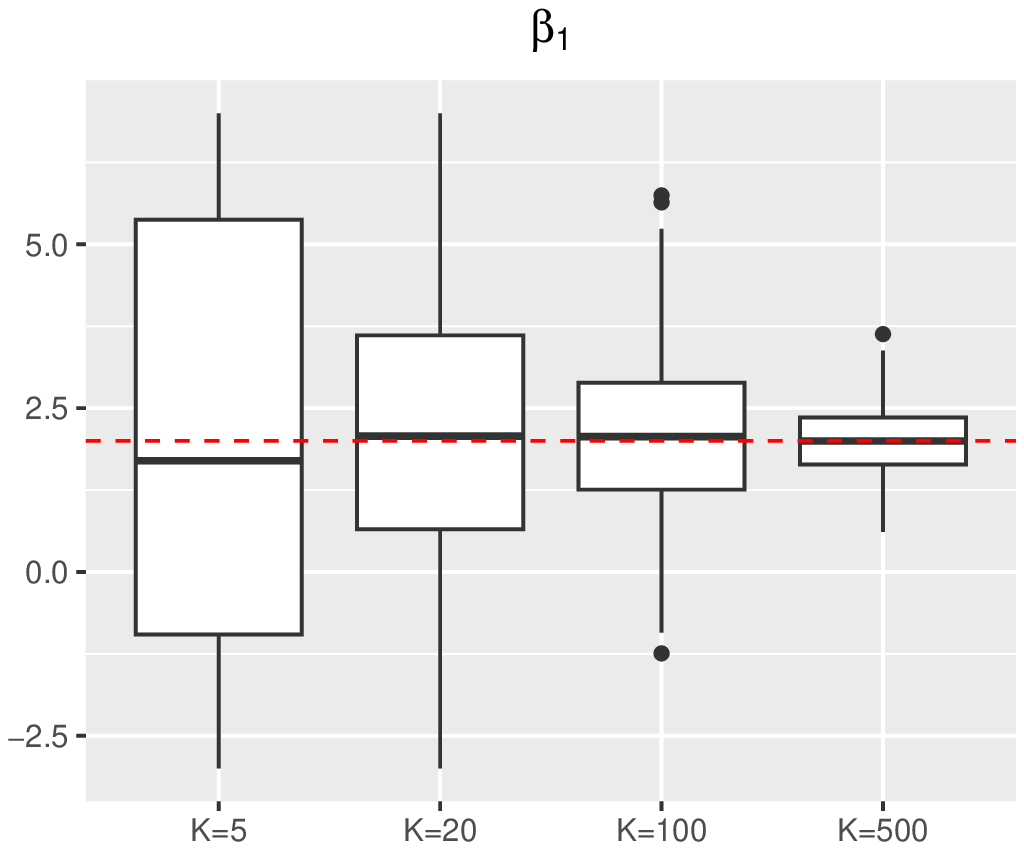}    \includegraphics[scale=0.5]{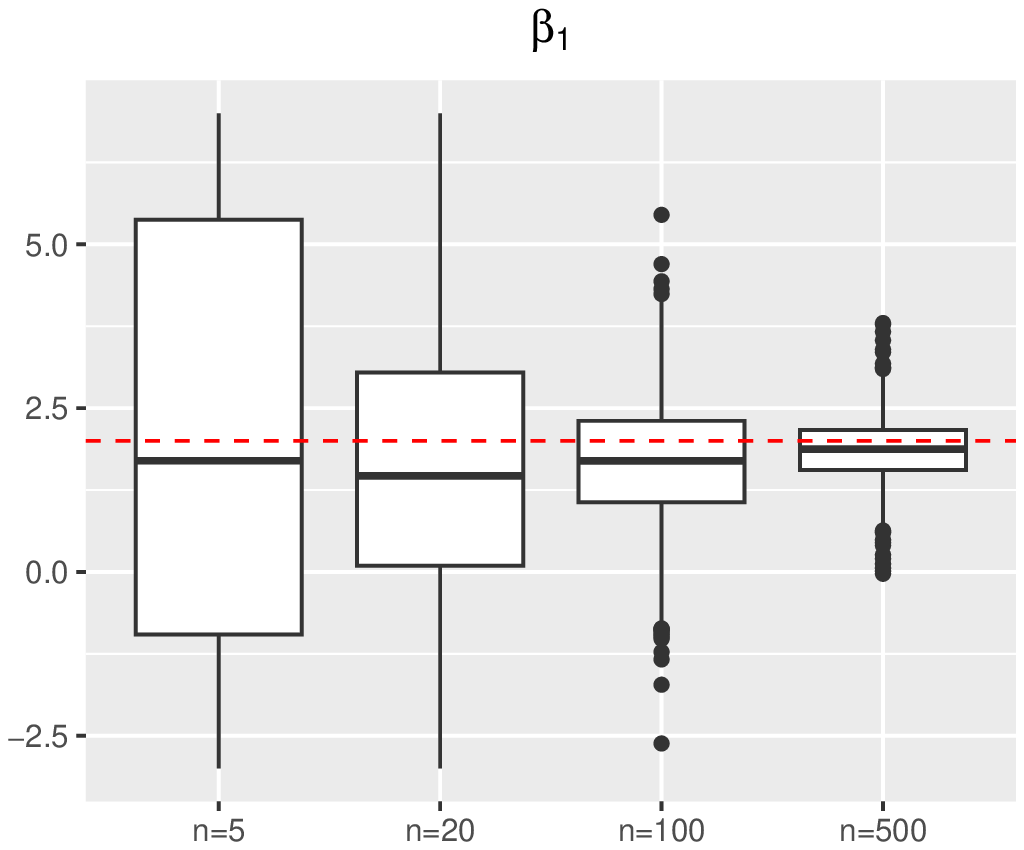}
 \includegraphics[scale=0.5]{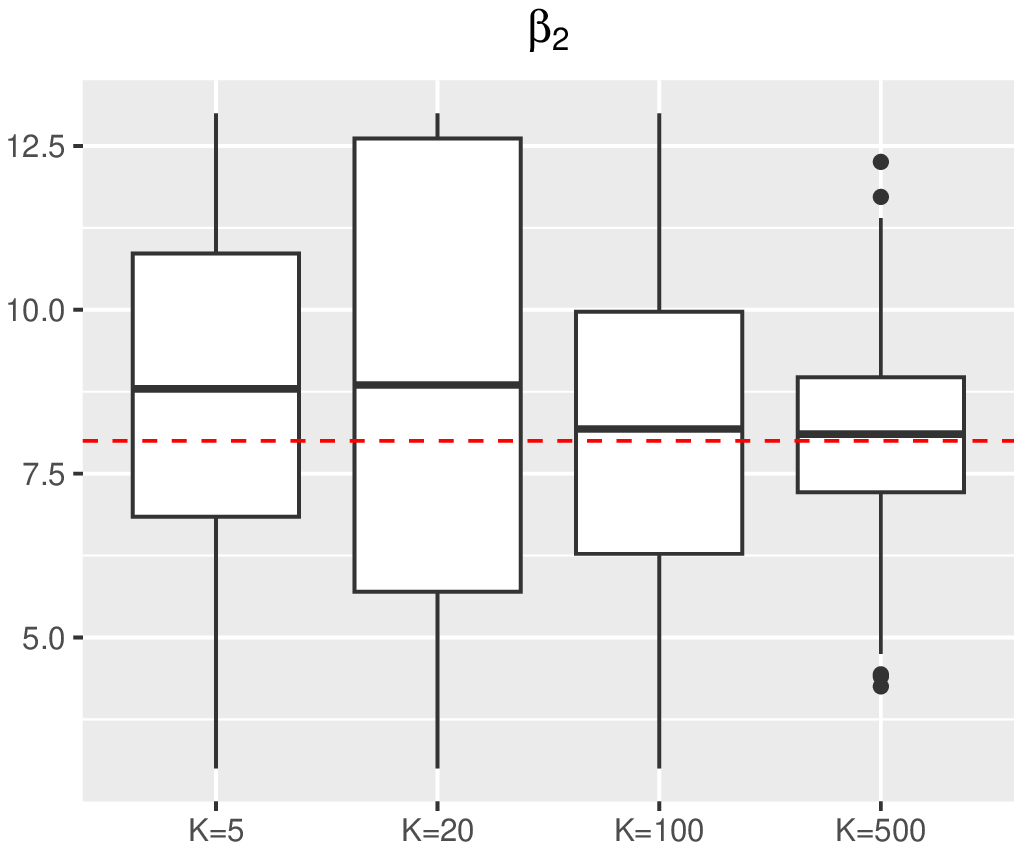}    \includegraphics[scale=0.5]{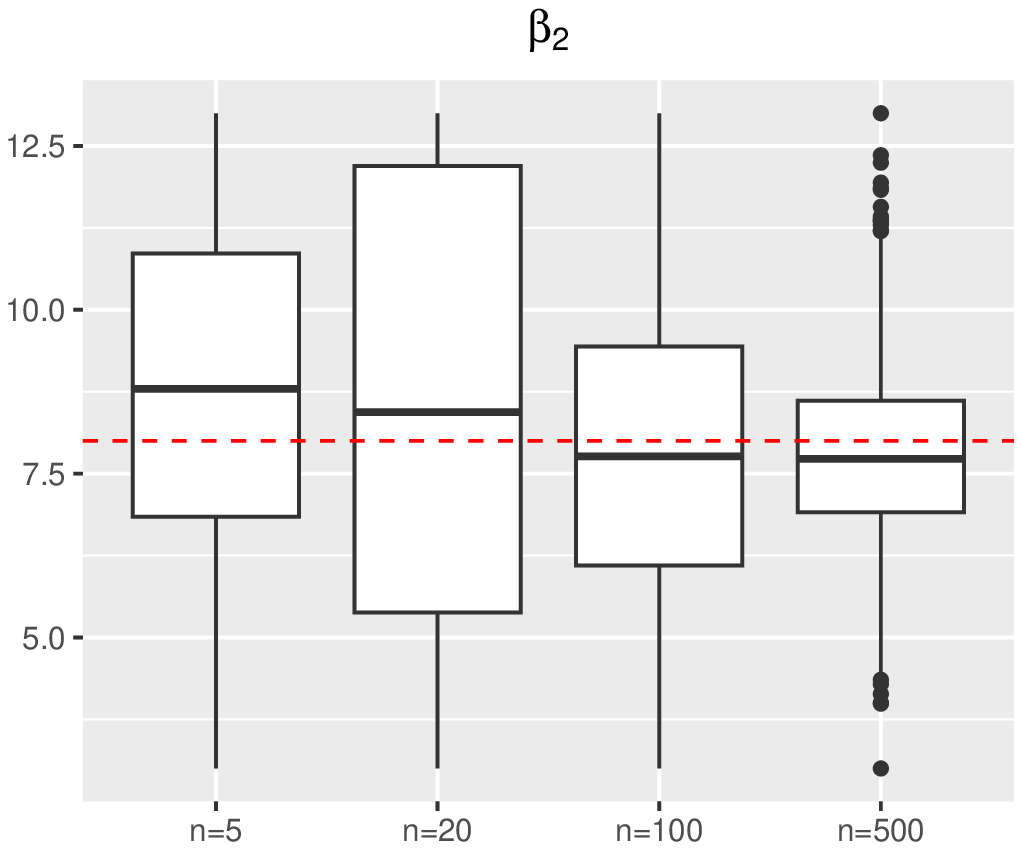}
 \includegraphics[scale=0.5]{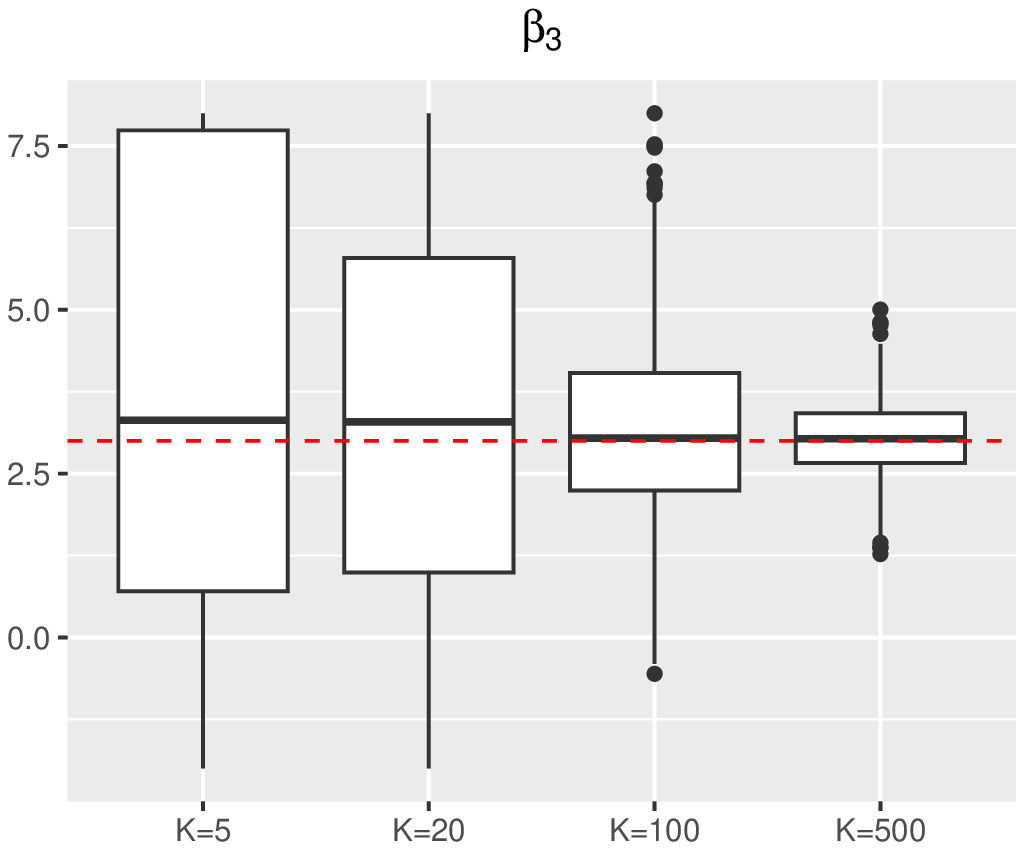}    \includegraphics[scale=0.5]{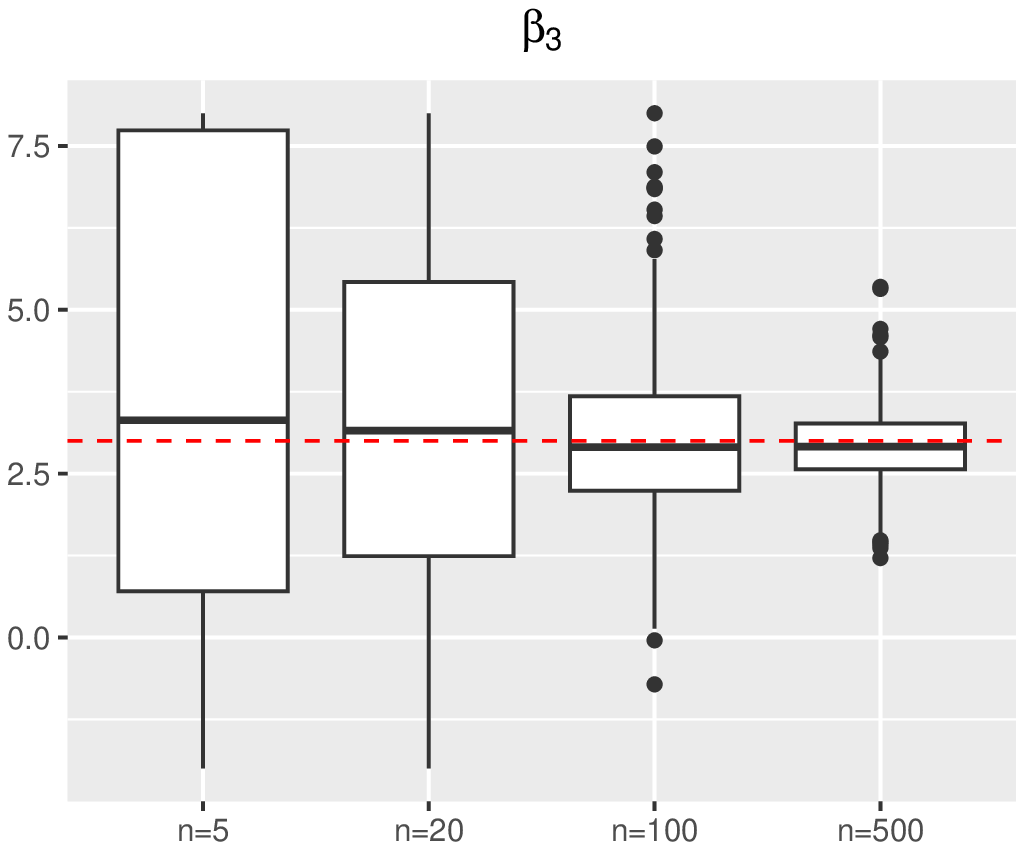}
     \caption{Exp8a: Boxplots of the estimation of $(\beta_1,\beta_2,\beta_3)$ for scenario 1 (left) and scenario 2 (right), based on 1000 samples from Frank bivariate copulas with parameters  $\phi(\bbbeta,\bX_{ki}) = 2+ 8U_{1ki} +3 U_{2ki}$ and Bernoulli  margins with parameters $ \left\{1+e^{-h(\balpha,\bX_{ki})}\right\}^{-1}$, where $h(\balpha,\bX_{ki})=1.5 -2U_{3ki} -.5U_{4ki}$, $k\in\setK$, $i\in \{1,\ldots, n_k\}$.}
    \label{fig:frankber-cop}
\end{figure}
\begin{figure}[ht!]
    \centering
 \includegraphics[scale=0.5]{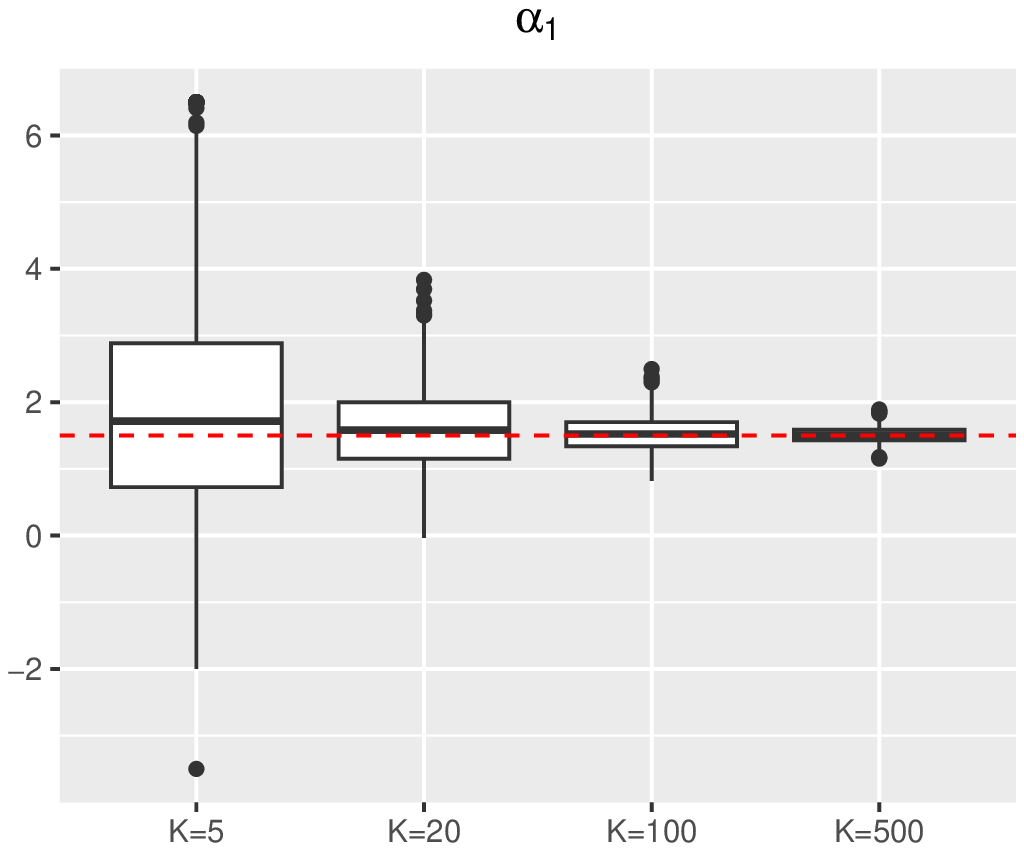}    \includegraphics[scale=0.5]{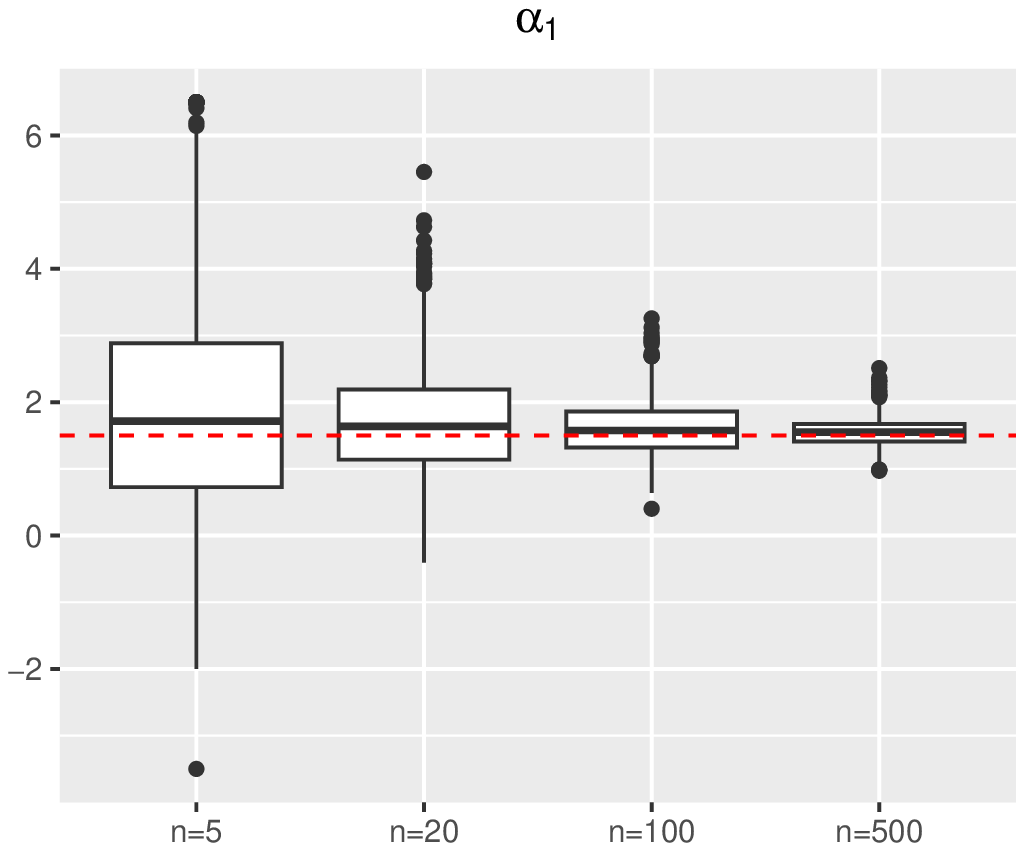}
 \includegraphics[scale=0.5]{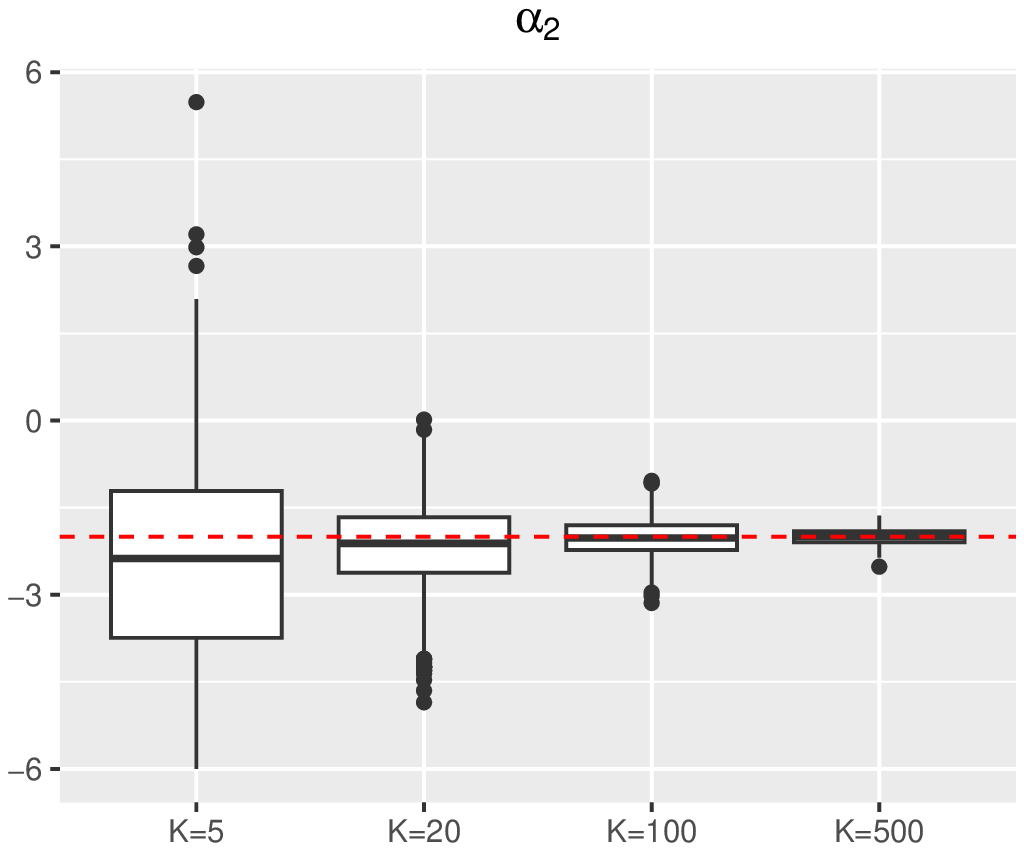}    \includegraphics[scale=0.5]{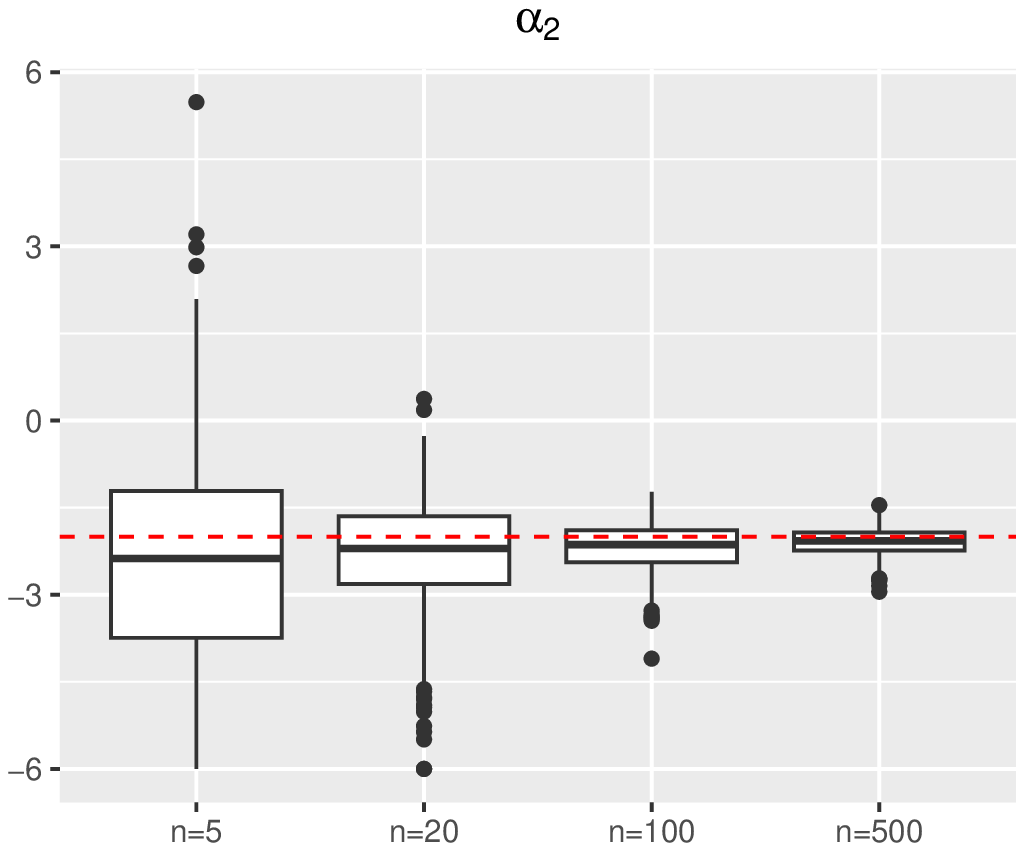}
 \includegraphics[scale=0.5]{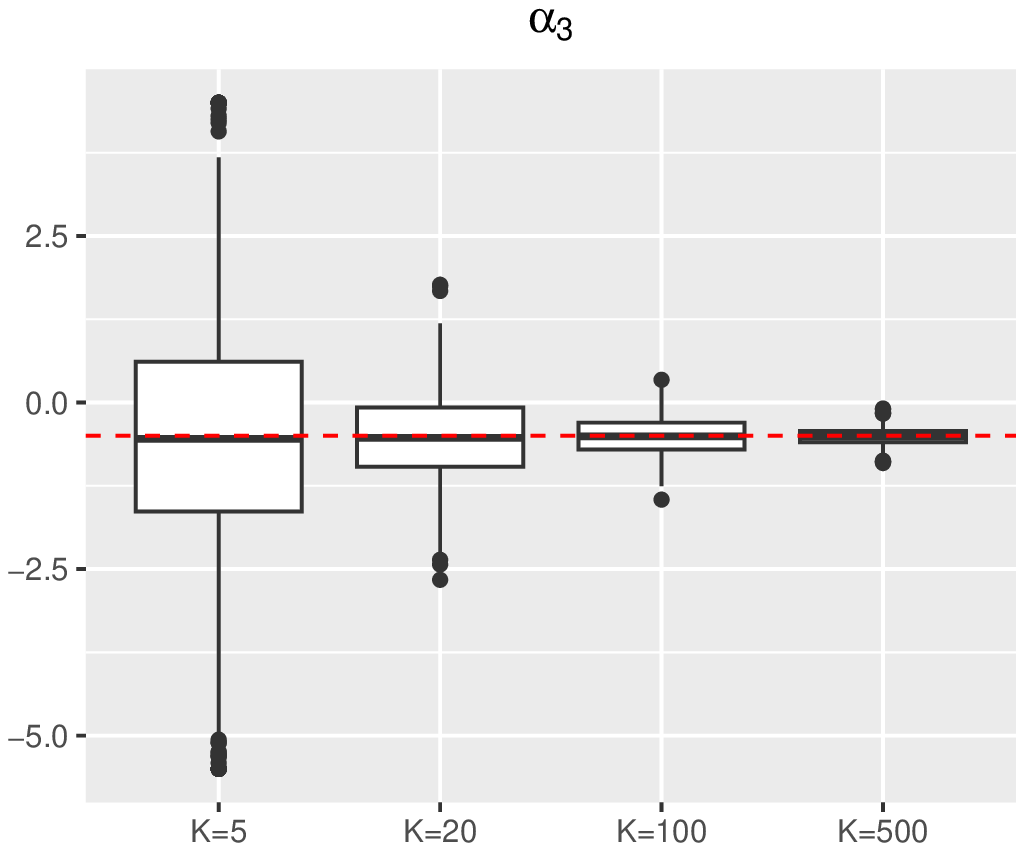}    \includegraphics[scale=0.5]{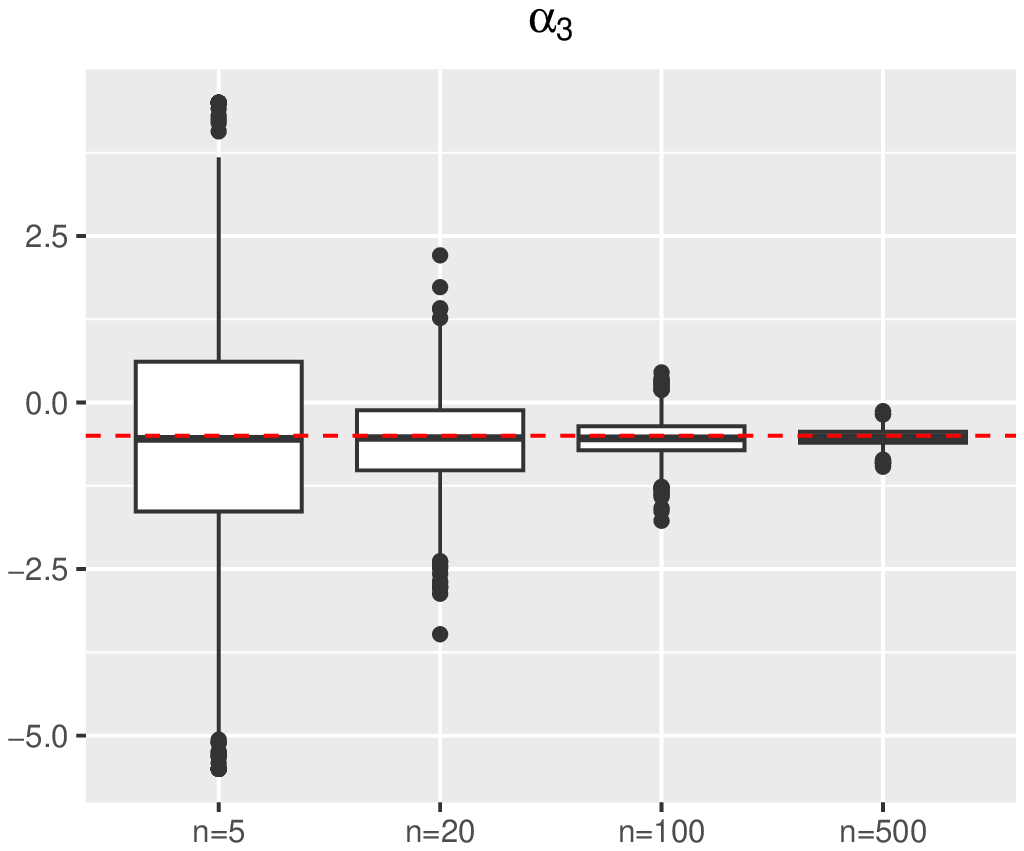}
    \caption{Exp8b: Boxplots of the estimation of $(\alpha_{1},\alpha_{2},\alpha_{3})$ for scenario 1 (left) and scenario 2 (right), based on 1000 samples from Frank bivariate copulas with parameters $\phi(\bbbeta,\bX_{ki}) = 2+ 8U_{1ki} +3 U_{2ki}$ and Bernoulli  margins with parameters $ \left\{1+e^{-h(\balpha,\bX_{ki})}\right\}^{-1}$, where $h(\balpha,\bX_{ki})=1.5 -2U_{3ki} -.5U_{4ki}$, $k\in\setK$, $i\in \{1,\ldots, n_k\}$.}
    \label{fig:frankber-mar}
\end{figure}
  Figure \ref{fig:exp1} illustrates that in the first scenario, when the number of clusters is large, parameter estimates converge to their true values, aligning with the expectations from Theorem \ref{thm:main}. In fact, the estimation of the parameters is good as soon as $K=20$. On the other hand, for the second scenario, while the parameter estimates seem asymptotically unbiased, they do not exhibit consistency when the number of clusters is fixed and the cluster size is increasing. This should not be surprising, as a similar behavior is observed in simple mixed linear models.
Next, according to Figures \ref{fig:exp2}--\ref{fig:exp5}, similar conclusions can be reached for the two scenarios: when the number of clusters is large, parameter estimates converge to their true values, while when the number of clusters is small as in the second scenario, parameter estimates seem  asymptotically unbiased, but they are still not consistent. These figures also demonstrate that when $K=20$, the estimation of $\beta_1,\beta_2$  can greatly vary, though the impact on the copula parameters $\phi(\bbbeta, \bx_{ki})$ might not be as significant. This is also true of the margin parameters, where their impact on the mean might not be significant.
Finally, when there are four covariates, as in experiments 6--8, the performance of the estimation is similar to the previous experiments. Once $K$ reaches a sufficiently large value, the estimations become highly accurate.

We also compared the precision of predictions by considering models with nine covariates for the margins. Table \ref{tab:pred-poisson} displays the results of 100 replications of 100 predictions for the (Poisson) GLMM  with log-link and copula-based models with
$\tau=0.5$, where the log-means for the Poisson margins are $\disp \alpha_0+\sum_{j=1}^9 \alpha_j X_{kij}$, and where at each iteration, the coefficients $\alpha_0,\ldots,\alpha_9$ are iid  uniform on $(0,1/2)$ and the covariates are iid
 standard Gaussian covariates. 11 coefficients were therefore estimated at each iteration, using $n=30$ observations per cluster. We computed the RMSE for the 100 predictions of new observations and
RMSE95, i.e., the RMSE for the 5 largest new observations. Table \ref{tab:pred-normal} displays similar results for a (Gaussian) GLMM, where the standard deviation of the Gaussian margins is uniformly distributed over $(3,10)$. In Table \ref{tab:pred-poisson}, where Poisson margins are employed, we see that the predictions consistently outperform those the copula-based models as they align with the true models. Turning to Table \ref{tab:pred-normal}, featuring Gaussian margins, the GLMM and copula-based models exhibit similar RMSE and RMSE95. However, for the Gumbel case, where the latter model yields superior predictions compared to the GLMM.

Next, we performed some numerical experiments to show the influence of the starting points on the estimation. Recall that since we are using the gradients, one can see if the algorithm converges or not. Specifically, in the cases of experiments 1, 5 and 7, we generated 1000 samples of different clusters and sample sizes, and we started the estimation from 6 or 7 sets of initial values, depending on the number of parameters. The RMSE for these experiments appears in Tables \ref{tab:sensExp1}--\ref{tab:sensExp7}. As exemplified in Table \ref{tab:sensExp1}, related to a Clayton copula and Gaussian margins,
for most starting points, the results are quite good. The RMSE decreases as $K$ increases, which is not the case as $n$ increases for $K=5$. The only problem is for the starting point $(0,5,1)$, where the marginal parameter is far from the true value. In this case, the algorithm does not converge, yielding NaNs and/or a large gradient.
The case where the initial values are labelled \textit{auto} corresponds to choosing the copula parameters all $0$ but the intercept with a value so that Kendall's tau is $0.4$, and the margin parameters are chosen as if the observations were independent. One can see that the results are
quite good, but for $K=5$. Next, in Table \ref{tab:sensExp5} related to Clayton copula with parameters $2e^{1-1.5U_{ki}}$ and Bernoulli margins with two parameters, the results are similar, with one starting point yielding poor convergence, one copula parameter being far from the true value. Here,
the case where the initial values are labelled \textit{auto} corresponds to choosing the copula parameters to be all $0$, which in this case yields a Kendall's tau of $0.5$, and the margin parameters are chosen as if the observations were independent, using glm with starting points all $0$ but for the intercept so that it yields the mean of the sample. Finally, from Table \ref{tab:sensExp7}, related to Frank's copula with parameters $2+8U_{1,ki}+3U_{2,ki}$ and Poisson margins with three parameters, there is also one starting point yielding large RMSE values for the copula parameters, while for the other starting points the results are correct. Also, for the margin parameters, there is another case where two parameters are not estimated correctly. The case where the initial values are labelled \textit{auto} corresponds to choosing the copula parameters all $0$, but the intercept, for which Kendall's tau is $0.5$, and the margin parameters are chosen as if the observations were independent using the R glm function and its automatic starting points. As a conclusion, one can see that the auto selection method provides results almost as good as if we were starting from the real values.
    \begin{table}
    \caption{RMSE and RMSE95 for copula-based Clayton, Frank, Gaussian and Gumbel models with Poisson margins having 10 parameters.}\label{tab:pred-poisson}
    \medskip
    \begin{tabular}{cccccccccc}
    &       & \multicolumn{4}{c}{Clayton} & \multicolumn{4}{c}{Frank}  \\
    &       & \multicolumn{2}{c}{RMSE}    & \multicolumn{2}{c}{RMSE95}    & \multicolumn{2}{c}{RMSE}   &    \multicolumn{2}{c}{RMSE95}\\
$n$ & $k$   &    GLM    & Cop        &    GLM  &    Cop   &   GLM    & Cop        &   GLM    & Cop          \\
30  &  25  & 0.1217 & 0.1093 & 0.3599 & 0.2859 & 0.1229 & 0.1083 & 0.3695 & 0.2734   \\
30  &  50  & 0.1215 & 0.1095 & 0.3545 & 0.2848 & 0.1213 & 0.1086 & 0.3528 & 0.2708   \\
30  &  100 & 0.1202 & 0.1090 & 0.3471 & 0.2824 & 0.1206 & 0.1079 & 0.3484 & 0.2668   \\
30  &  200 & 0.1216 & 0.1100 & 0.3530 & 0.2853 & 0.1221 & 0.1087 & 0.3561 & 0.2708   \\
\hline\hline    &       & \multicolumn{4}{c}{Gaussian} & \multicolumn{4}{c}{Gumbel}  \\
    &       & \multicolumn{2}{c}{RMSE}    & \multicolumn{2}{c}{RMSE95}    & \multicolumn{2}{c}{RMSE}   &    \multicolumn{2}{c}{RMSE95}\\
$n$ & $k$   &    GLM   & Cop       &    GLM  &    Cop    &  GLM    & Cop       &   GLM    & Cop  \\
30  &  25  &  0.11369 & 0.10854 & 0.2427 & 0.2107& 0.1183 & 0.0990 & 0.3651 & 0.2228    \\
30  &  50  &  0.12555 & 0.10971 & 0.3358 & 0.2438& 0.1166 & 0.0997 & 0.3431 & 0.2210    \\
30  &  100 &  0.12904 & 0.10831 & 0.3975 & 0.2579& 0.1169 & 0.0910 & 0.3484 & 0.2171    \\
30  &  200 &  0.09232 & 0.08706 & 0.2285 & 0.1980& 0.1184 & 0.0998 & 0.3550 & 0.2209   \\
\hline
    \end{tabular}
    \end{table}

    \begin{table}
    \caption{RMSE and RMSE95 for copula-based Clayton, Frank, Gaussian and Gumbel models with Gaussian margins having 11 parameters.}\label{tab:pred-normal}
    \begin{tabular}{cccccccccc}
    &       & \multicolumn{4}{c}{Clayton} & \multicolumn{4}{c}{Frank}  \\
    &       & \multicolumn{2}{c}{RMSE}    & \multicolumn{2}{c}{RMSE95}    & \multicolumn{2}{c}{RMSE}   &    \multicolumn{2}{c}{RMSE95}\\
$n$ & $k$   &    GLM   & Cop       &   GLM  &    Cop   &   GLM   & Cop       &   GLM   & Cop  \\
30  &  25   &   0.4632 & 0.5326 & 1.054  & 1.0506  & 0.5300 & 0.5288 & 1.0203 & 1.0146 \\
30  &  50   &   0.4650 & 0.5195 & 1.0720 & 1.0696  & 0.5158 & 0.5145 & 0.9922 & 0.9890 \\
30  &  100  &   0.4684 & 0.5210 & 1.0798 & 1.0774  & 0.5127 & 0.5116 & 0.9891 & 0.9849 \\
30  &  200  &   0.4870 & 0.5256 & 1.1227 & 1.1205  & 0.5181 & 0.5170 & 0.9970 & 0.9924\\
\hline\hline
    &       & \multicolumn{4}{c}{Gaussian} & \multicolumn{4}{c}{Gumbel}  \\
    &       & \multicolumn{2}{c}{RMSE}    & \multicolumn{2}{c}{RMSE95}    & \multicolumn{2}{c}{RMSE}   &    \multicolumn{2}{c}{RMSE95}\\
$n$ & $k$   &    GLM   & Cop       &    GLM  &    Cop   &  GLM     & Cop       &   GLM   & Cop  \\
30  &  25   &  0.5104 & 0.5135 & 0.8494 & 0.8496  & 0.1174 & 0.0984 & 0.3551 & 0.2191  \\
30  &  50   &  0.4901 & 0.4939 & 0.8151 & 0.8151  & 0.1139 & 0.0967 & 0.3381 & 0.2092  \\
30  &  100  &  0.4674 & 0.4699 & 0.7749 & 0.7748  & 0.1164 & 0.0982 & 0.3500 & 0.2160  \\
30  &  200  &  0.4841 & 0.4848 & 0.7996 & 0.7993  & 0.1151 & 0.0979 & 0.3419 & 0.2128 \\
\hline
    \end{tabular}
    \end{table}

\begin{table}
\caption{Exp1: RMSE from different starting points, using $B=1000$ replications, from  Clayton copula with parameter $2$ and $N(10,1)$ margins, i.e.,  $\beta_1=0$, $\alpha_1=10$, and $\alpha_2=1$. }
\label{tab:sensExp1}
\begin{tabular}{cc|cccc||cccc}
$\btheta_0$ & Param & \multicolumn{4}{c}{$K$} & \multicolumn{4}{c}{$n$}\\
     &            & 5 & 20 & 100 & 500 &   5 & 20 & 100 & 500 \\
     \hline
0    & $\beta_1 $ & 0.8340 & 0.3050 & 0.1371 & 0.0868 & 0.3315 & 0.1954 & 0.0784 & 0.2143 \\
10   & $\alpha_1$ & 0.3412 & 0.1664 & 0.0817 & 0.0497 & 0.0720 & 0.1842 & 0.1904 & 0.1810 \\
1    & $\alpha_2$ & 0.1963 & 0.0944 & 0.0480 & 0.0324 & 0.0340 & 0.1206 & 0.0277 & 0.0650 \\
\hline
1    & $\beta_1 $ & 0.8496 & 0.3428 & 0.2078 & 0.1014 & 0.3315 & 0.1954 & 0.0784 & 0.2143 \\
10   & $\alpha_1$ & 0.3550 & 0.1929 & 0.1247 & 0.0573 & 0.0720 & 0.1842 & 0.1904 & 0.1810 \\
1    & $\alpha_2$ & 0.2121 & 0.1082 & 0.0758 & 0.0376 & 0.0340 & 0.1206 & 0.0277 & 0.0650 \\
\hline
0    & $\beta_1 $ & 1.2124 & 2.0000 & 1.1162 & 1.1095 & 1.1636 & 2.0000 & 2.0000 & 2.0000 \\
5    & $\alpha_1$ & 2.8720 & 2.8850 & 2.9307 & 2.9357 & 2.9013 & 3.3604 & 2.7316 & 2.8066 \\
1    & $\alpha_2$ & 3.9897 & 3.1358 & 4.0000 & 4.0000 & 4.0000 & 2.7714 & 3.8386 & 3.6526 \\
\hline
0    & $\beta_1 $ & 0.8498 & 0.3287 & 0.1687 & 0.0999 & 0.3315 & 0.1954& 0.0784& 0.2143 \\
10   & $\alpha_1$ & 0.3550 & 0.1832 & 0.0984 & 0.0568 & 0.0720 & 0.1842& 0.1904& 0.1810 \\
2.5  & $\alpha_2$ & 0.2016 & 0.1037 & 0.0593 & 0.0371 & 0.0340 & 0.1206& 0.0277& 0.0650 \\
\hline
0.3  & $\beta_1 $ & 0.8865 & 0.3673 & 0.1976 & 0.0969& 0.3315 & 0.1954 & 0.0784 & 0.2143 \\
8    & $\alpha_1$ & 0.3860 & 0.2166 & 0.1196 & 0.0546& 0.0720 & 0.1842 & 0.1904 & 0.1810 \\
1.5  & $\alpha_2$ & 0.3199 & 0.1201 & 0.0726 & 0.0361& 0.0340 & 0.1206 & 0.0277 & 0.0650 \\
\hline
\multirow{3}{*}{auto}
& $\beta_1 $ & 2.6541 & 0.3078 & 0.1538 & 0.0912 & 2.6541 & 1.1728 & 0.3297 & 0.2160 \\
& $\alpha_1$ & 0.3485 & 0.1695 & 0.0885 & 0.0517 & 0.3485 & 0.2841 & 0.2627 & 0.2056 \\
& $\alpha_2$ & 0.1993 & 0.0968 & 0.0551 & 0.0335 & 0.1993 & 0.1330 & 0.0837 & 0.0522 \\
\hline
    \end{tabular}
\end{table}

\begin{table}
\caption{Exp5: RMSE from different starting points, using $B=1000$ replications,
from Clayton bivariate copulas with parameters $2e^{\phi(\bbbeta,\bX_{ki})}$, where $\phi(\bbbeta,\bX_{ki})=1 - 1.5 U_{ki}$,  and Bernoulli margins with parameters $ \left\{1+e^{-h(\balpha,\bX_{ki})}\right\}^{-1}$, where $h(\balpha,\bX_{ki})=2-3Z_{ki}$, i.e., $\beta_1=1$, $\beta_2=1.5$, $\alpha_1=2$, and $\alpha_2=-3$.}
\label{tab:sensExp5}
\begin{tabular}{cc|cccc||cccc}
$\btheta_0$ & Param & \multicolumn{4}{c}{$K$} & \multicolumn{4}{c}{$n$}\\
     &            & 5 & 20 & 100 & 500 &   5 & 20 & 100 & 500 \\
     \hline
1    & $\beta_1$  & 1.1445 & 0.8397& 0.5213& 0.2376& 1.1450& 0.8859& 0.5254&  0.2588  \\
1.5  & $\beta_2$  & 1.2688 & 1.3665& 0.8585& 0.3946& 1.2688& 1.2024& 0.6200&  0.2596  \\
2    & $\alpha_1$ & 1.2526 & 0.7173& 0.2772& 0.1150& 1.2526& 1.0284& 0.7029&  0.3980  \\
-3   & $\alpha_2$ & 1.4615 & 0.8526& 0.3229& 0.1296& 1.4615& 1.0342& 0.5948&  0.3089  \\
\hline
2    & $\beta_1$  &  1.1932& 0.8448& 0.5254& 0.2376& 1.1932& 0.8958& 0.5242&  0.2588  \\
-0.5 & $\beta_2$  &  1.2987& 1.3598& 0.8574& 0.3945& 1.2987& 1.2015& 0.6218&  0.2597  \\
2    & $\alpha_1$ &  1.0670& 0.7200& 0.2740& 0.1150& 1.0670& 1.0435& 0.6970&  0.3981  \\
-3   & $\alpha_2$ &  1.2320& 0.8447& 0.3191& 0.1298& 1.2320& 1.0288& 0.5947&  0.3085  \\
\hline
-0.5 & $\beta_1$  &   1.5214& 1.2072& 0.9113& 0.4534& 1.5214& 1.1287& 0.5951&  0.2588 \\
1    & $\beta_2$  &   2.2399& 1.9210& 1.4973& 0.7519& 2.2399& 1.6766& 0.7585&  0.2596 \\
2    & $\alpha_1$ &   0.6600& 0.5332& 0.2440& 0.1106& 0.6600& 0.8164& 0.7040&  0.3976 \\
-3   & $\alpha_2$ &   0.8199& 0.6596& 0.2891& 0.1252& 0.8199& 0.8704& 0.5999&  0.3095 \\
\hline
1   & $\beta_1$  &  1.1313& 0.8425& 0.5169& 0.2376& 1.1313& 0.8895& 0.5496&  0.4025 \\
-1.5& $\beta_2$  &  1.2544& 1.3700& 0.8644& 0.3945& 1.2544& 1.2111& 0.6083&  0.2425 \\
0.5 & $\alpha_1$ &  1.3392& 0.6952& 0.2707& 0.1150& 1.3392& 1.0222& 0.7205&  0.5598 \\
-1  & $\alpha_2$ &  1.5527& 0.8230& 0.3147& 0.1298& 1.5527& 0.9799& 0.6331&  0.5190 \\
\hline
1.5 & $\beta_1$  & 1.0315& 0.8397& 0.5217& 0.2376& 1.0315& 0.8863& 0.5353&  0.2618   \\
-2  & $\beta_2$  & 1.3022& 1.3584& 0.8624& 0.3945& 1.3022& 1.1915& 0.6262&  0.2604   \\
1   & $\alpha_1$ & 1.1848& 0.6834& 0.2707& 0.1150& 1.1848& 1.0027& 0.7006&  0.3986   \\
-2  & $\alpha_2$ & 1.3579& 0.8038& 0.3162& 0.1298& 1.3579& 0.9860& 0.5916&  0.3206  \\
\hline
-0.5 & $\beta_1$  &  1.5127& 1.3243& 1.3669& 1.4757& 1.5127& 1.2780& 1.1414&  1.0333  \\
1    & $\beta_2$  &  2.2817& 2.1785& 2.2874& 2.4604& 2.2817& 2.0738& 1.8723&  1.6982  \\
1    & $\alpha_1$ &  1.0125& 0.9008& 0.9067& 0.9836& 1.0125& 0.9594& 0.8335&  0.7398  \\
-2   & $\alpha_2$ &  1.1572& 0.9653& 0.9104& 0.9837& 1.1572& 1.0102& 0.7986&  0.7123  \\
\hline
\multirow{4}{*}{auto}
  & $\beta_1$  & 1.4557 & 1.0318 & 0.6616 & 0.3671 & 1.4557 & 1.0782 & 0.6005 & 0.5867     \\
  & $\beta_2$  & 1.8882 & 1.3665 & 0.7417 & 0.3531 & 1.8882 & 1.2154 & 0.4853 & 0.3918     \\
  & $\alpha_1$ & 2.9419 & 0.8087 & 0.3293 & 0.1436 & 2.9419 & 1.5352 & 0.6977 & 0.6900     \\
  & $\alpha_2$ & 4.2813 & 0.9686 & 0.3669 & 0.1652 & 4.2813 & 1.9871 & 0.9015 & 1.0416     \\
\hline
 \end{tabular}
\end{table} 



\begin{table}
\caption{Exp7: RMSE from different starting points, using $B=1000$ replications, from Frank bivariate copulas with parameters $\phi(\bbbeta,\bX_{ki}) = 2+ 8 U_{1,ki} + 3 U_{2,ki}$ and Poisson margins with means  $e^{h(\balpha,\bX_{ki})}$, where $h(\balpha,\bX_{ki})=3-U_{3,ki}-0.5U_{4,ki}$,
i.e., $\beta_1=2$, $\beta_2=8$, $\beta_3=3$, $\alpha_1=3$, $\alpha_2=-1$, and $\alpha_3=-0.5$.}
\label{tab:sensExp7}
\begin{tabular}{cc|cccc||cccc}
$\btheta_0$ & Param & \multicolumn{4}{c}{$K$} & \multicolumn{4}{c}{$n$}\\
     &            & 5 & 20 & 100 & 500 &   5 & 20 & 100 & 500 \\
     \hline
2    & $\beta_1$  &  2.8723 & 1.6921 & 0.7783 & 0.3352 & 2.8723 & 1.6446 & 0.7082   & 0.2990 \\
8    & $\beta_2$  &  3.4642 & 3.1500 & 1.6042 & 0.7137 & 3.4642 & 2.9746 & 1.3700   & 0.5896 \\
3    & $\beta_3$  &  3.7059 & 2.0429 & 0.8149 & 0.3552 & 3.7059 & 1.7521 & 0.6770   & 0.2821 \\
3    & $\alpha_1$ &  0.1135 & 0.0472 & 0.0201 & 0.0095 & 0.1135 & 0.0572 & 0.0270   & 0.0131 \\
-1   & $\alpha_2$ &  0.1793 & 0.0773 & 0.0325 & 0.0149 & 0.1793 & 0.0730 & 0.0327   & 0.0142 \\
-0.5 & $\alpha_3$ &  0.0882 & 0.0385 & 0.0169 & 0.0075 & 0.0882 & 0.0380 & 0.0168   & 0.0073 \\
\hline
5    & $\beta_1$  & 2.9705 & 1.7270 & 0.7783 & 0.3351 & 2.9705 & 1.6605 & 0.7081   & 0.2990  \\
5    & $\beta_2$  & 3.6805 & 3.1562 & 1.6041 & 0.7134 & 3.6805 & 2.9718 & 1.3699   & 0.5896  \\
1    & $\beta_3$  & 3.5037 & 2.0241 & 0.8149 & 0.3552 & 3.5037 & 1.7364 & 0.6770   & 0.2821  \\
3    & $\alpha_1$ & 0.1130 & 0.0472 & 0.0201 & 0.0095 & 0.1130 & 0.0576 & 0.0270   & 0.0131  \\
-1   & $\alpha_2$ & 0.1843 & 0.0772 & 0.0325 & 0.0149 & 0.1843 & 0.0733 & 0.0327   & 0.0142  \\
-0.5 & $\alpha_3$ & 0.0867 & 0.0383 & 0.0169 & 0.0075 & 0.0867 & 0.0383 & 0.0168   & 0.0073  \\
\hline
  -2 & $\beta_1$  &  3.7080 & 2.8451 & 1.4378 & 0.3763 & 3.7080 & 2.8603 & 1.3527   & 0.3470 \\
   4 & $\beta_2$  &  4.0057 & 3.6281 & 2.1447 & 0.8011 & 4.0057 & 3.5043 & 1.7140   & 0.6020 \\
  -1 & $\beta_3$  &  3.9501 & 2.9285 & 1.4783 & 0.4191 & 3.9501 & 2.6177 & 1.2234   & 0.4514 \\
   3 & $\alpha_1$ &  0.1126 & 0.0505 & 0.0234 & 0.0117 & 0.1126 & 0.0714 & 0.0399   & 0.0145 \\
  -1 & $\alpha_2$ &  0.1794 & 0.0883 & 0.0376 & 0.0191 & 0.1794 & 0.0984 & 0.0513   & 0.0148 \\
-0.5 & $\alpha_3$ &  0.0927 & 0.0443 & 0.0224 & 0.0083 & 0.0927 & 0.0473 & 0.0205   & 0.0077 \\
\hline
2   & $\beta_1$  &  1.6146 & 0.5055 & 0.4018 & 0.3830 & 1.6146 & 0.6623 & 0.7253   & 0.7912 \\
8   & $\beta_2$  &  1.6576 & 0.6478 & 0.5605 & 0.3904 & 1.6576 & 0.8056 & 0.8652   & 0.9576 \\
3   & $\beta_3$  &  1.7768 & 0.5872 & 0.4405 & 0.3149 & 1.7768 & 0.5913 & 0.7743   & 0.9051 \\
3   & $\alpha_1$ &  1.3132 & 0.8999 & 0.4428 & 0.1978 & 1.3132 & 1.0832 & 0.8329   & 0.9249 \\
2   & $\alpha_2$ &  2.0697 & 2.6770 & 2.8684 & 2.9738 & 2.0697 & 2.5397 & 2.4847   & 2.2434 \\
2.5 & $\alpha_3$ &  2.2203 & 2.7332 & 2.8407 & 2.9547 & 2.2203 & 2.6017 & 2.4353   & 2.1722 \\
\hline
1   & $\beta_1$  &  2.3891 & 1.7354 & 1.2465 & 1.0866 & 2.3891 & 1.6391 & 1.1206   & 1.03581   \\
6   & $\beta_2$  &  2.6610 & 2.3857 & 2.0817 & 2.0195 & 2.6610 & 2.2789 & 2.0667   & 2.02165   \\
1   & $\beta_3$  &  2.9733 & 2.1958 & 2.1113 & 2.0487 & 2.9733 & 2.1385 & 2.0344   & 2.00923   \\
1   & $\alpha_1$ &  1.2624 & 1.5249 & 1.7913 & 1.9476 & 1.2624 & 1.6226 & 1.9072   & 1.99194   \\
-2  & $\alpha_2$ &  0.9007 & 0.8935 & 0.9380 & 0.9888 & 0.9007 & 0.9079 & 0.9730   & 1.01428   \\
0.5 & $\alpha_3$ &  0.8161 & 0.8834 & 0.9984 & 1.0313 & 0.8161 & 0.9210 & 1.0025   & 1.01608   \\
\hline
1   & $\beta_1$  &  2.9378 & 1.6777 & 0.8375 & 0.4207 & 2.9378 & 1.6394 & 0.8354   & 0.3946   \\
6   & $\beta_2$  &  3.4256 & 3.1149 & 1.6852 & 0.8767 & 3.4256 & 2.9722 & 1.4659   & 0.7713   \\
1   & $\beta_3$  &  3.6971 & 2.0279 & 0.9608 & 0.5045 & 3.6971 & 1.7665 & 0.8526   & 0.5330   \\
2   & $\alpha_1$ &  0.1494 & 0.2131 & 0.2346 & 0.1321 & 0.1494 & 0.1938 & 0.2564   & 0.1904   \\
-2  & $\alpha_2$ &  0.2034 & 0.2323 & 0.2472 & 0.1388 & 0.2034 & 0.2040 & 0.2588   & 0.1949   \\
-1.5& $\alpha_3$ &  0.1313 & 0.2133 & 0.2421 & 0.1457 & 0.1313 & 0.1918 & 0.2618   & 0.1958   \\
\hline
\multirow{6}{*}{auto}
& $\beta_1$  &5.2858 & 1.8143 & 0.7613 & 0.3352 & 5.2858 & 1.6166 & 0.6742 & 0.2978\\
& $\beta_2$  &7.5362 & 2.8155 & 1.4361 & 0.7130 & 7.5362 & 2.4763 & 1.2516 & 0.5871\\
& $\beta_3$  &5.8312 & 2.1921 & 0.8100 & 0.3554 & 5.8312 & 1.8649 & 0.6865 & 0.2814\\
& $\alpha_1$ &0.1193 & 0.0474 & 0.0201 & 0.0095 & 0.1193 & 0.0604 & 0.0271 & 0.0130\\
& $\alpha_2$ &0.1885 & 0.0778 & 0.0327 & 0.0149 & 0.1885 & 0.0740 & 0.0328 & 0.0142\\
& $\alpha_3$ &0.0911 & 0.0383 & 0.0170 & 0.0075 & 0.0911 & 0.0390 & 0.0168 & 0.0073\\
\hline
\end{tabular}
\end{table} 

\subsection{Comparisons of models}\label{ssec:comp}
In this section, we compare the performance of the copula-based models with GAMM competitors in terms of the selection percentage of the models with respect to AIC and BIC criteria. To do so, for each of the eight clustered Data Generating Processes described next, we simulated data $(Y_{ki},X_{ki})$, $i\in\{1,\ldots,30\}$ and $k\in\{1,\ldots,K\}$, with $K\in \{25, 50, 100, 200\}$. We then computed the percentage of time each model among GAMM1--GAMM5 and the five copula-based models were selected based on AIC and BIC criteria. The results are displayed in Table \ref{tab:comp}, where the results for the theoretical model are highlighted in bold.\\

We considered eight clustered Data Generating Processes, defined as follows:
for $i\in\{1,\ldots,30\}$, $k\in\{1,\ldots,K\}$, $X_{ki} $ are iid uniform, $s_1(x) = 10+0.5 x$,
$s_2(x) =  (2+x)\I\{x<0.5\}+ (4-3x)\I\{x\ge 0.5\}$, $s_7(x) = -1.69+3x$, $s_8(x) = (-2.7+10.6x)\I(x\le 0.5) + (4.6-4x)I(x > 0.5) $, and
\begin{itemize}
\item
DGP1: $Y_{ki} = \eta_k + s_1(X_{ki})  + \epsilon_{ki}$, where $\eta_k\sim N(0,1)$, $\epsilon_{ki}\sim N\left(0,1.5^2\right)$;
\item
DGP2: $Y_{ki} = \eta_k + s_2(X_{ki})+ \epsilon_{ki}$, where $\eta_k\sim N(0,1)$, $\epsilon_{ki}\sim N\left(0,1.5^2\right)$;
\item
DGP3: Copula-based model with Gumbel copula with parameter $2$ and Gaussian margins with mean $s_1(X_{ki}) $ and variance $1.5^2$.
\item
DGP4: Copula-based model with Gumbel copula with parameter $2$ and Gaussian margins with mean $s_2(X_{ki})$ and variance $1.5^2$.
\item DGP5: Copula-based model with Clayton copula with parameter $2$ and Poisson margins with mean $e^{1.5X_{ki}} $;
\item DGP6: Copula-based model with Clayton copula with parameter $2$ and Poisson margins with mean $e^{s_2(X_{ki})} $;
\item DGP7: Copula-based model with Frank  copula with parameter $6$ and Bernoulli margins with mean $p_{ki} = \left(1+ e^{-s_7(X_{ki})}\right)^{-1}$, i.e.,  $y_{ki}=\I\{U_{ki}> 1-p_{ki}\}$, where $(U_{ki},V_{k})$ are generated from Frank copula;
\item DGP8: Copula-based model with Frank copula with parameter $6$ and Bernoulli margins with mean $p_{ki} = \left(1+ e^{-s_8(X_{ki})}\right)^{-1}$.
\end{itemize}
Next, consider the linear functions $h_1(\bb,x) = b_0+b_1 x$, $h_2(\bb,x) = b_0+b_1 x+b_2 x^2$,
$h_3(\bb,z)= b_0+b_1 sp_{31}(x)+b_2 sp_{32}(x)$, where $sp_{31}$ and $sp_{32}$ are the basis functions for B-splines of degree $1$ and one knot at $0.5$, $h_4(\bb,x)= b_0+b_1 sp_{41}(x)+b_2 sp_{42}(x)+b_3 sp_{43}(x)+b_4 sp_{44}(x)$, where $sp_{41},\ldots,sp_{44}$ are the basis functions for B-splines of degree $2$ and knots at $0.33,0.67$, and $h_5(\bb,x)= b_0+b_1 sp_{51}(x)+b_2 sp_{52}(x)+b_3 sp_{53}(x)+b_4 sp_{54}(x)+b_5 sp_{55}(x)$, where $sp_{51},\ldots,sp_{55}$ are the basis functions for B-splines of degree $2$ and knots at $0.25,0.5,0.75$.
We denote the corresponding GAMM models as GAMM1--GAMM5. For DGP1--DGP4, the link function is the identity and we implemented a Gaussian mixed linear regression, while for DGP5--DGP6, we implemented a Poisson mixed linear regression and log link function. Finally, for DGP7--DGP8, we implemented a binomial mixed linear regression with ``logit'' link function.
When implementing the copula-based models with constant Clayton, Gaussian, Gumbel, Frank, and Student(15) copulas,
we assumed the margins were Gaussian with mean $h_1(\bb,X_{ki}) $ and constant variance for DGP1 and DGP3, and
  Gaussian with a mean of $h_3(\bb,X_{ki})$ and a constant variance for DGP2, DGP4.
For DGP5 and DGP6, the margins are Poisson, with respective log-means of $h_1(\bb,X_{ki})$ and  $h_3(\bb,X_{ki})$.
  Finally, for DGP7--DGP8, the margins are Bernoulli with logit functions $h_1(\bb, X_{ki})$ and $h_3(\bb, Z_{ki})$ respectively.\\

For DGP1, both GAMM1 and the Gaussian copula-based model are equivalent, as shown in Appendix \ref{app:gaussian}. However, as can be seen from the results in Table \ref{tab:comp}, as $K$ increases, the copula-based model is mostly preferred to the linear mixed model. This also occurs in DGP2. This might be due to the fact that the coefficient of the copula (correlation) is easier to estimate than the variances of the Gaussian GAMM. For DGP3 to DGP6, the best model is the copula-based model, and the percentage of success is almost 100\%, even for small $K$. Finally, for DGP7 and DGP8, the correct copula-based model consistently outperforms other models, with a success rate increasing with $K$. As noted in Proposition \ref{prop:gamm}, there is an infinite number of copulas yielding the Logistic model, so it is not a surprise that the real underlying copula is not selected as often as in the previous cases. However, with $K=200$, the success rate is over 90\%.

\begin{landscape}
\begin{table}[ht!]
\caption{
Percentage of time each model is chosen, according to the AIC and BIC criteria over 100 replications, for models DGP1--DGP8,
using GAMM1--GAMM5, and copula-based models with Clayton, Gaussian, Gumbel,
Frank, and Student(15) copulas. The scores for the theoretical 
model are in bold.}\label{tab:comp}
{\tiny \begin{tabular}{cccccccccccccccccccccc}
DGP  & K & \multicolumn{2}{c}{GAMM1} & \multicolumn{2}{c}{GAMM2}  & \multicolumn{2}{c}{GAMM3}  & \multicolumn{2}{c}{GAMM4}
& \multicolumn{2}{c}{GAMM5}  & \multicolumn{2}{c}{Clayton}  & \multicolumn{2}{c}{Gaussian} &  \multicolumn{2}{c}{Gumbel} &  \multicolumn{2}{c}{Frank}
&  \multicolumn{2}{c}{Student(15)}\\
&& AIC & BIC  & AIC & BIC & AIC & BIC & AIC & BIC & AIC & BIC & AIC & BIC & AIC & BIC & AIC & BIC & AIC & BIC & AIC & BIC \\
\hline
\multirow{4}{*}{1}
&  25  &  \textbf{33}  & \textbf{40}  &  7  &  0  &  7  &  0  &  1  &  0  &  2  &   0   &  0  &   0  &  \textbf{24}  &  \textbf{31} &    0  & 0   &  4  &   4  &  22  &  25\\
&  50  &  \textbf{31}  & \textbf{41}  &  8  &  1  &  6  &  1  &  9  &  0  &  2  &   0   &  0  &   0  &  \textbf{30} &  \textbf{42} &    0  & 0   &  1  &   1  &  13  &  14\\
&  100 &  \textbf{30}  & \textbf{42}  &  8  &  0  &  8  &  1  &  7  &  0  &  2  &   0   &  0  &   0  &  \textbf{41}  &  \textbf{51} &    0  & 0   &  0  &   0  &   4  &   6\\
&  200 &  \textbf{26}  & \textbf{36}  &  7  &  0  & 12  &  0  &  3  &  0  &  3  &   0   &  0  &   0  &  \textbf{48}  &  \textbf{63} &    0  & 0   &  0  &   0  &   1  &   1\\
\hline
\multirow{4}{*}{2}
&  25  &  0  &  0 &  21  & 22  & \textbf{25}  & \textbf{28}  &  2  &  0  &  3   &  0 &    0  &   0  &  \textbf{24}  &  \textbf{25} &    0  & 0   &  3   &  3  &  22  &  22   \\
&  50  &  0  &  0 &  16  & 17  & \textbf{33}  & \textbf{36}  &  7  &  0  &  3   &  0 &    0  &   0  &  \textbf{29}  &  \textbf{33} &    0  & 0   &  0   &  1  &  12  &  13   \\
&  100 &  0  &  0 &   6  &  7  & \textbf{31}  & \textbf{36}  &  9  &  0  &  2   &  0 &    0  &   0  &  \textbf{47}  &  \textbf{51} &    0  & 0   &  0   &  0  &   5  &   6    \\
&  200 &  0  &  0 &   1  &  1  & \textbf{30}  & \textbf{36}  &  4  &  0  &  6   &  0 &    0  &   0  &  \textbf{58}  &  \textbf{62} &    0  & 0   &  0   &  0  &   1  &   1   \\
\hline
\multirow{4}{*}{3}
&  25  & 0  &  0  &  1   & 0  &  0  &  0  &  1   & 0  &  0  &   0   &  0   &  0   &  0  &   0  & \textbf{ 98}  & \textbf{ 99}   &  0  &   1   &  0   &  0   \\
&  50  & 0  &  0  &  0   & 0  &  0  &  0  &  0   & 0  &  0  &   0   &  0   &  0   &  0  &   0  & \textbf{100}  & \textbf{100}   &  0  &   0   &  0   &  0   \\
&  100 & 0  &  0  &  0   & 0  &  0  &  0  &  0   & 0  &  0  &   0   &  0   &  0   &  0  &   0  & \textbf{100}  & \textbf{100}   &  0  &   0   &  0   &  0   \\
&  200 & 0  &  0  &  0   & 0  &  0  &  0  &  0   & 0  &  0  &   0   &  0   &  0   &  0  &   0  & \textbf{100}  & \textbf{100}   &  0  &   0   &  0   &  0   \\
\hline
\multirow{4}{*}{4}
&  25  & 0  &  0   & 0 &   0 &   0  &  0 &   0  &  0  &  0  &   0   &  0  &   0  &   0  &   0  & \textbf{100}  & \textbf{100}  &   0  &   0  &   0  &   0 \\
&  50  & 0  &  0   & 0 &   0 &   0  &  0 &   0  &  0  &  0  &   0   &  0  &   0  &   0  &   0  & \textbf{100}  & \textbf{100}  &   0  &   0  &   0  &   0 \\
&  100 & 0  &  0   & 0 &   0 &   0  &  0 &   0  &  0  &  0  &   0   &  0  &   0  &   0  &   0  & \textbf{100}  & \textbf{100}  &   0  &   0  &   0  &   0 \\
&  200 & 0  &  0   & 0 &   0 &   0  &  0 &   0  &  0  &  0  &   0   &  0  &   0  &   0  &   0  & \textbf{100}  & \textbf{100}  &   0  &   0  &   0  &   0 \\
\hline
\multirow{4}{*}{5}
&  25  &18   & 22 & 3 & 1 & 0& 0 & 3 & 0& 0 & 0 & \textbf{75}& \textbf{76}& 1  & 1  & 0  & 0   & 0  & 0 & 0  & 0  \\
&  50  & 9   & 11 & 0 & 0 & 1& 0 & 1 & 0& 0 & 0 & \textbf{89}& \textbf{89}& 0  & 0  & 0  & 0   & 0  & 0 & 0  & 0 \\
&  100 & 2   &  4 & 0 & 0 & 1& 0 & 1 & 0& 0 & 0 & \textbf{96}& \textbf{96}& 0  & 0  & 0  & 0   & 0  & 0 & 0  & 0 \\
&  200 & 1   &  1 & 0 & 0 & 0& 0 & 0 & 0& 0 & 0 & \textbf{99}& \textbf{99}& 0  & 0  & 0  & 0   & 0  & 0 & 0  & 0 \\   \hline
\multirow{4}{*}{6}
&  25  & 0  &  0  &  0 &   0  &  0  &  0   & 0  &  0 &   0  &   0 &  \textbf{100} &  \textbf{100}   &  0   &  0   &  0 & 0   &  0   &  0  &   0  &   0  \\
&  50  & 0  &  0  &  0 &   0  &  0  &  0   & 0  &  0 &   0  &   0 &  \textbf{100} &  \textbf{100}   &  0   &  0   &  0 & 0   &  0   &  0  &   0  &   0\\
&  100 & 0  &  0  &  0 &   0  &  0  &  0   & 0  &  0 &   0  &   0 &  \textbf{100} &  \textbf{100}   &  0   &  0   &  0 & 0   &  0   &  0  &   0  &   0\\
&  200 & 0  &  0  &  0 &   0  &  0  &  0   & 0  &  0 &   0  &   0 &  \textbf{100} &  \textbf{100}   &  0   &  0   &  0 & 0   &  0   &  0  &   0  &   0\\
\hline
\multirow{4}{*}{7}
&  25  & 2 & 2 & 2 & 0 & 2 & 0 & 1 & 0 & 0  & 0 & 11& 12& 4 & 4 & 2& 2  & \textbf{69} & \textbf{72}& 7 & 8\\
&  50  & 1 & 2 & 1 & 0 & 1 & 0 & 0 & 0 & 1  & 0 &  0&  0& 4 & 4 & 1& 1  & \textbf{83} & \textbf{84}& 8 & 9\\
&  100 & 1 & 2 & 0 & 0 & 0 & 0 & 3 & 0 & 0  & 0 &  0&  0& 1 & 1 & 0& 1  & \textbf{90} & \textbf{90}& 5 & 6\\
&  200 & 0 & 0 & 0 & 0 & 0 & 0 & 0 & 0 & 1  & 0 &  0&  0& 1 & 1 & 0& 0  & \textbf{97} & \textbf{97}& 1 & 2\\
\hline
\multirow{4}{*}{8}
&   25 & 0 &   0  &  1  &  1 &   6  &  7  &  2  &  0  &  2  &   0  &   3   &  3  &   8  &   9  & 6 & 6  &  \textbf{69}  &  \textbf{71}  &   3  &   3\\
&   50 & 0 &   0  &  0  &  0 &   2  &  2  &  0  &  0  &  2  &   0  &   0   &  0  &   4  &   4  & 1 & 1  &  \textbf{88}  &  \textbf{90}  &   3  &   3\\
&  100 & 0 &   0  &  0  &  0 &   3  &  3  &  1  &  0  &  0  &   0  &   0   &  0  &   3  &   3  & 0 & 0  &  \textbf{93}  &  \textbf{94}  &   0  &   0\\
&  200 & 0 &   0  &  0  &  0 &   0  &  1  &  1  &  0  &  0  &   0  &   0   &  0  &   0  &   0  & 0 & 0  &  \textbf{99}  &  \textbf{99}  &   0  &   0\\
\hline
\end{tabular} }
\end{table}
\end{landscape}
Beyond calculating the selection percentage for models DGP1--DGP8, each with no more than two covariates, we also determined the percentage of selection for models DGP9--DGP12, each having nine covariates. More precisely, we considered the following four models:
\begin{itemize}
\item
DGP9: Copula-based model with the Gaussian copula with $\tau=0.5$ and Gaussian margins with means $\disp \alpha_0+\sum_{j=1}^9 \alpha_j X_{kij}$ and a variance $1$; at each iteration, the coefficients $\alpha_0,\ldots,\alpha_9$ and the covariates are iid uniform variates on $(0,1)$.
\item
DGP10: Copula-based model with Gumbel copula with $ \tau=0.5$ and Gaussian margins with the same means an in DGP9, but with the standard deviation uniformly selected in $(3,10)$.
\item DGP11: Copula-based model with Clayton copula with $\tau=0.5$ and Poisson margins with log-means $\disp \beta_0+\sum_{j=1}^9 \beta_j X_{kij}$; at each iteration, the coefficients $\alpha_0,\ldots,\alpha_9$ are iid uniform variates on $(0,0.5)$ and $X_{k1},\ldots, X_{k9}$ are iid $N(0,1)$.
\item DGP12: Copula-based model with Frank copula with $\tau=0.5$ and Bernoulli margins with $\disp \log\left(\frac{p_{ki}}{1-p_{ki}}\right) = \alpha_0+\sum_{j=1}^9 \alpha_j X_{kij}$; at each iteration, the coefficients $\alpha_0,\ldots,\alpha_9$ are iid uniform variates on $(0,0.1)$ and $X_{ki1},\ldots, X_{ki9}$ are iid $N(0,1)$.
\end{itemize}
The corresponding GLMMs were used for comparison. For DGP9 and DGP10, we employed a Gaussian mixed linear regression with the identity link function. DGP11 was modeled using a Poisson GLMM with a log link function, and for DGP12, a GLMM with a logit link function was implemented.
 As can be seen from the results in Table \ref{tab:comp10}, even with nine covariates, we obtain results that are similar to those reported in Table \ref{tab:comp}. In particular, the copula-based model is mostly preferred to the corresponding GLMM. For DGP10 and DGP11, the best model is the copula-based model, and the percentage of success is almost 100\%, even for a small $K$. Finally, for DGP12, the correct copula-based model is always the best, with a large rate of success increasing with $K$. For the same models, we computed the RMSE for the 10 coefficients of the link function. The results, displayed in Table \ref{tab:RMSE10}, show that the real copula-based models are more precise for DGP10--DGP12. For DGP9, the RMSE of the GLMM and the Gaussian-copula model are obviously equivalent and better than the other models.

Finally, in order to address the mis-specification problem for the copula from the prediction angle, we performed two more experiments, generating observations from a given copula model, and computing the RMSE and RMSE95 for 100 new observations, according to the GLMM and five copula-based models: Clayton, Gaussian, Gumbel, Frank, and Student with 15 degrees of freedom. The cluster sizes we considered are $K\in \{25, 50, 100, 200\}$. The results for Poisson margins with ten parameters are displayed in Table \ref{tab:pred10Poisson}, while those for  Gaussian margins with 11 parameters appear in Table \ref{tab:pred10Gaussian}.
The results in Table \ref{tab:pred10Poisson} show that for all copula but the Student, the RMSE and RMSE95 are significantly better for the true copula model, while for the Student copula, the results are very close to those obtained with the other copula models. For Gaussian margins, however, for all models but the Frank copula, the best results are not obtained with the true copula model, but the RMSE and RMSE95 are quite close for each model.
 \begin{table}[ht!]
 \caption{Percentage of time each model is chosen, according to the AIC and BIC criteria over 100 replications, for models DGP9--DGP12, using GLMM 
 and copula-based models with Clayton, Gaussian, Gumbel,
 Frank, and Student(15) copulas. The scores for the theoretical model are in bold.}\label{tab:comp10}
 {\tiny
 \begin{tabular}{cccccccccccccc}
 DGP  & K & \multicolumn{2}{c}{GLMM} & \multicolumn{2}{c}{Clayton}  & \multicolumn{2}{c}{Gaussian} &  \multicolumn{2}{c}{Gumbel} &  \multicolumn{2}{c}{Frank} &  \multicolumn{2}{c}{Student(15)}\\
 && AIC & BIC  & AIC & BIC & AIC & BIC & AIC & BIC & AIC & BIC & AIC & BIC  \\
 \hline
\multirow{4}{*}{9}
&  25  &  \textbf{100} &  \textbf{100}  &   0   &   0  &  \textbf{  0}  &  \textbf{  0}   &   0   &   0   &   0   &    0    &   0   &    0  \\ &  50  &  \textbf{100} &  \textbf{100}  &   0   &   0  &  \textbf{  0}  &  \textbf{  0}   &   0   &   0   &   0   &    0    &   0   &    0  \\ &  100 &  \textbf{0  } &  \textbf{0  }  &   0   &   0  &  \textbf{100}  &  \textbf{100}   &   0   &   0   &   0   &    0    &   0   &    0  \\ &  200 &  \textbf{50 } &  \textbf{50 }  &   0   &   0  &  \textbf{ 50}  &  \textbf{ 50}   &   0   &   0   &   0   &    0    &   0   &    0  \\
\hline
\multirow{4}{*}{10}
&  25  & 2  &  2  &  0  &  0  &  0  &  0 & \textbf{ 97} & \textbf{ 97} &   0  &   0   &  1  &   1 \\
&  50  & 0  &  0  &  0  &  0  &  0  &  0 & \textbf{100} & \textbf{100} &   0  &   0   &  0  &   0\\
&  100 & 0  &  0  &  0  &  0  &  0  &  0 & \textbf{100} & \textbf{100} &   0  &   0   &  0  &   0\\
&  200 & 0  &  0  &  0  &  0  &  0  &  0 & \textbf{100} & \textbf{100} &   0  &   0   &  0  &   0\\
\hline
\multirow{4}{*}{11}
&  25  & 0 &   0 & \textbf{ 99} & \textbf{ 99} &   0  &  0  &  0  &  0  &  1  &   1  &   0   &  0\\
&  50  & 0 &   0 & \textbf{100} & \textbf{100} &   0  &  0  &  0  &  0  &  0  &   0  &   0   &  0\\
&  100 & 0 &   0 & \textbf{100} & \textbf{100} &   0  &  0  &  0  &  0  &  0  &   0  &   0   &  0\\
&  200 & 0 &   0 & \textbf{100} & \textbf{100} &   0  &  0  &  0  &  0  &  0  &   0  &   0   &  0\\
\hline
\multirow{4}{*}{12}
&  25  &  9  &  9  &  9  &  9  &  0  &  0  & 17 &  17 & \textbf{ 65}  & \textbf{ 65}   &  0   &  0  \\ &  50  &  6  &  6  &  0  &  0  &  0  &  0  & 16 &  16 & \textbf{ 78}  & \textbf{ 78}   &  0   &  0 \\
&  100 &  2  &  2  &  1  &  1  &  0  &  0  &  2 &   2 & \textbf{ 95}  & \textbf{ 95}   &  0   &  0 \\
&  200 &  0  &  0  &  0  &  0  &  0  &  0  &  0 &   0 & \textbf{100}  & \textbf{100}   &  0   &  0 \\
\hline
\end{tabular}
}
\end{table}

 \begin{table}[ht!]
 \caption{
 RMSE for the 10 parameters of the mean over 100 replications, for models DGP9--DGP12, using GLMM and copula-based models 
 with Clayton, Gaussian, Gumbel,  Frank, and Student(15) copulas. The scores for the theoretical model are in bold.}\label{tab:RMSE10}
 \begin{tabular}{cccccccc}
 DGP  & K & {GLMM} & {Clayton}  & {Gaussian} &  {Gumbel} &  {Frank} &  {Student(15)}\\
  \hline
\multirow{4}{*}{9}
&  25  &  \textbf{0.2731} & 0.2952 & \textbf{0.2708} & 0.3364 & 0.2587 & 0.2647  \\
&  50  &  \textbf{0.2488} & 0.3878 & \textbf{0.2472} & 0.3665 & 0.2679 & 0.2594  \\
&  100 &  \textbf{0.1827} & 0.3347 & \textbf{0.1703} & 0.3108 & 0.2433 & 0.1789  \\
&  200 &  \textbf{0.0987} & 0.2099 & \textbf{0.0989} & 0.2590 & 0.1402 & 0.1164  \\
\hline
\multirow{4}{*}{10}
&  25  &  2.3456 & 2.9170&  2.2735 & \textbf{1.9189}&  3.2548&  2.8671 \\
&  50  &  1.5162 & 2.3487&  1.4772 & \textbf{1.3489}&  2.9099&  1.6031\\
&  100 &  1.1422 & 2.1906&  1.2355 & \textbf{0.9984}&  3.0118&  1.2628\\
&  200 &  0.7470 & 2.0069&  0.7931 & \textbf{0.6865}&  2.8120&  0.9766\\
\hline
\multirow{4}{*}{11}
&  25  &0.2295 & \textbf{0.0955} & 0.1356&  0.2325&  0.2789 & 0.1312\\
&  50  &0.1863 & \textbf{0.0657} & 0.1031&  0.2244&  0.2578 & 0.1052\\
&  100 &0.1629 & \textbf{0.0492} & 0.0692&  0.2119&  0.2587 & 0.0725\\
&  200 &0.1494 & \textbf{0.0341} & 0.0507&  0.2056&  0.2471 & 0.0616\\
\hline
\multirow{4}{*}{12}
&  25  & 0.4896 & 0.3289 & 0.3331 & 0.3152 & \textbf{0.2751} & 0.3297  \\
&  50  & 0.3844 & 0.2982 & 0.2589 & 0.2415 & \textbf{0.2005} & 0.2631 \\
&  100 & 0.2623 & 0.2353 & 0.1673 & 0.1680 & \textbf{0.1358} & 0.1695 \\
&  200 & 0.2115 & 0.2213 & 0.1234 & 0.1308 & \textbf{0.0946} & 0.1252 \\
\hline
\end{tabular}
\end{table}

\begin{landscape}
\begin{table}[ht!]
\caption{RMSE  and RMSE95 for predictions of copula-based models with Poisson margins with 10 parameters (similar to DGP11), using GLMM and
 copula-based models with Clayton, Gaussian, Gumbel,  Frank, and Student(15) copulas. The scores for the theoretical model are in bold.}\label{tab:pred10Poisson}
\resizebox{19cm}{!}{
 \begin{tabular}{cccccccccccccc}
 True    & K & \multicolumn{6}{c}{RMSE} & \multicolumn{6}{c}{RMSE95}  \\
 Copula  & & GLMM & Clayton & Gaussian & Gumbel &  Frank & Student(15)& GLMM & Clayton & Gaussian & Gumbel &  Frank & Student(15)\\
  \hline
\multirow{4}{*}{Clayton}
&  25  &0.1217 & \textbf{0.1093} &0.1387 & 0.1106 & 0.1103 & 0.1234 & 0.3599 & \textbf{0.2859} & 0.3568 & 0.2931 & 0.2953 & 0.3194 \\
&  50  &0.1215 & \textbf{0.1095} &0.1404 & 0.1111 & 0.1106 & 0.1274 & 0.3545 & \textbf{0.2848} & 0.3575 & 0.2939 & 0.2940 & 0.3235 \\
&  100 &0.1202 & \textbf{0.1090} &0.1400 & 0.1102 & 0.1100 & 0.1348 & 0.3471 & \textbf{0.2824} & 0.3549 & 0.2896 & 0.2912 & 0.3374 \\
&  200 &0.1216 & \textbf{0.1100} &0.1417 & 0.1109 & 0.1108 & 0.1386 & 0.3530 & \textbf{0.2853} & 0.3592 & 0.2912 & 0.2943 & 0.3453 \\
\hline
\multirow{4}{*}{Gaussian}
&  25   & 0.1206 & 0.1067 & \textbf{0.1046} & 0.1048 &0.1052& 0.1046& 0.3645 & 0.2650 & \textbf{0.2535} & 0.2541 & 0.2572 & 0.2536\\
&  50   & 0.1188 & 0.1075 & \textbf{0.1048} & 0.1052 &0.1057& 0.1048& 0.3433 & 0.2661 & \textbf{0.2501} & 0.2523 & 0.2555 & 0.2504\\
&  100  & 0.1187 & 0.1071 & \textbf{0.1041} & 0.1044 &0.1052& 0.1041& 0.3445 & 0.2625 & \textbf{0.2459} & 0.2470 & 0.2524 & 0.2461\\
&  200  & 0.1201 & 0.1077 & \textbf{0.1049} & 0.1052 &0.1059& 0.1050& 0.3523 & 0.2653 & \textbf{0.2498} & 0.2509 & 0.2554 & 0.2500\\
\hline
\multirow{4}{*}{Gumbel}
&  25  & 0.1183 & 0.1036 & 0.1385 & \textbf{0.0991} & 0.0999 & 0.1125 & 0.3651 & 0.2498 & 0.3510 & \textbf{0.2233} & 0.2267 & 0.2800\\
&  50  & 0.1166 & 0.1052 & 0.1405 & \textbf{0.0998} & 0.1010 & 0.1127 & 0.3431 & 0.2536 & 0.3530 & \textbf{0.2213} & 0.2274 & 0.2738\\
&  100 & 0.1169 & 0.1051 & 0.1398 & \textbf{0.0991} & 0.1005 & 0.1093 & 0.3484 & 0.2521 & 0.3519 & \textbf{0.2174} & 0.2257 & 0.2595\\
&  200 & 0.1184 & 0.1056 & 0.1413 & \textbf{0.0999} & 0.1013 & 0.1098 & 0.3550 & 0.2543 & 0.3559 & \textbf{0.2213} & 0.2285 & 0.2605\\
\hline
\multirow{4}{*}{Frank}
&  25  & 0.1229 & 0.1092 & 0.1394 & 0.1087 & \textbf{0.1083} & 0.1313 & 0.3695 & 0.2784 & 0.3590& 0.2766 & \textbf{0.2734} & 0.3365\\
&  50  & 0.1214 & 0.1096 & 0.1405 & 0.1091 & \textbf{0.1086} & 0.1364 & 0.3528 & 0.2769 & 0.3572& 0.2749 & \textbf{0.2708} & 0.3462\\
&  100 & 0.1206 & 0.1087 & 0.1398 & 0.1082 & \textbf{0.1079} & 0.1387 & 0.3484 & 0.2717 & 0.3525& 0.2689 & \textbf{0.2668} & 0.3514\\
&  200 & 0.1221 & 0.1096 & 0.1416 & 0.1091 & \textbf{0.1087} & 0.1410 & 0.3561 & 0.2755 & 0.3595& 0.2729 & \textbf{0.2708} & 0.3602\\
\hline
 \multirow{4}{*}{Student(15)}
&  25  & 0.1404 & 0.1403 & 0.1402 & 0.1404 & 0.1404 & \textbf{0.1402} & 0.3638 & 0.3623 & 0.3633 & 0.3642 & 0.3636 & \textbf{0.3639}\\
&  50  & 0.1410 & 0.1410 & 0.1409 & 0.1409 & 0.1410 & \textbf{0.1409} & 0.3598 & 0.3583 & 0.3591 & 0.3599 & 0.3592 & \textbf{0.3598}\\
&  100 & 0.1406 & 0.1406 & 0.1405 & 0.1405 & 0.1405 & \textbf{0.1405} & 0.3567 & 0.3557 & 0.3565 & 0.3573 & 0.3565 & \textbf{0.3572}\\
&  200 & 0.1412 & 0.1414 & 0.1414 & 0.1414 & 0.1414 & \textbf{0.1414} & 0.3564 & 0.3577 & 0.3585 & 0.3593 & 0.3585 & \textbf{0.3592}\\
\hline
\end{tabular}
}
\end{table}
\end{landscape}

\begin{landscape}
\begin{table}[ht!]
\caption{
RMSE and RMSE95 for predictions of copula-based models with Gaussian margins with 11 parameters (similar to DGP12), using GLMM and
 copula-based models with Clayton, Gaussian, Gumbel,  Frank, and Student(15) copulas. The scores for the theoretical model are in bold.}\label{tab:pred10Gaussian}
\resizebox{19cm}{!}{
 \begin{tabular}{cccccccccccccc}
 True    & K & \multicolumn{6}{c}{RMSE} & \multicolumn{6}{c}{RMSE95}  \\
Copula  & & GLMM & Clayton & Gaussian & Gumbel &  Frank & Student(15)& GLMM & Clayton & Gaussian & Gumbel &  Frank & Student(15)\\
 \hline
\multirow{4}{*}{Clayton}
&  25    &  0.4660 & \textbf{0.5201} & 0.4776 & 0.4792 & 0.4674 & 0.5349 & 1.0597 &\textbf{1.0564} & 1.0684 & 1.0735 & 1.0664 & 1.0605\\
&  50    &  0.4660 & \textbf{0.5199} & 0.4825 & 0.4757 & 0.4665 & 0.5238 & 1.0639 &\textbf{1.0605} & 1.0689 & 1.0786 & 1.0726 & 1.0651\\
&  100   &  0.4641 & \textbf{0.4997} & 0.4784 & 0.4662 & 0.4651 & 0.4948 & 1.0667 &\textbf{1.0646} & 1.0827 & 1.0824 & 1.0664 & 1.0685\\
&  200   &  0.4624 & \textbf{0.4974} & 0.4844 & 0.4690 & 0.4632 & 0.4909 & 1.0632 &\textbf{1.0611} & 1.0923 & 1.0792 & 1.0608 & 1.0651\\
\hline
\multirow{4}{*}{Gaussian}
&  25   &  0.4762 & 0.5028 &  \textbf{0.4817} & 0.4815 & 0.4780 & 0.4899 & 0.7909 & 0.8774 & \textbf{0.8110} & 0.7949 & 0.8037 & 0.8719\\
&  50   &  0.4759 & 0.5017 &  \textbf{0.4884} & 0.4774 & 0.4782 & 0.4830 & 0.7897 & 0.8830 & \textbf{0.8198} & 0.7934 & 0.8038 & 0.8102\\
&  100  &  0.4746 & 0.4979 &  \textbf{0.4823} & 0.4759 & 0.4777 & 0.4795 & 0.7879 & 0.9011 & \textbf{0.8146} & 0.7939 & 0.8093 & 0.8183\\
&  200  &  0.4735 & 0.5029 &  \textbf{0.4887} & 0.4748 & 0.4770 & 0.4781 & 0.7842 & 0.9014 & \textbf{0.8373} & 0.7882 & 0.8066 & 0.8042\\
\hline
\multirow{4}{*}{Gumbel}
&  25  & 0.4740 & 0.5021 & 0.4860 & \textbf{0.4793} & 0.4748 & 0.5055 & 0.5928 & 0.7447 & 0.7305 & \textbf{0.6585} & 0.5788 & 0.9502\\
&  50  & 0.4735 & 0.5112 & 0.4875 & \textbf{0.4767} & 0.4746 & 0.4991 & 0.5809 & 0.7559 & 0.6937 & \textbf{0.6143} & 0.5628 & 0.8912\\
&  100 & 0.4724 & 0.5054 & 0.4884 & \textbf{0.4789} & 0.4739 & 0.4973 & 0.5804 & 0.8032 & 0.6940 & \textbf{0.6683} & 0.5642 & 0.8999\\
&  200 & 0.4713 & 0.5025 & 0.4862 & \textbf{0.4761} & 0.4729 & 0.4844 & 0.5681 & 0.8110 & 0.6673 & \textbf{0.6366} & 0.5475 & 0.7554\\
\hline
\multirow{4}{*}{Frank}
&  25  & 0.4969 & 0.5122 & 0.4975 & 0.4969 & \textbf{0.4959} & 0.4966 & 0.9526 & 0.9842 & 0.9536 & 0.9577 & \textbf{0.9481} & 0.9515\\
&  50  & 0.4973 & 0.5112 & 0.4987 & 0.4972 & \textbf{0.4962} & 0.4969 & 0.9567 & 0.9866 & 0.9600 & 0.9617 & \textbf{0.9518} & 0.9552\\
&  100 & 0.4955 & 0.5063 & 0.5099 & 0.4953 & \textbf{0.4945} & 0.4952 & 0.9560 & 0.9891 & 0.9944 & 0.9616 & \textbf{0.9517} & 0.9549\\
&  200 & 0.4943 & 0.5073 & 0.5205 & 0.4942 & \textbf{0.4933} & 0.4940 & 0.9518 & 0.9865 & 1.0202 & 0.9572 & \textbf{0.9472} & 0.9506\\
\hline
 \multirow{4}{*}{Student(15)}
&  25  & 0.6628 & 0.6625 & 0.6624 & 0.6628 & 0.6630 & \textbf{0.6623} & 1.3651 & 1.3603 & 1.3649 & 1.3678 & 1.3655 & \textbf{1.3647}\\
&  50  & 0.6627 & 0.6626 & 0.6625 & 0.6627 & 0.6628 & \textbf{0.6625} & 1.3675 & 1.3669 & 1.3674 & 1.3702 & 1.3676 & \textbf{1.3674}\\
&  100 & 0.6604 & 0.6604 & 0.6603 & 0.6604 & 0.6605 & \textbf{0.6603} & 1.3766 & 1.3726 & 1.3765 & 1.3763 & 1.3766 & \textbf{1.3765}\\
&  200 & 0.6592 & 0.6593 & 0.6592 & 0.6592 & 0.6592 & \textbf{0.6592} & 1.3717 & 1.3678 & 1.3717 & 1.3750 & 1.3718 & \textbf{1.3717}\\
\hline
\end{tabular}
}
\end{table}
\end{landscape}

As a general conclusion for these numerical experiments, the proposed methodology works fine, but can depend on the starting points. However, from the computation of the gradients, we can check that the algorithm converges, and since the selection of the initial values of the margin parameters might be important, we propose to start from estimated values obtained by GLM or GAM, as if there were no clusters. This procedure as also been implemented in our R package \textit{CopulaGAMM}. As shown by the results of the last two experiments displayed in Tables \ref{tab:pred10Poisson}--\ref{tab:pred10Gaussian}, the choice of the copula is more important for small or large quantiles of the predicted values.

\section{Case study: COVID-19 vaccine hesitancy survey data from 23 countries}\label{sec:ex}
As an illustrative application, we analyze survey data on COVID-19 vaccine hesitancy collected from 23 countries. \citep{Lazarus/Wyka/White/Picchio/Rabin/Ratzan/ParsonsLeigh/Hu/El-Mohandes:2022}. Here, the responses of the individuals are the observations and each country represent the clusters. These data have been recently analyzed  by \cite{Awasthi/Nagori/Nasri:2023} using linear mixed models combined with a Bayesian network. The 23 countries considered in the study are: Brazil, Canada, China, Ecuador, France, Germany, Ghana, India, Italy, Kenya, Mexico,  Nigeria,  Peru, Poland, Russia, Singapore,  South Africa, South Korea,  Spain,  Sweden, T\"{u}rkiye, UK, and USA. This study aims to comprehend the main factors contributing to vaccine hesitancy and how these factors vary across the studied countries. Our variable of interest consists of responses to two questions "I will take the COVID-19 vaccine when it is available to me" (Q8 in the survey data) or "Have you received at least one dose of a COVID-19 vaccine?" (Q7 in the survey data). Q8 has four possible answers: Y=0 (Strongly disagree), Y=1 (Somewhat disagree or Unsure/No opinion), Y=2 (Somewhat agree), and Y=3 (Strongly agree or answered yes to Q7). We then defined a new variable of interest, namely  $Y_1 = \I(Y\ge 2)$, i.e., $Y_1=0$ if the person disagrees with taking the vaccine when available, and $Y_1=1$ if the person agrees.

Furthermore, we incorporated five covariates derived as the averages of scores from questions pertaining to 'Perception of risk' (Q1-Q4), 'Trust' (Q5, Q6, Q24-Q26), 'Restriction measures' (Q12-Q17), 'Financial status' (Q18, Q54, corresponding to 'Loss of Income' and 'Monthly House Income'), and 'age' (Q27). Our dataset contains 23135 observations. We considered a copula-based model with a binomial margin for $Y_1$ and the five covariates. We determined the optimal copula-based model by selecting the one with the lowest BIC from a set of five copula families (Frank, Gumbel, Clayton, Gaussian, Student) in combination with various covariates. The selection process remained consistent when employing the AIC criterion, leading to the same results.

We found that the Frank copula, incorporating covariates $X_1$ (Restriction measures score),  $X_2$ (Trust score), and $X_3$ (Perception of risk score) best fits the data, all taking values in $[0,3]$. It is noteworthy that ``Financial status'' and ``Age'' are not significant predictors for vaccine hesitancy.
The estimated Kendall's tau for Frank copula is $-0.178$ (corresponding to  a parameter $\beta_1 =-1.648949$), while the equation for the probability of the margin is the logit function evaluated at $-4.425428 + 1.447811 x_1 +1.026878 x_2 + 1.44122 x_3$, which means that the estimated probability of agreeing to be vaccinated is an increasing function of these three scores. This is illustrated in Figure \ref{fig:UsvsSouthAfrica}, where the estimated probability of agreeing to be vaccinated, given values of the three covariates, are displayed for South Africa  ($\hat V = 0.9506$), France ($\hat V = 0.5674$), India ($\hat V = 0.2428$), and  Spain ($\hat V =0.0481$), where the values of the estimated $V$  by country are displayed in Table \ref{tab:Vbin}. From Figure \ref{fig:UsvsSouthAfrica}, we can see that when the level of trust and the perception of risk are both high, the probability of being vaccinated is almost independent of the level of the restriction measures, especially if the estimation of the latent variable is large. However, this factor becomes more important when the combination of trust and the perception of risk is sufficiently high. This observation holds true across the studied countries with some variation in the estimated probability of agreeing to be vaccinated which depends on the latent variable $V$. It is worth noting that for Frank copula, the conditional expectation is a decreasing (resp. increasing) function of $V$ when Kendall's tau is negative (resp. positive). This explains the results displayed in Figure  \ref{fig:UsvsSouthAfrica}.
Here, we used the median of the conditional density given by Equation \eqref{eq:Vfivenobs} to estimate $V$ for each country, as implemented in the CopulaGAMM package \citep{Krupskii/Nasri/Remillard:2023b}. It is essential to note that latent variables, being introduced into the model through a copula, are challenging to interpret. However, in the present case, their values can be employed for cross-country comparisons.

Finally, for this example, we also computed the Logistic model with random effect, i.e., a GLMM with Bernoulli margins and logit link. This model produced very similar predictions, as displayed in Figure \ref{fig:compvaccine} of Appendix \ref{app:compvaccine}. In this case, it is worth noting that the logit function is $   -4.515  +  1.498 x_1  +    1.058 x_2      +  1.491 x_3 $.
Furthermore, for the GLMM model, the AIC is 9640.9 
and the BIC is 9681.1, while for the copula-based model, the AIC is 9635.1 and the BIC is 9675.3. 
Therefore, opting for the model with the lowest AIC or BIC prompts the preference for the copula-based model over the GLMM. Furthermore, in terms of predictions, the RMSE across all observations is $8.5\times 10^{-7}$ for the copula-based model, while it is $1.1\times 10^{-5}$ for the GLMM.

\begin{figure}
    \centering
   \includegraphics[width=2in, height=2.1in]{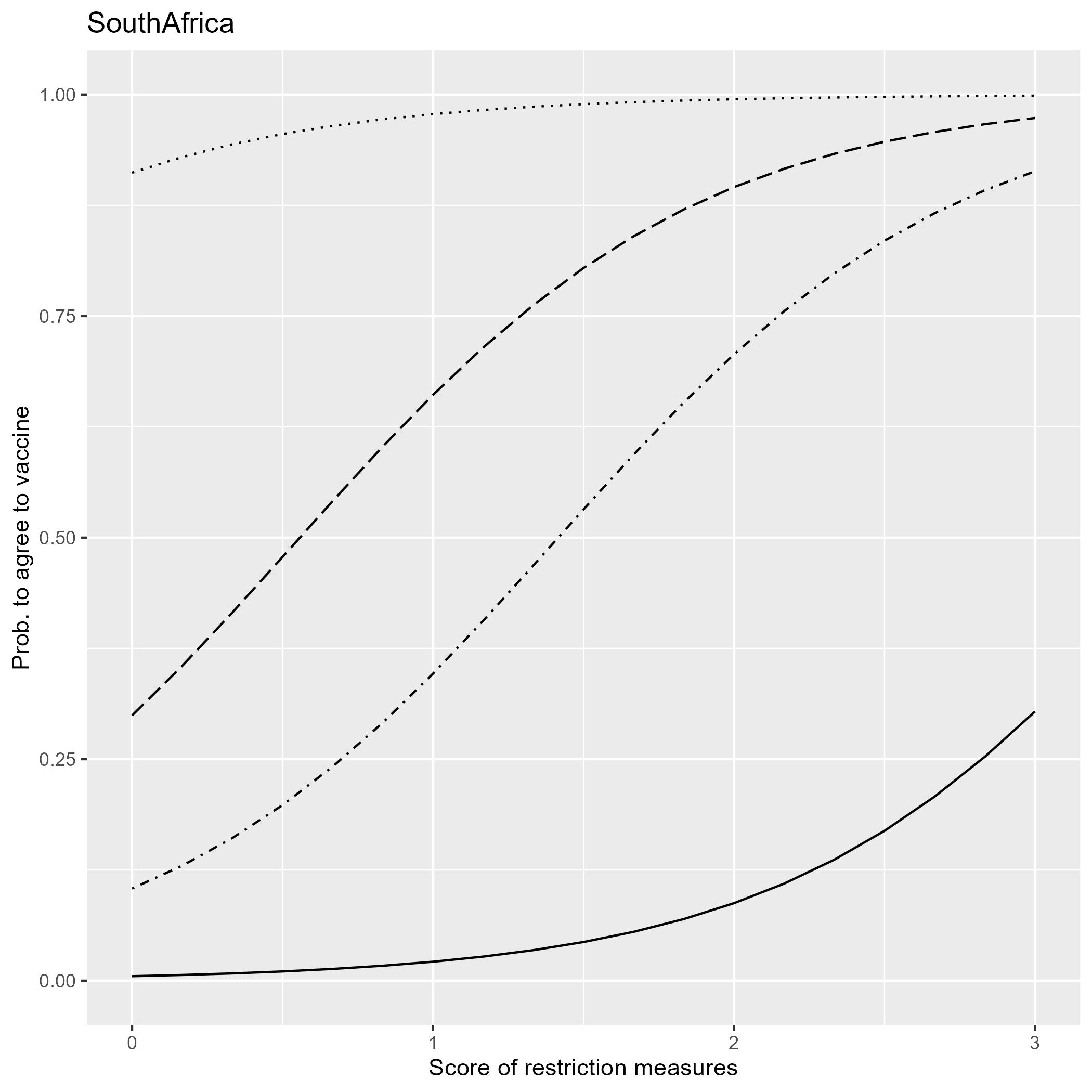}
   \includegraphics[width=2in, height=2.1in]{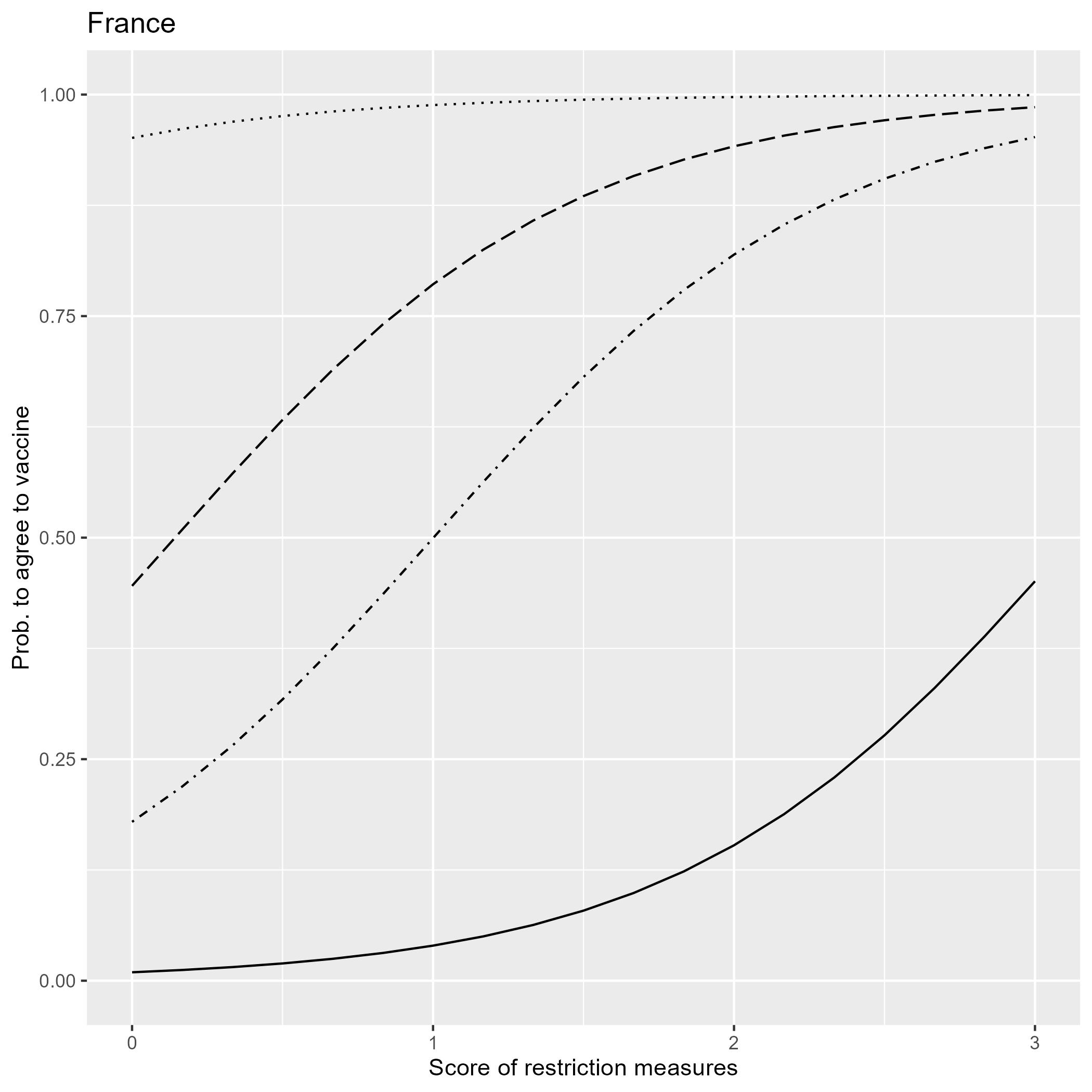}
   \includegraphics[width=2in, height=2.1in]{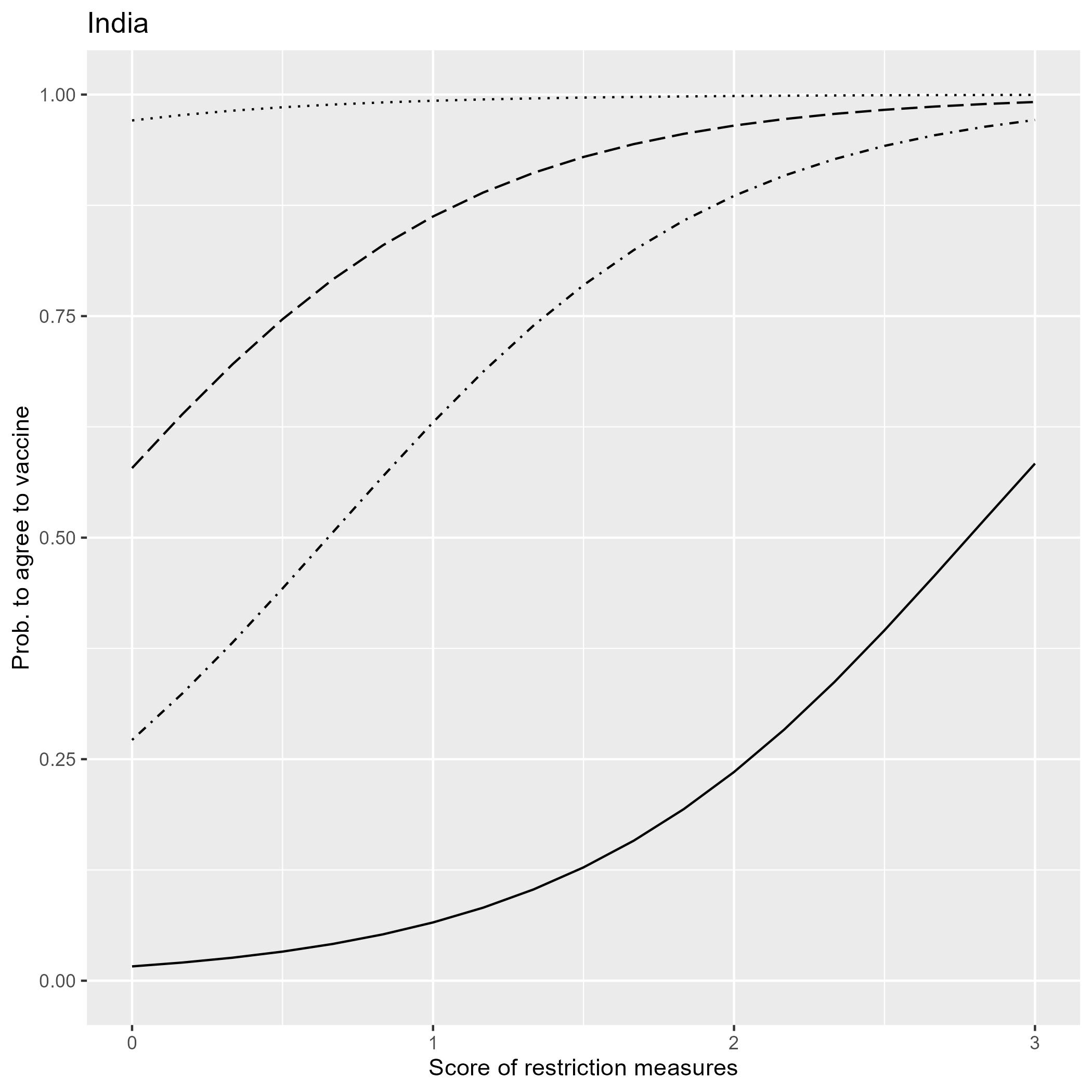}
   \includegraphics[width=2in, height=2.1in]{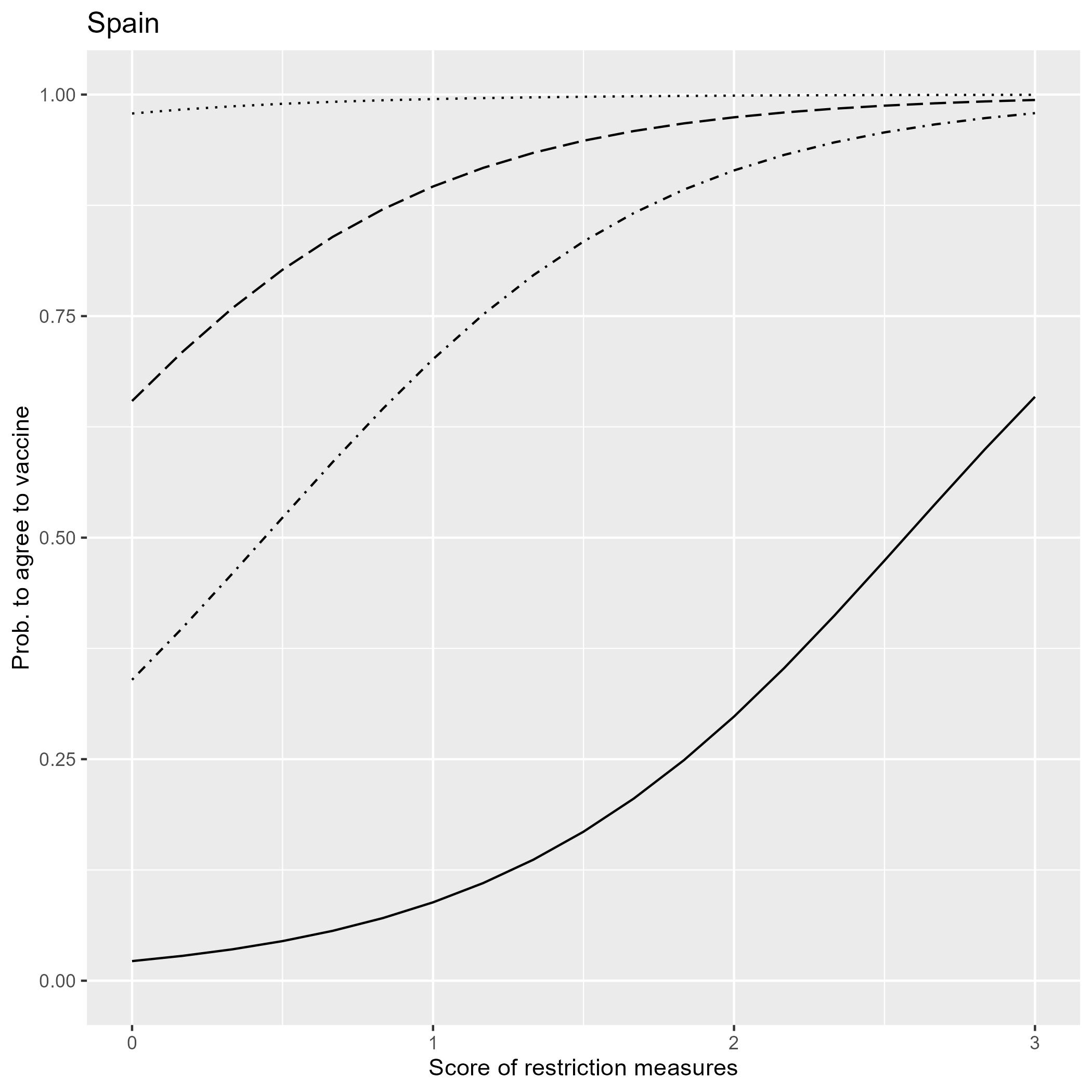}
    \caption{Estimated probabilities to agree to be vaccinated by scores of restriction measures ($X_1$) by values of trust ($X_2$) and perception of risk ($X_3$) for South Africa  ($\hat V = 0.9506$), France ($\hat V = 0.5674$), India ($\hat V = 0.2428$), and  Spain ($\hat V =0.0481$). The line types are ``solid'' for $(X_2,X_3)=(0,0)$, ``dotdash'' for $(X_2,X_3)=(3,0)$, ``longdash'' for $(X_2,X_3)=(0,3)$, and ``dotted'' for $(X_2,X_3)=(3,3)$.}
    \label{fig:UsvsSouthAfrica}
\end{figure}

\begin{table}[!ht]
\caption{Estimated value of the latent variable by country.}\label{tab:Vbin}
    \centering

    {\small
 \begin{tabular}{|c|c|c|c|c|c|c|}
    \hline
Country   & South Africa  & Kenya    & Ghana   & Nigeria & Russia  & Ecuador       \\
$1-V$     & 0.0494        & 0.0760   & 0.0786  & 0.0816  & 0.1041  & 0.1468           \\
\hline
Country   & South Korea & T\"{u}rkiye & Peru   & France & Poland & Sweden     \\
$1-V$     & 0.2025      & 0.3679           & 0.3869 & 0.4326 & 0.5228 & 0.6338     \\
\hline
Country   & India    & Mexico & Brazil & USA     & China  & Canada               \\
$1-V$     & 0.7572   & 0.7795 & 0.7820 & 0.7848  & 0.7900 & 0.8534              \\
\hline
Country       & Italy           & Singapore        & Germany  & UK      & Spain  &    \\
$1-V$         & 0.8892          &  0.8907          & 0.8954   & 0.8992 & 0.9519  &   \\
\hline
    \end{tabular}
    }
\end{table}

\section{Conclusion}
We demonstrate how factor copulas can be used to model the dependence for clustered data. Factor copulas provide a framework for predicting a continuous or a discrete response variable using a link function of covariates that are used to model the margins and the copula parameters. Our proposed models cover GAMMs and they generalize several existing copula-based approaches proposed in the literature. We also demonstrate the convergence of the estimators, a process that depends on the size and the number of clusters in the model.
In comparison to a traditional GAMM, our proposed models exhibit favorable performance. All computations proposed in this paper are implemented in the R package \textit{CopulaGAMM}, which is available through CRAN. This package facilitates the estimation of the conditional mean and quantiles for the proposed model using ten copula families and their rotations. It also supports various the common discrete and continuous margins for the variable of interest.

\appendix

\section{Recovering Mixed linear Gaussian model using the factor copula-based mixed model}\label{app:gaussian}

The  mixed linear Gaussian model defined by \eqref{eq:mixed-reg1}, with $\epsilon_{ki}\sim N(0,\sigma_\epsilon^2)$  and $\eta_k\sim N(0,\sigma_\eta^2)$,  corresponds to $\eta_k = \rho\sigma \Phi^{-1}(V_k)$, with a Gaussian cdf $G_{\bh(\balpha,\bx)}$ with mean  $h_1(\balpha,\bx) = \bb^\top \bx = \sum_{j=1}^d \alpha_{j}x_j$, a standard deviation of $\sigma = h_2(\balpha,\bx)= \sqrt{\sigma_\eta^2+\sigma_\epsilon^2} = \alpha_{d+1}$, and a normal copula with parameter $\rho = \frac{\sigma_\eta}{\sigma}= \bphi(\btheta,\bx) =  \tanh(\beta_1)= \left(1-e^{-2\beta_1}\right)\Big{/}\left(1+e^{-2\beta_1}\right)$, implying that $\bphi$ does not depend on $\balpha$. Here, the parameters are $(\bb,\sigma,\rho)$, which leads to $\sigma_\eta = \rho \sigma$ and $\sigma_\epsilon^2 = \sigma^2\left(1-\rho^2\right)$. As a result, setting with $\bar e_k = \bar y_k - \bb^\top \bar \bx_k$ and
$\lambda_k = \dfrac{n_k \rho^2}{1-\rho^2+n_k\rho^2}$, $k\in\setK$, the cdf of the posterior of $\eta_k$ given the observation in cluster $k$ in the linear mixed model is  Gaussian with mean $\lambda_k\bar e_k$ and variance $\sigma^2_\eta(1-\lambda_k)$, while in the copula-based model, the cdf of the posterior of $V_k$ given the observation in cluster $k$ is
$\Phi\left( \frac{ \Phi^{-1}(v) - \lambda_k \bar e_k /\sigma_\eta}{ \sqrt{1-\lambda_k}} \right)$.  Therefore, the posterior median $m_k$ of $V_k$ satisfies $\Phi^{-1}(m_k) = \dfrac{\lambda_k  \bar e_k}{\sigma_\eta}$.
Since the conditional distribution of $Y_{ki}$ given $V_k=v$ and $\bX_{ki}=\bx_{ki}$ is Gaussian with a mean of $\bb^\top \bx_{ki}+\sigma_\eta\Phi^{-1}(v)$ and a standard deviation of $\sigma_\epsilon$, it follows that the prediction of $Y_{ki}$ is
$\hat Y_{ki} = \bb^\top \bx_{ki}+ \lambda_k \bar e_k$,
which is an identical prediction as in the mixed linear model.
As for the popular model of mixed regression with random intercept and random slopes given by
\begin{equation}\label{eq:mixed-reg2}
Y_{ki} = \bb ^\top \bX_{ki}+ \bbeta_k^\top \bZ_{ki}+\ve_{ki},
\end{equation}
where $\bbeta_k$ is a centered Gaussian vector with covariance matrix $\Sigma_\bbeta$, independent of $\ve_{ki}$, and the first column of $\bZ_{ki}$ is identically $1$, it corresponds to the copula-based model with $\bbeta_k^\top \bz =\left(\bz^\top \Sigma_\bbeta \bz\right)^{1/2} \Phi^{-1}(V_k)$, where  $V_k$ is the associated latent variable, the cdf of $Y_{ki}$ given $\bX_{ki}=\bx,\bZ_{ki}=\bz$ is Gaussian with a mean of $h_1(\balpha,\bx) = \bb^\top \bx$, a standard deviation of $h_2(\balpha,\bz) = \left(\sigma_\ve^2 +\bz^\top \Sigma_\bbeta \bz \right)^{1/2} = \left(\bz^\top \Sigma \bz \right)^{1/2}$, with $\Sigma = \Sigma_\bbeta+\sigma_\epsilon^2 \be_1\be_1^\top$,
and the copula is a normal copula with parameter
$\rho(\bz) = \bphi(\btheta,\bz)  = \left(\dfrac{\bz^\top \Sigma_\bbeta \bz}{\sigma_\ve^2+ \bz^\top \Sigma_\bbeta \bz}\right)^{1/2}$. One can see that $\balpha$ is used to estimate $b$ and $\Sigma$, while $\bbbeta=\beta_1$ is used to estimate $ r = \tanh(\beta_1) =\dfrac{\sigma_\bbeta}{\sigma}$, viz. $\rho^2(\bz) = 1- \dfrac{\sigma^2\left(1-r^2\right)}{h_2^2(\balpha,\bz)} $, where $\sigma_\bbeta^2 = (\Sigma_\bbeta)_{11}$,
$\sigma^2 = \sigma_\bbeta^2+\sigma_\epsilon^2 = \Sigma_{11}$. In this case, it is observed that the copula parameter really depends on both $\balpha$ and $\bbbeta$, not just $\bbbeta$. In addition, one gets $\sigma_\bbeta=r\sigma$, $\sigma_\epsilon^2=\sigma^2\left(1-r^2\right)$, implying that $\Sigma_\bbeta = \Sigma - \sigma_\epsilon^2 \be_1\be_1^\top $.
Note that for the mixed regression model described by Equation \eqref{eq:mixed-reg2}, the conditional distribution of $\bbeta_k$ given $(y_{ki},\bx_{ki},\bz_{ki})$, $i\in \{1,\ldots, n_k\}$, is a Gaussian distribution with a mean of
$\mu_k = A_k \dfrac{1}{\sigma_\epsilon^2}\sum_{i=1}^{n_k} z_{ki}(y_{ki}-b^\top x_{ki})$ and covariance matrix $A_k = \left(\disp \Sigma_\bbeta^{-1}+ \dfrac{1}{\sigma_\epsilon^2}\sum_{i=1}^{n_k} z_{ki}z_{ki}^\top \right)^{-1}$.

\section{Graphical illustrations of the difference between GAMM and factor copula-based approach}\label{difference}

On the one hand, we can consider a logistic model with one covariate, where the link prediction of the GAMM model satisfies $\disp
\log\left\{\frac{P(Y=1|X=x,\eta)}{P(Y=0|X=x,\eta)}\right\} = s(x)+\eta $, where $s(x)$ is a spline function.
 In this case, for different values of $\eta$, the link predictions are random translations of each other.
On the other hand, the link prediction of the copula-based model satisfies
 $\disp
 \log\left\{\frac{P(Y=1|X=x,V=v)}{P(Y=0|X=x,V=v)}\right\} =
 \log\left[\frac{1-\cC\{1-p(x),v\}}{\cC\{1-p(x),v\}}\right]$,
with $p(x) = P(Y=1|X=x)$, $\log\left\{\frac{p(x)}{1-p(x)}\right\} = s(x)$, and
$\cC(u,v) = \partial_v C(u,v)$, where $C$ is the associated copula. In Figure \ref{fig:pred}, the link predictions for $V\in \{0.1,0.5.0.9\}$, for a normal copula with parameter $\rho=0.9$ and a non-central Gaussian copula \citep{Nasri:2020} with parameters $\rho=0.9$, $a_1=0.1$, $a_2=2.9$ are displayed for the linear case $s_1(x) = -1+0.5x$ and the non-linear case $s_2(x) =  \min(x,1)+2\max(0,x-1)$. For $s_1$, we can see that for each copula, the three curves do not have constant slopes and they are not parallel. For $s_2$, we can see that the three curves are not parallel. Additionally, changing the copula family (Gaussian copula vs non-central squared copula) modifies the shape of the curves.

\begin{figure}[ht!]
    \centering
   a) \includegraphics[scale=0.45]{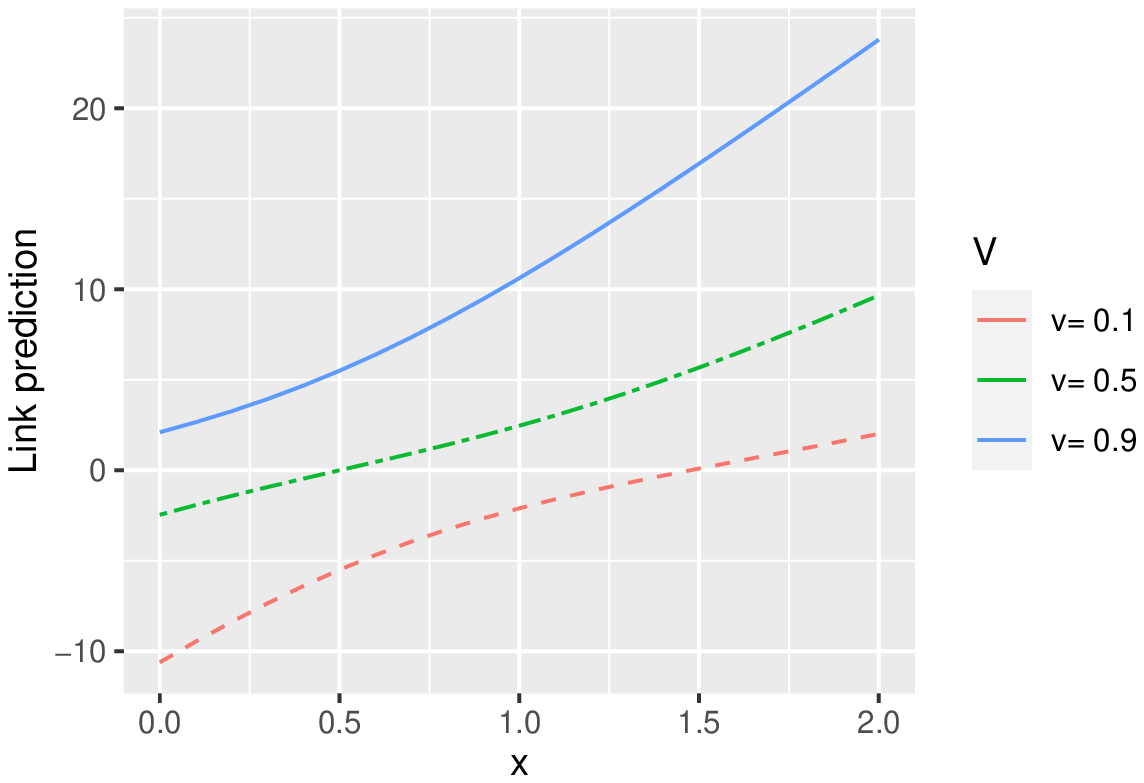}
   b) \includegraphics[scale=0.45]{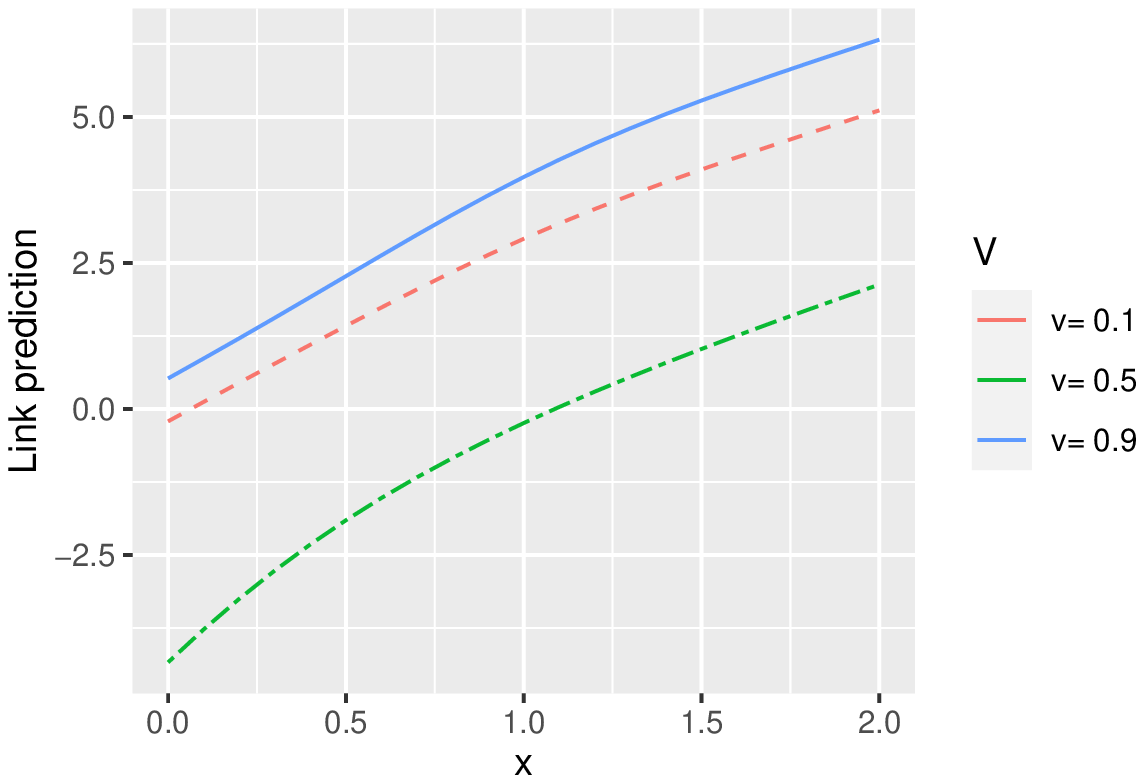}

   c) \includegraphics[scale=0.45]{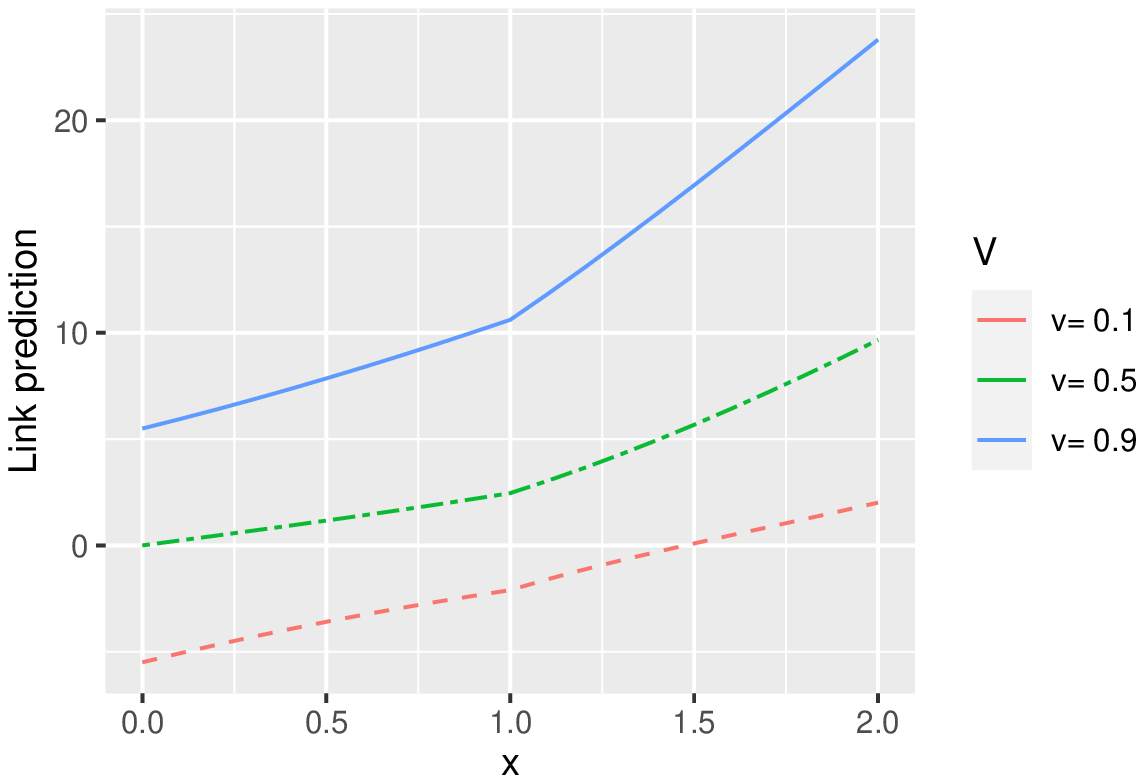}
   d) \includegraphics[scale=0.45]{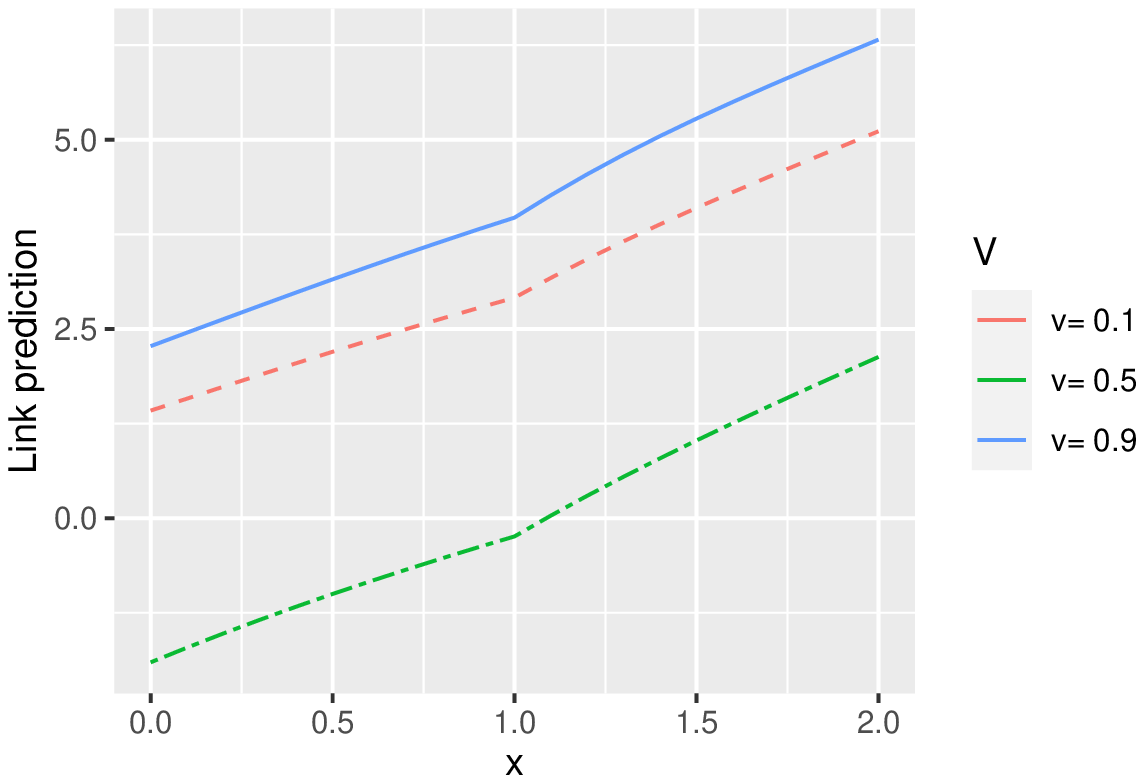}
     \caption{Graphs of the link prediction for $s_1(x) = \beta_0 +\beta_1 x = -1+0.5x$ (panels a,b) and $s_2(x) = \min(x,1)+2\max(0,x-1) $(panels c, d), for $v=0.1$ (dashed red line), $v=0.5$ (two-dashed green line) and $v=0.9$ (solid blue line). The left panels (a,c) are computed with the Gaussian copula, while the right panels (b,d) are computed with the non-central Gaussian copula  with parameters $\rho=0.9$, $a_1=0.1$, $a_2=2.9$. }
    \label{fig:pred}
\end{figure}


\section{Central limit theorem for a copula-based mixed model}\label{app:clt}

\subsection{Auxiliary results}\label{app:aux}
\begin{thm}[Lindeberg's condition]\label{thm:CLT-Lindeberg}
Suppose that $\bxi_1,\ldots, \bxi_K$ are independent, and as $K\to\infty$, $N_K\to\infty$,
\begin{equation}
    \label{eq:cov-lim}
   \frac{1}N_K \sum_{k=1}^K \var(\bxi_k)\to \Sigma,
\end{equation}
and
\begin{equation}\label{eq:Lindeberg}
\frac{1}N_K \sum_{k=1}^K E\left\{\|\bxi_{k}\|^2 \I\left(\left\|\bxi_k\right\|> \epsilon\sqrt{N_K}\right)\right\} \to 0,\quad \text{  for any $\epsilon>0$ }.
\end{equation}
Then $\frac{1}{N_K^{1/2}}\sum_{k=1}^K \bxi_k $ converges in distribution to a centered Gaussian random vector with covariance matrix $\Sigma$.
\end{thm}

\begin{cor}\label{rem:lindeberg} If $\bxi_k = \sum_{i=1}^{n_k}\bxi_{ki}$, where $\{\bxi_{ki}: \; 1\le i\le n_k\}$ are exchangeable and square integrable, and $\lambda_K = \frac{1}{N_K}\sum_{k=1}^K n_k^2\to \lambda \in [1,\infty)$, then \eqref{eq:cov-lim} holds with $\Sigma = \Sigma_{11}+(\lambda-1)\Sigma_{12}$, and either of the following assumptions implies \eqref{eq:Lindeberg}:
\begin{eqnarray}\label{eq:condLind0}
 \sum_{k=1}^K \frac{n_k^2}{N_K}E\left\{\|\xi_{k1}\|^2 \I\left(\left\|\bxi_k\right\|> \epsilon\sqrt{N_K}\right)\right\}
 &\to& 0,\\
\label{eq:condLind1}
\max_{1\le k \le K}E\left\{\|\bxi_{k1}\|^2 \I\left(\left\|\bxi_k\right\|> \epsilon\sqrt{N_K}\right)\right\} &\to& 0.
\end{eqnarray}
If, in addition, $n_k\equiv n$, then \eqref{eq:Lindeberg} holds since $\bxi_1$ is square integrable and $\lambda=n$. In particular, $\frac{1}{N_K^{1/2}}\sum_{k=1}^K \bxi_k $ converges in distribution to a centered Gaussian random vector with covariance matrix $\Sigma$.
\end{cor}

\begin{rem}\label{rem:lp}
 Statistics such as $\sum_{k=1}^K \bxi_k = \sum_{k=1}^K \sum_{i=1}^{n_k}\bxi_{ki}$ are referred to as multisample statistics in \cite{Vaillancourt:1995}. It is worth noting that \eqref{eq:Lindeberg} holds if for some $p > 2$, $\|\bxi_{ki}\|^p$ is integrable and $\frac{1}{N_K^{p/2}}\sum_{k=1}^K n_k^{p}\to 0$. This condition is only useful when the $n_k$'s are not equal.
\end{rem}




\subsection{Proof of Theorem \ref{thm:main}}\label{app:pf-main}

We recall that the log-likelihood $L_K(\btheta)$ can be written as
$\disp
L(\btheta) = \sum_{k=1}^K \log f_{\btheta,k}(\by_k)$, where
$\disp f_{\btheta,k}(\by_k) = \int_0^1 \left\{\prod_{j=1}^{n_k} f_{\btheta,kj}(y_{kj},v)\right\}dv$,
and for $v\in (0,1)$,
$f_{\btheta,kj}(y_{kj},v)$ is a density with respect to the reference measure $\nu$. As a result, because of Assumption \ref{hyp:mu}, we can differentiate under the integral sign, implying
$\dot L(\btheta_0) = \disp \left. \nabla_{\btheta}L(\btheta)\right|_{\btheta=\btheta_0} = \sum_{k=1}^K\sum_{i=1}^{n_k} \bbeta_{ki}$,
where
$$
\bbeta_{ki} = \frac{1}{f_{\btheta_0,k}(\by_k)}\int \left\{\prod_{j\neq i} f_{\btheta_0,kj}(y_{kj},v)\right\}\left. \nabla_\btheta f_{\btheta,ki}(y_{ki},v)\right|_{\btheta=\btheta_0}dv.
 $$
Now, for any $i\in\{1,\ldots, n_k\}$, $\disp \int_{\dR} \dot f_{\btheta_0, ki}(y,v)\nu(dy) = \left. \nabla_\btheta \left\{ \int_{\dR} \dot f_{\btheta, ki}(y,v)\nu(dy)\right\} \right|_{\btheta=\btheta_0} =  \left. \nabla_\btheta \;  \{1\}  \right|_{\btheta=\btheta_0}=0$. Next, for bounded measurable functions
 $\Psi_j$, $j\neq i$, we obtain
\begin{multline*}
  E\left[\bbeta_{ki}\prod_{j\neq i}\Psi_j(Y_{kj}) \right] \\
  =
  \int \int_0^1 \left[ \prod_{j\neq i} \left\{ f_{\btheta_0,kj}(y_j,v)\Psi_j (y_j)\right\}\right]
 \dot  f_{\btheta_0,ki}(y_i,v) dv \; \nu(dy_1) \cdots \nu(dy_{n_k})=0.
\end{multline*}
 As a result, $E\left(\left. \bbeta_{ki}\right|Y_{kl}, \; l\neq i\right) = 0$. In particular, $E(\bbeta_{ki})=0$. We note that this does not imply that $E\left(\bbeta_{ki}\bbeta_{kj}^\top\right) = 0$, $i\neq j$, since $\bbeta_{kj}$ is not measurable with respect to $\{Y_{kl}: l\neq i\}$! However, $\bbeta_{k,1},\ldots, \bbeta_{k,n_k}$ are exchangeable with a mean of $0$.  We further set $\Sigma_{11} = E\left(\bbeta_{k1}\bbeta_{k1}^\top \right)$,
$\Sigma_{12} = E\left(\bbeta_{k1}\bbeta_{k2}^\top \right)$.
We then set,
$$
\ddot L(\btheta_0) = \sum_{k=1}^K \frac{\ddot f_{\btheta_0,k}(\by_k)}{f_{\btheta_0,k}(\by_k)} -
\sum_{k=1}^K \frac{\dot f_{\btheta_0,k}(\by_k)\dot f_{\btheta_0,k}(\by_k)^\top }{f_{\btheta_0,k}^2(\by_k)},
$$
where
\begin{eqnarray*}
\ddot f_{\btheta_0,k}(\by_k) &=& \sum_{i=1}^{n_k}\int_0^1 \left\{\prod_{j\neq i} f_{\btheta_0,kj}(y_{kj},v)\right\}\ddot f_{\btheta_0,ki}(y_{ki},v) dv \\
&& \qquad +
\sum_{i\neq j}\int_0^1 \left\{\prod_{l\neq i,j} f_{\btheta_0,kl}(y_{kl},v)\right\}\dot f_{\btheta_0,ki}(y_{ki},v) \dot f_{\btheta_0,kj}(y_{kj},v)^\top dv\\
&=& \sum_{i=1}^{n_k}\bzeta_{ki}+\sum_{i\neq j}\bzeta_{kij}.
\end{eqnarray*}
As seen previously, we obtain that $E(\bzeta_{ki}|Y_{kl},l\neq i)=0$ and  $E(\bzeta_{kij}|Y_{kl},l\neq i)=0$
Set $A_K = \frac{1}{N_K}\sum_{k-1}^K \bbeta_k\bbeta_k^\top$. Then, as $K\to \infty$,  $E(A_K)= \Sigma_{11}+(\lambda_K-1)\Sigma_{12}\to \Sigma $. If $n_k\equiv n$, then $A_K\to \Sigma$ by the strong law of large numbers, assuming only that $\bbeta_{ki}$ is square integrable. Otherwise,  $A_K\stackrel{Pr}{\longrightarrow}\Sigma$ if $\var(A_K)=\frac{1}{N_K^2}\sum_{k=1}^K \var\left(\bbeta_k\bbeta_k^\top\right)\to 0$ which is equivalent to
\begin{multline*}
  \frac{1}{N_K^2}\sum_{k=1}^K E\left(\bbeta_k\bbeta_k^\top \bbeta_k\bbeta_k^\top \right) -
\frac{1}{N_K^2}\sum_{k=1}^K\left\{ E\left(\bbeta_k\bbeta_k^\top  \right) \right\}^2  \\
=
\frac{1}{N_K^2}\sum_{k=1}^K E\left(\bbeta_k\bbeta_k^\top \bbeta_k\bbeta_k^\top \right) - \frac{1}{N_K^2}\sum_{k=1}^{n_k}\left\{n_k\Sigma_{11}+n_k(n_k-1)\Sigma_{12}\right\}^2\to 0,
\end{multline*}
which is sufficient to demonstrate that
$\disp
\frac{1}{N_K^2}\sum_{k=1}^K E\left(\bbeta_k\bbeta_k^\top \bbeta_k\bbeta_k^\top \right) \to 0$. Since $E\left(\|\bbeta_{ki}\|^4\right)=V<\infty$, we obtain
\begin{multline*}
\left\|\frac{1}{N_K^2}\sum_{k=1}^K E\left(\bbeta_k\bbeta_k^\top \bbeta_k\bbeta_k^\top \right) \right\|=
\frac{1}{N_K^2}\left\|\sum_{k=1}^K \sum_{i=1}^{n_k}\sum_{j=1}^{n_k}\sum_{l=1}^{n_k}\sum_{m=1}^{n_k}
E\left(\bbeta_{ki}\bbeta_{kj}^\top \bbeta_{kl}\bbeta_{km}^\top\right)\right\| \\
\le
\frac{V}{N_K^2}\sum_{k=1}^{n_k} n_k^4\to 0.
\end{multline*}
This leads to Remark \ref{rem:lp} with $p=4$ that the conditions of Theorem \ref{thm:CLT-Lindeberg} are met. In addition,
$\dT_K(\btheta) = N_K^{-1/2}\{\dot L_K(\btheta)-\bmu(\btheta)\}$ converges in $C(\bar \cN)$ to $\dT$, where $\dT$ is a continuous centered Gaussian process. This further stems from the fact that the finite dimensional distributions of $\dT_{N_K}$ converge to those of $\dT$, and from moment conditions in Assumptions in \ref{hyp:mu} and \ref{hyp:l4bound}, $E\left[\left\{T_{N_k}(\btheta)-
T_{N_k}(\btheta')\right\}^2\right] \le C_1 \|\btheta-\btheta'\|^2$, for some constant $C_1$, for all $\btheta,\btheta'\in\bar\cN$. As a result, the sequence $\dT_{N_K}$ is narrow in $C(\bar \cN)$ and thus $\dT_K(\btheta) $ converges in $C(\bar \cN)$ to $\dT$. Next,
By Skorokhod's theorem, see, e.g.,  \citet{Ethier/Kurtz:1986}, we can consider a new probability space so that the stochastic processes $(\tilde \dT_{K},\tilde \dT)$ have the same law as $(\dT_{K},\dT)$ and $\disp \sup_{\btheta\in\cN}
\left|\tilde \dT_{N_K}(\btheta)-\tilde \dT(\btheta)\right|\to 0$, as $K\to\infty$. As a result, since $\bar\cN$ can be chosen compact, there is at least one converging subsequence. If $\btheta_{K_l}$ is any converging subsequence such that $\tilde L_{K_l}(\btheta_{K_l})=0$, and $\btheta_{K_l}\to \btheta^\star$, then $\bmu\left(\btheta^\star\right)=0$ , since $\tilde \dT_{K_l}(\btheta_{K_l})= -N_{K_l}^{1/2}\bmu(\btheta_{K_l})\to \tilde \dT(\btheta^\star)$. By hypothesis, one must have $\btheta^\star=\btheta_0$. It then follows that $\tilde \btheta_K \to \btheta_0$. Finally,
$0 = N_K^{-1/2} \tilde {\dot L}_K(\tilde \btheta_K) = \tilde \dT_K(\tilde \btheta_K) +  N_K^{-1/2}\bmu(\tilde \btheta_K) =  \tilde \dT_K(\btheta_0)
+  \dot \bmu(\check\btheta_K)\tilde \bTheta_k $, for some $\check \btheta_K$ converging to $\btheta_0$. Since
$\tilde \dT_K(\tilde\btheta_K) \to \tilde \dT(\btheta_0)\sim N(0,\Sigma)$ and $\dot \bmu(\tilde\btheta_K)\to \dot\bmu(\btheta_0)=-\Sigma$, it follows that $\tilde \bTheta_K\to \tilde \bTheta$, so $\bTheta_K$ converges in distribution  to $\Sigma^{-1}\dT(\btheta_0)\sim N\left(0,\Sigma^{-1}\right)$. In the case $n_k\equiv n$, note that the $\bbeta_{\btheta,k}$ are iid.
\qed

\section{Comparison of predictions}\label{app:compvaccine}

\begin{figure}
    \centering

   \includegraphics[scale=0.055]{predSouthAfrica.jpg}
   \includegraphics[scale=0.055]{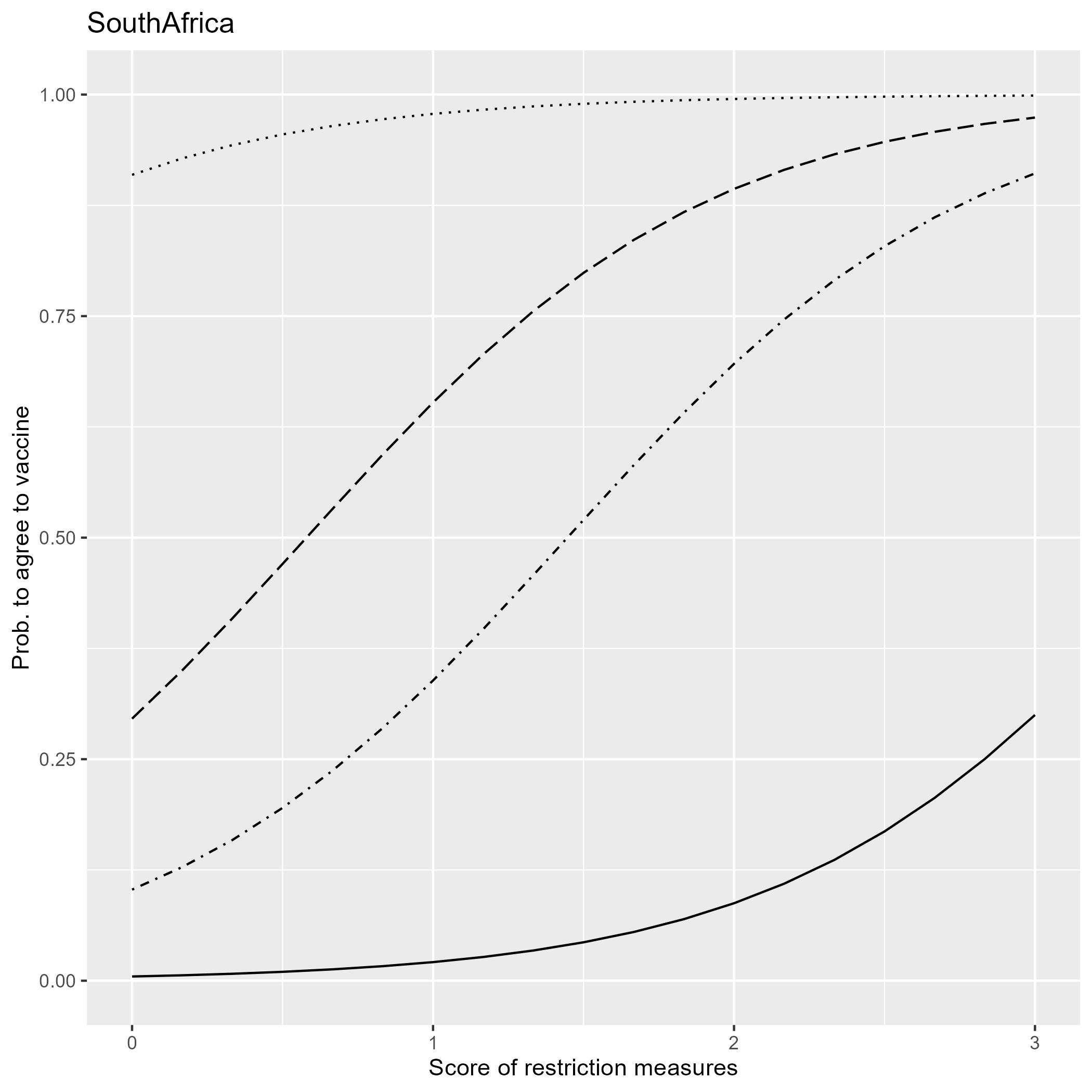}
   \includegraphics[scale=0.055]{predFrance.jpg}
   \includegraphics[scale=0.055]{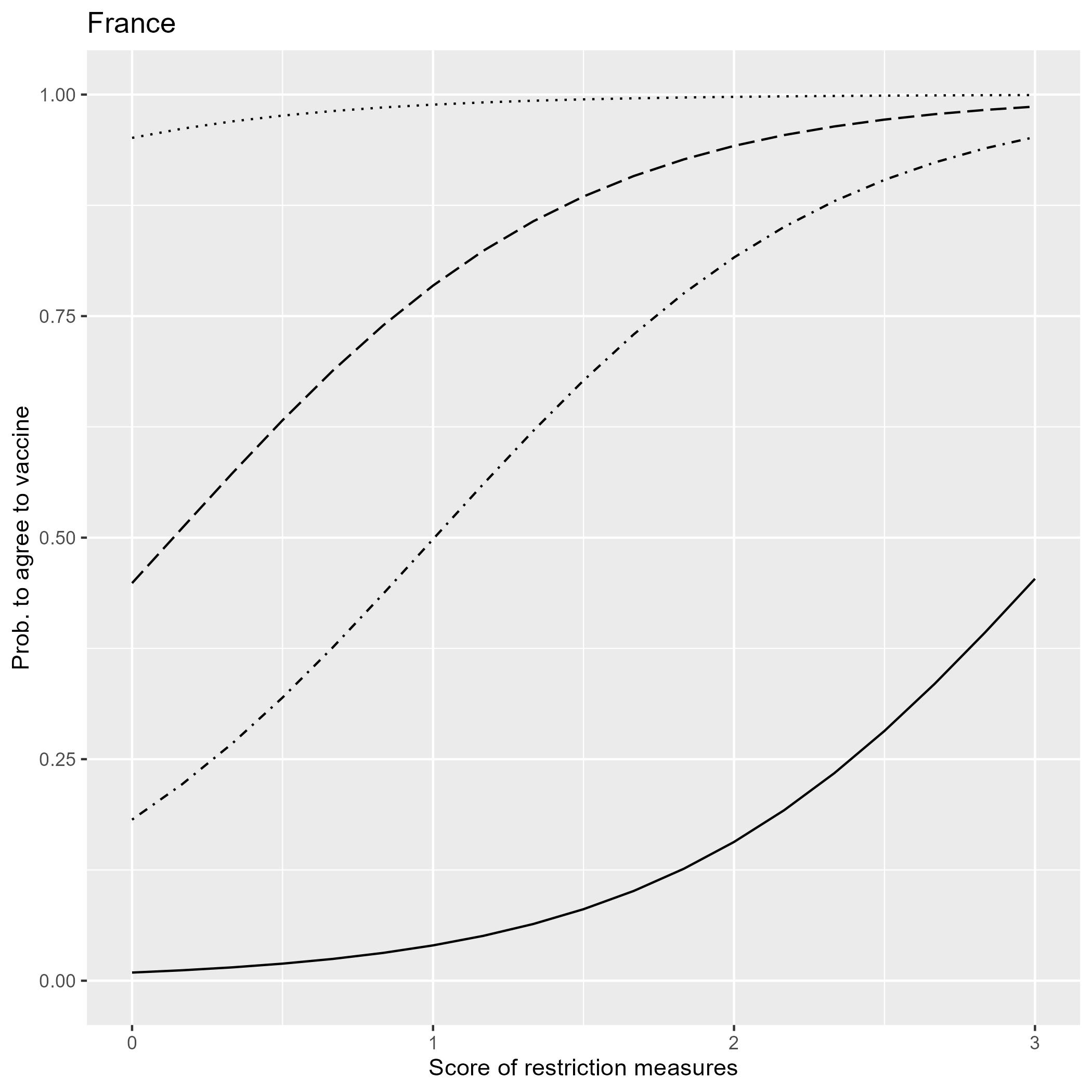}
   \includegraphics[scale=0.055]{predIndia.jpg}
   \includegraphics[scale=0.055]{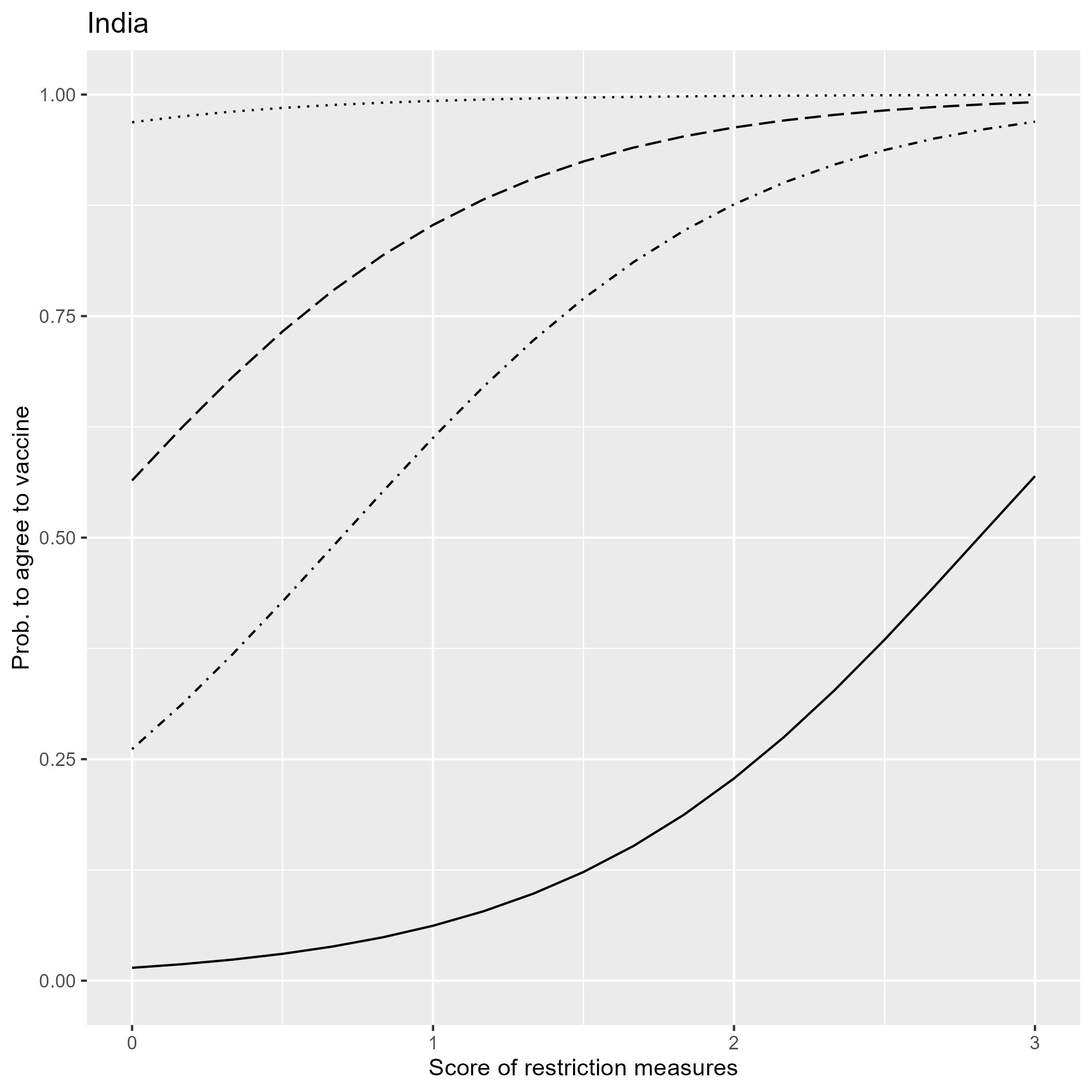}
   \includegraphics[scale=0.055]{predSpain.jpg}
   \includegraphics[scale=0.055]{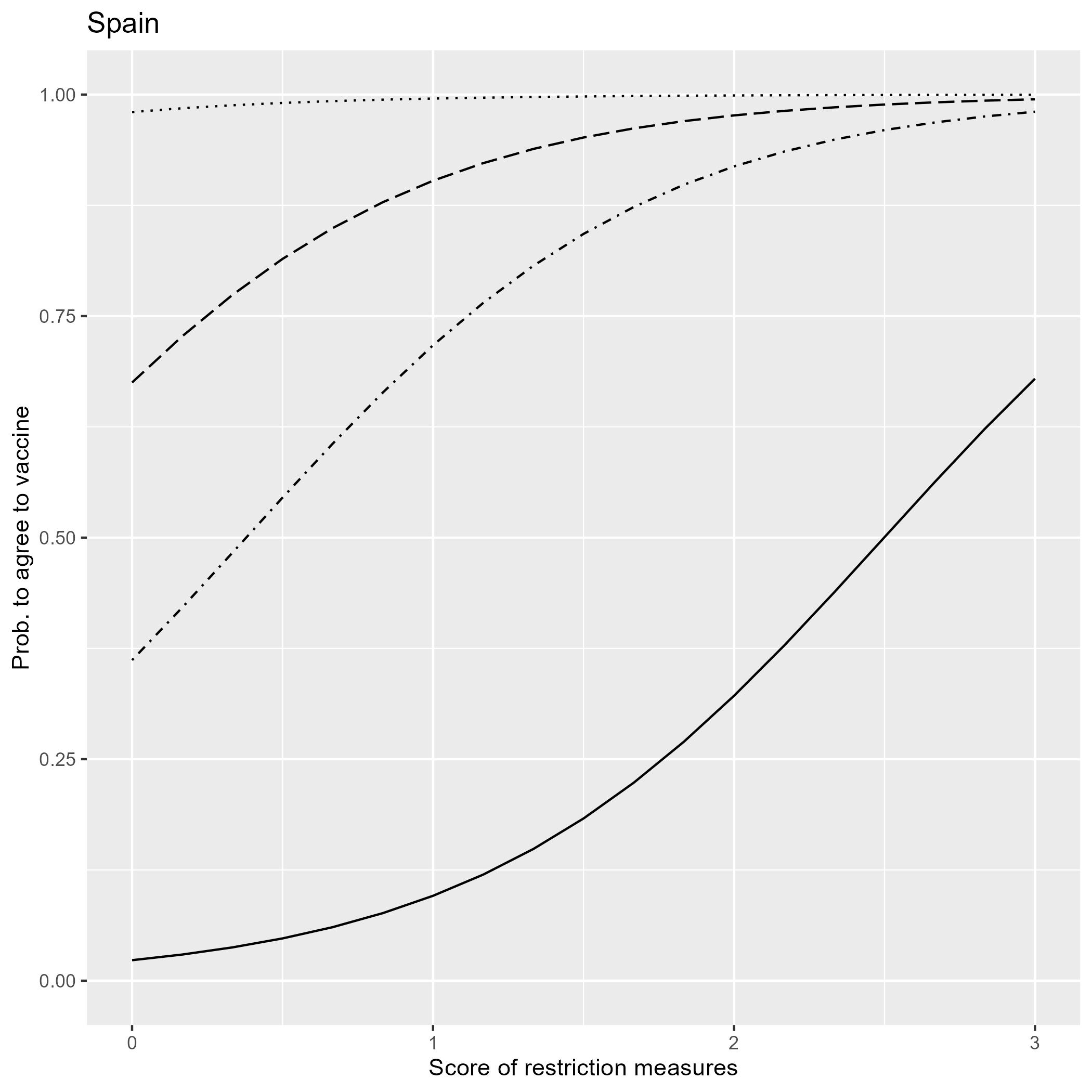}
 \caption{Estimated probabilities to agree to be vaccinated by scores of restriction measures ($X_1$) by values of trust ($X_2$) and perception of risk ($X_3$) for South Africa  ($\hat V = 0.9506$), France ($\hat V = 0.5674$), India ($\hat V = 0.2428$), and  Spain ($\hat V =0.0481$). The line types are ``solid'' for $(X_2,X_3)=(0,0)$, ``dotdash'' for $(X_2,X_3)=(3,0)$, ``longdash'' for $(X_2,X_3)=(0,3)$, and ``dotted'' for $(X_2,X_3)=(3,3)$. The graphs on the left panel are computed with the copula-based model, while the graphs on the right panel are computed with the Logistic model with random effect.}
    \label{fig:compvaccine}
\end{figure}
\newpage
\bibliographystyle{apalike}




\end{document}